\documentclass[a4paper,UKenglish,thm-restate,autoref,cleveref,numberwithinsect,pdfa]{lipics-v2021}

\pdfoutput=1
\hideLIPIcs
\nolinenumbers
\allowdisplaybreaks
\usepackage{mathtools}
\usepackage{xfrac}
\usepackage{tikz}
\usetikzlibrary{arrows, automata}

\newcommand{\prob}{P}
\newcommand{\xfrac}[2]{%
	\mbox{\raisebox{0.3ex}{\ensuremath{\displaystyle #1}\hspace{-0.2ex}}%
		/%
		\raisebox{-0.3ex}{\footnotesize{\ensuremath{#2}}}%
	}%
}
\newcommand{\Until}{\! \mbox{$\, {\sf U}\,$} \!}
\newcommand{\Next}{\bigcirc}

\title{A Spectrum of Approximate Probabilistic Bisimulations} 

\author{Timm Spork}{Technische Universität Dresden, Dresden, Germany}{timm.spork@tu-dresden.de}{https://orcid.org/0009-0008-4461-0667}{}

\author{Christel Baier}{Technische Universität Dresden, Dresden, Germany}{christel.baier@tu-dresden.de}{https://orcid.org/0000-0002-5321-9343}{}

\author{Joost-Pieter Katoen}{RWTH Aachen University, Aachen, Germany}{katoen@cs.rwth-aachen.de}{https://orcid.org/0000-0002-6143-1926}{}

\author{Jakob Piribauer}{Technische Universität Dresden, Dresden, Germany \and Universität Leipzig, Leipzig, Germany}{jakob.piribauer@tu-dresden.de}{https://orcid.org/0000-0003-4829-0476}{}

\author{Tim Quatmann}{RWTH Aachen University, Aachen, Germany}{tim.quatmann@cs.rwth-aachen.de}{https://orcid.org/0000-0002-2843-5511}{This research was funded by a KI-Starter grant from the Ministerium f\"ur Kultur und Wissenschaft NRW.}
	
\authorrunning{T. Spork, C. Baier, J.-P. Katoen, J. Piribauer and T. Quatmann} 

\Copyright{Timm Spork, Christel Baier, Joost-Pieter Katoen, Jakob Piribauer, and Tim Quatmann}



\relatedversion{This paper is the full version of a paper accepted for publication at CONCUR 2024} 

\funding{\textit{Christel Baier, Jakob Piribauer} and \textit{Timm Spork}: This work was partly funded by the DFG Grant 389792660 as part of TRR 248 (Foundations of Perspicuous Software Systems) and the Cluster of Excellence EXC 2050/1 (CeTI, project ID 390696704, as part of Germany’s Excellence Strategy).}

\acknowledgements{We thank the reviewers for their helpful feedback, comments and suggestions. In particular, we thank the reviewer who pointed out the work in \cite{FMOSCSMDP,RDPTLCSS} on approximate simulation relations which we were previously not aware of.}

\begin{document}

	
	\maketitle

	
	\begin{abstract}
		This paper studies various notions of approximate probabilistic bisimulation on labeled Markov chains (LMCs). We introduce approximate versions of weak and branching bisimulation, as well as a notion of $\varepsilon$-perturbed bisimulation that relates LMCs that can be made (exactly) probabilistically bisimilar by small perturbations of their transition probabilities. We explore how the notions interrelate and establish their connections to other well-known notions like $\varepsilon$-bisimulation. 
	\end{abstract}

	
	\section{Introduction}\label{Section: Introduction}
	Probabilistic model checking is widely used for the automatic verification of probabilistic models, like labeled Markov chains (LMC), against properties specified in (temporal) logics like $\mathrm{PCTL}^*$ \cite{PoMC}. In practice, a big obstacle is the \emph{state space explosion problem}: the number of states required to model a system can make its verification intractable \cite{BMMSUPMC,PoMC,PMCL}.  
	
	To circumvent this issue, a well-established approach is the use of abstractions. For a given LMC $\mathcal{M}$, an abstraction $\mathcal{A}$ is a model derived from $\mathcal{M}$ that is (oftentimes) smaller than $\mathcal{M}$ and preserves some properties of interest. Instead of verifying a formula on $\mathcal{M}$, one does so on $\mathcal{A}$ and afterwards transfers the result back to the original model \cite{PoMC,RPCTLMC,ABM}.
	
	A prominent type of abstraction are probabilistic bisimulation quotients. They are constructed w.r.t. probabilistic bisimulations, a class of behavioral equivalences introduced by Larsen and Skou \cite{BTPT} as an extension of Milner's bisimulation \cite{CCS} to probabilistic models. A probabilistic bisimulation is an equivalence $R$ on the state space of an LMC $\mathcal{M}$ that only relates states that behave exactly the same, i.e., that have the same local properties, and transition to $R$-equivalence classes with equal probability. The coarsest probabilistic bisimulation $\sim$, called (probabilistic) bisimilarity, is the union of all probabilistic bisimulations in $\mathcal{M}$ \cite{PoMC}. The bisimilarity relation can be computed efficiently \cite{PTATPBS,OSSLMC,SOTMCL} and preserves $\mathrm{PCTL}^*$ state formulas \cite{IUWTLSS,LRTR}. Since verifying $\mathrm{PCTL}^*$ on bisimulation quotients can significantly speed up the verification process \cite{BMMSUPMC}, their use is a vital part of probabilistic model checkers such as, e.g., \textsc{Storm} \cite{PMCS}. 
	
	Other notions of behavioral equivalence are \emph{weak} and \emph{branching} probabilistic bisimulations \cite{CC,BTABS,WBFPP,BSRMC,GBA,CABBSPLPC}, which were introduced  with the intention to abstract from sequences of \emph{internal actions} or \emph{stutter steps} a model can perform. Intuitively, these notions can abstract from the possibility of a state to, for some time, only visit equally labeled states (weak) or stay in its own equivalence class (branching) \cite{NLTAWBMC}. 
	It is well-known that weak and branching probabilistic bisimilarity, denoted $\approx^w$ and $\approx^b$, respectively, coincide for LMCs \cite{WBFPP}, and that they characterize satisfaction equivalence for a variant of $\mathrm{PCTL}^*$ \cite{WBSCPCTL}.
	
	A problem with all of the above notions lies, however, in their lack of robustness against errors in the transition probabilities. The requirement of related states to have \emph{exactly} the same transition probabilities to equivalence classes implies that even an infinitesimally small perturbation of any of these probabilities can cause two bisimilar states to become non-bisimilar, resulting in larger quotients \cite{MfLMS,CISBCMC,RPCTLMC}. 
	This disadvantage was first observed in \cite{ARfPCS}, where the use of \emph{approximate} notions of bisimulation is suggested for its mitigation. 
	
	The literature proposes various types of approximate bisimilarity, the most well-known and well-studied one being $\varepsilon$-bisimilarity ($\sim_\varepsilon$) \cite{AAPP}. Other notions include approximate probabilistic bisimilarity with precision $\varepsilon$ ($\equiv_\varepsilon$), or $\varepsilon$-APB for short \cite{RPCTLMC,AMBPBGSSMP, PMCLMPFAB},  up-to-$(n, \varepsilon)$-bisimilarity ($\sim_\varepsilon^n$) \cite{AAPP,SAAPBPCTL}, or $\varepsilon$-lumpability of a given LMC \cite{EOLFMC,BQLMC,CRQLMC}. Here, we propose definitions for approximate versions of weak ($\approx_\varepsilon^w$) and branching probabilistic bisimilarity ($\approx_\varepsilon^b$). Similar notions have, to the best of our knowledge, only been discussed sporadically in the context of noninterference under the term ``weak bisimulation with precision $\varepsilon$'' \cite{TFAANP,SAPNRP,PAAAPN,QANPS,MCPS,NAWPB}. Moreover, we introduce $\varepsilon$-perturbed bisimilarity ($\simeq_\varepsilon$) which relates two LMCs if they can be made bisimilar by small perturbations of their transition probabilities. Implicitly, this relation arises in the work \cite{ABM} on a type of abstraction called $\varepsilon$-quotients. With our definition, two LMCs are $\varepsilon$-perturbed bisimilar iff they have bisimilar $\varepsilon$-quotients. 
	
	All of the approximate notions have in common that they allow a small tolerance, say $\varepsilon > 0$, in the transition probabilities of related states, but differ in the specifics of where and how this tolerance is put to use. Broadly speaking, we can distinguish two groups of relations: while $\sim_\varepsilon$, $\equiv_\varepsilon, \sim_\varepsilon^n$ and $\approx_\varepsilon^w$ are additive in their tolerances and are closer to classic process relations, the notions underlying $\sim_\varepsilon^*$ and $\equiv_\varepsilon^*$, denoting \emph{transitive} $\varepsilon$-bisimilarity and \emph{transitive} $\varepsilon$-APB, respectively, as well as $\simeq_\varepsilon$  and $\approx_\varepsilon^b$ are better suited for the construction of abstractions since they are required to be equivalences. Collapsing the equivalence classes of such a relation into single states yields quotient models, which in some cases are such that formulas given in specific (fragments of) logics are (approximately) preserved between the original LMC and its quotient. However, as we will see later, requiring transitivity can cause some unnatural behavior like the possibility to distinguish probabilistically bisimilar LMCs and a lack of additivity. Furthermore, the induced bisimilarity relations, which are again defined as the union of all corresponding relations in the model $\mathcal{M}$ (e.g., $\approx_\varepsilon^b$ is the union of all branching $\varepsilon$-bisimulations in $\mathcal{M}$) might themselves not be of the respective type anymore (e.g., $\approx_\varepsilon^b$ is not necessarily a branching $\varepsilon$-bisimulation). This contrasts the non-transitive case, where the induced bisimilarity relations are always of the respective type. We summarize the relations we consider, together with some of their properties, in \Cref{Table}. 
	\begin{table}[t]
		\caption{Overview of the notions of approximate bisimulation we consider and some of their properties. Being suitable for ``quotienting'' is meant w.r.t. the underlying bisimulation relation.}
		\centering
		\resizebox{!}{0.09\textheight}{
			\begin{tabular}{l c c c c}
				Notion & Symbol for Union&& Additive & Quotienting \\ \hline
				$\varepsilon$-Bisimulation \cite{AAPP,RBBTEAPC} & $\sim_\varepsilon$ & \multirow{4}{*}{$\left.\begin{array}{l}
						\\
						\\
						\\
						\\
					\end{array}\right\rbrace$} & \multirow{4}{*}{\checkmark} & \multirow{4}{*}{$\times$} \\
				$\varepsilon$-APB \cite{RPCTLMC,AMBPBGSSMP,PMCLMPFAB} & $\equiv_\varepsilon$ &&  & \\
				Up-To-$(n, \varepsilon)$-Bisimulation \cite{AAPP,SAAPBPCTL} & $\sim_\varepsilon^n$ && & \\
				Weak $\varepsilon$-Bisimulation & $\approx_\varepsilon^w$  && &  \\ \hline
				Transitive $\varepsilon$-Bisimulation & $\sim_\varepsilon^*$  &\multirow{4}{*}{$\left.\begin{array}{l}
						\\
						\\
						\\
						\\
					\end{array}\right\rbrace$}& \multirow{4}{*}{$\times$} & \multirow{4}{*}{\checkmark} \\
				Transitive $\varepsilon$-APB & $\equiv_\varepsilon^*$ && &\\
				Branching $\varepsilon$-Bisimulation & $\approx_\varepsilon^b$ &&&\\
				$\varepsilon$-Perturbed Bisimulation \cite{ABM} & $\simeq_\varepsilon$ &&  &  
		\end{tabular}}
		\label{Table}
	\end{table}
	
	\smallskip
	\noindent
	\textbf{Main Contributions.} The main contributions are as follows: 
	\begin{enumerate}
		\item Starting with the classic notion of $\varepsilon$-bisimilarity, we show tightness of a bound from \cite{RDPTLCSS} on the absolute difference of unbounded reachability probabilities in $\varepsilon$-bisimilar states (\Cref{ex:unbounded_reach}). 
		\item We introduce \emph{$\varepsilon$-perturbed bisimilarity} ($\simeq_\varepsilon$), a notion that relates two LMCs if they have bisimilar $\varepsilon$-quotients á la \cite{ABM}, i.e., if they can be made probabilistically bisimilar by small perturbations of their transition probabilities. We show that $\simeq_\varepsilon$ is strictly finer than (transitive) $\varepsilon$-bisimilarity $\sim_\varepsilon^{(*)}$ (\Cref{Corollary: Perturbed Finer Than Transitive,Theorem: Epsilon-Bisimulation does not imply existence of common 1/4-quotient}) and that deciding both $\simeq_\varepsilon$ and $\sim_\varepsilon^*$ is \textsf{NP}-complete (\Cref{Theorem: NP-completeness of common epsilon-quotient}). Furthermore, we characterize $\simeq_\varepsilon$ in terms of transitive $\varepsilon$-bisimulations satisfying a \emph{centroid property} (\Cref{Theorem: Characterisation Perturbed Epsilon Bisimilar and Transitive Epsilon Bisimulation}) and discuss some anomalies of $\simeq_\varepsilon$: the relation is not always an $\varepsilon$-perturbed bisimulation itself, it is not additive in $\varepsilon$ and it can distinguish bisimilar LMCs (\Cref{Proposition: Simeq can distinguish bisimilar states}). 
		\item We define approximate versions of weak ($\approx^w_\varepsilon$) and branching probabilistic bisimilarity ($\approx_\varepsilon^b$). Our definitions can be evaluated locally and coincide with the exact notions $\approx^b$ and $\approx^w$, respectively, if $\varepsilon = 0$. We discuss how $\approx_\varepsilon^w$ and $\approx_\varepsilon^b$ are related to one another, as well as to $\varepsilon$-bisimilarity (\Cref{Proposition: Branching and Weak are incomparable,Proposition: Weak and Branching Epsilon-Bisimulation and Epsilon-Bisimulation are incomparable}). Moreover, we extend the bounds for reachability probabilities of  \Cref{thm:unbounded_reach} to states related by $\approx_\varepsilon^w$ and $\approx_\varepsilon^b$ (\Cref{Corollary: Unbounded Reach for Branching Epsilon Bisim,Corollary: Bound for weak bisim}), and prove that deciding $\approx_\varepsilon^b$ is \textsf{NP}-complete (\Cref{Theorem: Deciding Branching Bisimilarity is NP-complete}).
	\end{enumerate} 
	
	Together with various known results from the literature and some easy observations, our results complete the relation between several notions of approximate probabilistic bisimulation, as summarized in \Cref{Figure: Visualization of Results}. 
	
	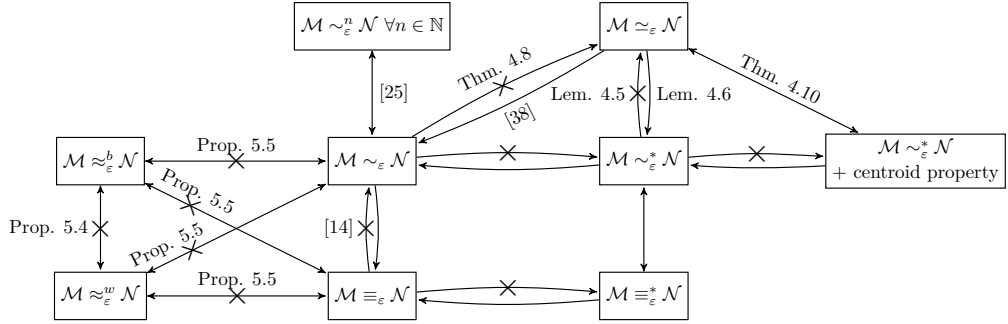
\begin{figure}[t!]
		\centering
		\resizebox{!}{0.19\textheight}{
			\begin{tikzpicture}[->,>=stealth',shorten >=1pt,auto, semithick]
				\tikzstyle{every state} = [text = black]
				\node[state, rectangle, draw] (epsilon-bisimulation) [] {$\mathcal{M} \sim_\varepsilon \mathcal{N}$}; 
				\node[state, rectangle, draw] (transitive-epsilon-bisimulation) [right of = epsilon-bisimulation, node distance = 5cm] {$\mathcal{M} \sim_\varepsilon^* \mathcal{N}$};
				\node[state, rectangle, draw, align = center] (transitive-epsilon-bisimulation-tau-star) [right of = transitive-epsilon-bisimulation, node distance = 5cm] {$\mathcal{M} \sim_\varepsilon^* \mathcal{N}$ \\$+ \text{ centroid property}$}; 
				\node[state, rectangle, draw, align=left] (epsilon-APB) [below of = epsilon-bisimulation, node distance = 2.5cm] {$\mathcal{M} \equiv_\varepsilon \mathcal{N}$};
				\node[state, rectangle, draw, align=left] (transitive-epsilon-APB) [right of = epsilon-APB, node distance = 5cm] {$\mathcal{M} \equiv_\varepsilon^* \mathcal{N}$};
				\node[state, rectangle, draw, align=left] (up-to-bisimulation) [above of = epsilon-bisimulation, node distance = 2.5cm] {$\mathcal{M} \sim_\varepsilon^n \mathcal{N} \, \, \forall n \in \mathbb{N}$};
				\node[state, rectangle, draw, align=left] (perturbed-epsilon-bisimilar) [above of = transitive-epsilon-bisimulation, node distance = 2.5cm] {$\mathcal{M} \simeq_\varepsilon \mathcal{N}$};	
				\node[state, rectangle, draw, align=left] (branching-epsilon-bisimulation) [left of = epsilon-bisimulation, node distance = 5cm] {$\mathcal{M} \approx_\varepsilon^b \mathcal{N}$};
				\node[state, rectangle, draw, align=left] (weak-epsilon-bisimulation) [left of = epsilon-APB, node distance = 5cm] {$\mathcal{M} \approx_\varepsilon^w \mathcal{N}$};
				
				\path[<->]
				(up-to-bisimulation) edge node [right] {\cite{AAPP}}(epsilon-bisimulation)
				(perturbed-epsilon-bisimilar) edge node [above, sloped, pos = 0.5] {Thm. \ref{Theorem: Characterisation Perturbed Epsilon Bisimilar and Transitive Epsilon Bisimulation}} (transitive-epsilon-bisimulation-tau-star)
				(transitive-epsilon-APB) edge (transitive-epsilon-bisimulation)
				
				(weak-epsilon-bisimulation) edge [] node [above, sloped, pos = 0.15] {Prop. \ref{Proposition: Weak and Branching Epsilon-Bisimulation and Epsilon-Bisimulation are incomparable}} node [auto = false, sloped, pos = 0.25] {\LARGE$\times$} (epsilon-bisimulation)
				
				(weak-epsilon-bisimulation) edge [] node [above] {Prop. \ref{Proposition: Weak and Branching Epsilon-Bisimulation and Epsilon-Bisimulation are incomparable}} node [auto = false, sloped, pos = 0.5] {\LARGE$\times$} (epsilon-APB)
				
				(branching-epsilon-bisimulation) edge [] node [above] {Prop. \ref{Proposition: Weak and Branching Epsilon-Bisimulation and Epsilon-Bisimulation are incomparable}} node [auto = false, sloped, pos = 0.5] {\LARGE$\times$} (epsilon-bisimulation)
				
				(branching-epsilon-bisimulation) edge [] node [above, sloped, pos = 0.25] {Prop. \ref{Proposition: Weak and Branching Epsilon-Bisimulation and Epsilon-Bisimulation are incomparable}} node [auto = false, sloped, pos = 0.25] {\LARGE$\times$} (epsilon-APB)
				
				(branching-epsilon-bisimulation) edge [] node [left, xshift = -0.1cm] {Prop. \ref{Proposition: Branching and Weak are incomparable}} node [auto = false, sloped] {\LARGE$\times$} (weak-epsilon-bisimulation)
				;
				
				\path
				(epsilon-bisimulation) edge [bend left = 5] node [sloped,pos = 0.5, above, yshift = 0.1cm] {Thm. \ref{Theorem: No common nepsilon quotient with graph isomorphism}} node [auto = false, sloped] {\LARGE$\times$} (perturbed-epsilon-bisimilar)
				
				(epsilon-bisimulation) edge [bend left = 7] node [pos = 0.5] {} (epsilon-APB)
				
				(epsilon-APB) edge [bend left = 7] node [left, xshift = -0.1cm] {\cite{RBBTEAPC}} node [auto = false, sloped, pos = 0.5] {\LARGE$\times$} (epsilon-bisimulation)
				
				(epsilon-APB) edge [bend left = 5] node [auto = false, sloped, pos = 0.5] {\LARGE$\times$} (transitive-epsilon-APB)
				
				(transitive-epsilon-APB) edge [bend left = 5] (epsilon-APB)
				
				(epsilon-bisimulation) edge [bend left = 5] node [auto = false, sloped, pos = 0.5] {\LARGE$\times$} (transitive-epsilon-bisimulation)
				
				(transitive-epsilon-bisimulation) edge [bend left = 5] (epsilon-bisimulation)
				
				(transitive-epsilon-bisimulation) edge [out = 5, in = 177, looseness = 1.15] node [auto = false, sloped, pos = 0.5] {\LARGE$\times$} (transitive-epsilon-bisimulation-tau-star)
				
				(transitive-epsilon-bisimulation-tau-star) edge [out = 183, in = 355, looseness = 0.85] (transitive-epsilon-bisimulation)
				
				(perturbed-epsilon-bisimilar) edge [bend left = 5] node [sloped, below] {\cite{ABM}} (epsilon-bisimulation)
				
				(transitive-epsilon-bisimulation) edge [bend left = 7] node [left, xshift = -0.1cm] {Lem. \ref{Corollary: Perturbed Epsilon Bisim implies Epsilon Bisim}} node [auto = false, sloped, pos = 0.5] {\LARGE$\times$} (perturbed-epsilon-bisimilar)
				
				(perturbed-epsilon-bisimilar) edge [bend left = 7] node [right] {Lem. \ref{Corollary: Perturbed Finer Than Transitive}} (transitive-epsilon-bisimulation)
				;
		\end{tikzpicture}}
		\caption{The relationship of different approximate probabilistic bisimulations.}
		\label{Figure: Visualization of Results}
	\end{figure}
	
	\smallskip 
	\noindent
	\textbf{Structure.} \Cref{Section: Preliminaries} presents preliminaries. \Cref{Section: Epsilon-Bisimulation} considers $\varepsilon$-bisimulations, $\varepsilon$-APBs and up-to-$(n, \varepsilon)$-bisimulations. \Cref{Section: Epsilon-Quotients} introduces and analyzes $\varepsilon$-perturbed bisimulations. \Cref{Section: Weak and Branching Epsilon Bisimulation} introduces weak and branching $\varepsilon$-bisimulations and establishes how they relate to $\varepsilon$-bisimulations. \Cref{Section: Conclusion} summarizes our results and points out future work.

	
	\section{Preliminaries}\label{Section: Preliminaries}
	\textbf{Distributions.} $\mathit{Distr}(S) = \{\mu \colon S \to [0,1] \mid \sum_{s \in S} \mu(s) = 1  \}$ is the set of \emph{distributions} over countable $S \neq \emptyset$. $\mu \in \mathit{Distr}(S)$ has \emph{support} $\mathit{supp}(\mu) = \{s \in S \mid \mu(s) > 0\}$, and for $A \subseteq S$ we set $\mu(A) = \sum_{s\in A}\mu(s)$. The \emph{$L_1$-distance} of $\mu, \nu \in \mathit{Distr}(S)$ is $\Vert \mu - \nu \Vert_1 = \sum_{s \in S} \vert \mu(s) - \nu(s) \vert$. 
	
	\smallskip
	\noindent
	\textbf{Labeled Markov chains.} Fix a countable set $AP$ of \emph{atomic propositions}. A \emph{labeled Markov chain} (\emph{LMC}) $\mathcal{M} = (S, \prob, s_{init}, l)$ has a countable set of \emph{states} $S \neq \emptyset$, a \emph{transition distribution function} $\prob \colon S \to \mathit{Distr}(S)$, a unique \emph{initial state} $s_{init}$, and a \emph{labeling function} $l\colon S \to 2^{AP}$. We use $\mathcal{M}$ and $\mathcal{N}$ to range over LMCs. For $s \in S$, let $L(s) = \{t \in S \mid l(s) = l(t)\}$.
	$\mathcal{M}$ is \emph{finitely branching} if $\vert \mathit{supp}(\prob(s)) \vert < \infty$ for all $s \in S$, and $\mathcal{M}$ is \emph{finite} if $\vert S \vert < \infty$. The \emph{direct sum} $\mathcal{M} \oplus \mathcal{N}$ is the LMC obtained from the disjoint union of $\mathcal{M}$ and $\mathcal{N}$. The initial state of $\mathcal{M} \oplus \mathcal{N}$ is not relevant for our purposes.
	
	For $s,t \in S$, $\prob(s)(t)$ denotes the probability to move from $s$ to $t$ in a single step. We write $Succ(s) = \mathit{supp}(P(s))$ for the set of direct successors of $s$. 
	$\pi = s_0s_1\dots \in S^{\omega}$ is an (\emph{infinite}) \emph{path} of $\mathcal{M}$ if $s_{i+1} \in Succ(s_i)$ for all $i \in \mathbb{N}$.
	$\pi[i]=s_i$ is the state at position $i$ of $\pi$, and  $trace(\pi) = l(s_0)l(s_1)\dots \in (2^{AP})^\omega$ is the \emph{trace} of $\pi$.
	The set of infinite paths is  $\mathit{Paths}(\mathcal{M})$. \emph{Finite} paths $\pi = s_0s_1\dots s_k \in S^{k+1}$ for some $k \in \mathbb{N}$ and their traces are defined analogously.
	
	Let $s \in S$. We consider the standard probability measure $\mathrm{Pr}_s^{\mathcal{M}}$ on sets of infinite paths of LMCs, defined via \emph{cylinder sets} $\mathit{Cyl}(\rho) = \{ \pi \in \mathit{Paths}(\mathcal{M}) \mid \rho \text{ is a prefix of } \pi \}$ of finite paths $\rho  \in S^*$.
	See \cite{PoMC} for details.
	For $\rho = s_0s_{1} \dots s_n$, we abbreviate $\mathrm{Pr}_s^{\mathcal{M}}(\mathit{Cyl}(\rho))$ by $\mathrm{Pr}_s^{\mathcal{M}}(\rho)$ and the measure yields $\mathrm{Pr}_s^{\mathcal{M}}(\rho) = 0$ if $s_0 \neq s$ and $\mathrm{Pr}_s^{\mathcal{M}}(\rho)= \prod_{j = 0}^{n-1} \prob(s_j)(s_{j+1})$ otherwise.
	We write $\mathrm{Pr}_{}^{\mathcal{M}}$ for $\mathrm{Pr}_{s_{init}}^{\mathcal{M}}$ and drop the superscript if $\mathcal{M}$ is clear from the context. Given a set of finite traces $T \subseteq (2^{AP})^{k+1}$ for some $k \in \mathbb{N}$, $\mathrm{Pr}_s(T)$ denotes the probability to follow, when starting in $s$, a finite path $\pi = ss_1\dots s_{k-1}$ with $trace(\pi) \in T$. 
	$\mathbb{E}^{\mathcal{M}}_s(X)$ or simply $\mathbb{E}_s(X)$ denotes the \emph{expected value} of a random variable $X$ on $\mathit{Paths}(\mathcal{M})$ w.r.t. $\mathrm{Pr}^{\mathcal{M}}_s$.
	
	\smallskip
	\noindent
	\textbf{LTL.} A popular logic for the specification of desired properties of LMCs is the \emph{linear temporal logic} (LTL) which can be used to, e.g., specify properties such as reachability, safety or liveness \cite{TLP,PoMC}. For $a \in AP$, $\mathrm{LTL}$ formulas are formed w.r.t. the grammar 
	\begin{align*}
		\varphi \Coloneqq \mathit{true} \mid a \mid \lnot \varphi \mid \varphi_1 \lor \varphi_2 \mid \Next \varphi \mid \varphi_1 \Until \varphi_2.
	\end{align*}
	Here, $\Next$ is the \emph{next} operator, so $\pi \in \mathit{Paths}(\mathcal{M})$ satisfies $\Next \varphi$ iff $\varphi$ is true in $\pi[1]$. For the \emph{until} operator $\Until$, $\pi$ satisfies $\varphi_1 \Until \varphi_2$ iff, alongside $\pi$, $\varphi_1$ holds until $\varphi_2$ is true. As syntactic sugar we define the \emph{reachability} operator $\lozenge \varphi \equiv \mathit{true} \Until \varphi$ and the \emph{always} operator $\Box \varphi \equiv \lnot \lozenge \lnot \varphi$. 
	
	For $B, C \subseteq S$ and $s \in S$, $\mathrm{Pr}_s(B \Until C)$ is the probability to reach a state in $C$ via a (finite) path from $s$ that only consists of states in $B$. Moreover, $\mathrm{Pr}_s(\lozenge^{\leq n} \varphi)$ denotes the probability to reach a state satisfying $\varphi$ from $s$ in at most $n \in \mathbb{N}$ steps. For details on LTL, see \cite{PoMC}.
	
	\smallskip
	\noindent
	\textbf{Relations.}
	Given a relation $R \subseteq S \times S$ and an $A \subseteq S$, $R(A) = \{t \in S \mid \exists \, s \in A \colon (s,t) \in R \}$ is the \emph{image of $A$ under $R$}. If $R$ is reflexive then $A \subseteq R(A)$, and $A$ is called \emph{$R$-closed} if $R(A) \subseteq A$. 
	When $R$ is an equivalence, i.e., when it is reflexive, symmetric and transitive, the \emph{equivalence class} of $s \in S$ is $[s]_R = R(\{s\}) = \{t \in S \mid (s,t) \in R\}$, and we set $\xfrac{S}{R} = \{[s]_R \mid s \in S\}$. For an equivalence $R$, the $R$-closed sets are precisely the (unions of) $R$ equivalence classes.
	
	\smallskip
	\noindent
	\textbf{Bisimulation.} An equivalence $R \subseteq S \times S$ is a (\emph{probabilistic}) \emph{bisimulation} on $\mathcal{M}$ if for all $(s,t) \in R$ and all $R$-equivalence classes $C$ it holds that $l(s) = l(t)$ and $\prob(s)(C) = \prob(t)(C)$.
	States $s, t \in S$ are (\emph{probabilistically}) \emph{bisimilar},  written $s \sim^{\mathcal{M}} t$ or simply $s \sim t$, if there is a bisimulation $R$ on $\mathcal{M}$ with $(s,t) \in R$. 
	We call two LMCs $\mathcal{M}, \mathcal{N}$ \emph{bisimilar}, written $\mathcal{M} \sim \mathcal{N}$, if  $s_{init}^\mathcal{M} \sim s_{init}^\mathcal{N}$ in $\mathcal{M} \oplus \mathcal{N}$.
	An alternative characterization of bisimulations can be found in, e.g., \cite{ALMP,AAPP,RPCTLMC,RBBTEAPC}: an equivalence $R$ is a bisimulation iff for all $(s,t) \in R$ and all $R$-closed sets $A \subseteq S$ it holds that $l(s) = l(t)$ and $\prob(s)(A) = \prob(t)(A)$. 
	The (\emph{probabilistic bisimulation}) \emph{quotient} of $\mathcal{M}$ is the LMC $\xfrac{\mathcal{M}}{\sim} = (\xfrac{S}{\sim}, \prob_\sim, [s_{init}]_\sim, l_\sim)$ with $l_\sim([s]_\sim) = l(s)$, and $\prob_\sim([s]_\sim)([t]_\sim) = \sum_{q \in [t]_\sim} \prob(s)(q)$ for all $[s]_\sim, [t]_\sim \in \xfrac{S}{\sim}$. It holds that $\mathcal{M} \sim \xfrac{\mathcal{M}}{\sim}$. An important result is that \emph{bisimilarity} $\sim$ preserves the satisfaction of $\mathrm{PCTL}^*$ state formulas \cite{IUWTLSS}.
	
	We also consider \emph{weak} and \emph{branching probabilistic bisimulations} \cite{CCS,BTABS,WBFPP,NLTAWBMC}. An equivalence $R$ is a \emph{weak probabilistic bisimulation} if, for all $(s,t) \in R$ and all $R$-equivalence classes $C \neq [s]_R = [t]_R$, it holds that $l(s) = l(t)$ and $\mathrm{Pr}_s(L(s) \Until C) = \mathrm{Pr}_t(L(t) \Until C)$. $R$ is a \emph{branching probabilistic bisimulation} if, instead of the second condition in the previous definition, $\mathrm{Pr}_s([s]_R \Until C) = \mathrm{Pr}_t([t]_R \Until C)$ holds. \emph{Weak probabilistic bisimilarity} $\approx^w$ and \emph{branching probabilistic bisimilarity} $\approx^b$  are defined like $\sim$, and lifted to LMCs in the same way.

	
	\section{\texorpdfstring{$\varepsilon$-Bisimulation, $\varepsilon$-APB and Up-To-$(n, \varepsilon)$-Bisimulation}{Epsilon-Bisimulation, Epsilon-APB and Up-To-(n, Epsilon)-Bisimulation}}\label{Section: Epsilon-Bisimulation}
	If not specified otherwise, we always assume $\varepsilon \in [0,1]$ and $\mathcal{M} = (S, \prob, s_{init}, l)$ to be finitely branching. This section summarizes various notions of approximate probabilistic bisimulation from the literature. We first provide their formal definitions and discuss how the notions interrelate. Afterwards, in \Cref{Section: Properties of Known Relations}, we present some logical preservation results.

	
	\subsection{Definitions and Interrelation}\label{Section: Definitions and Interrelation of Known Notions}
	We start with the seminal notion of \emph{$\varepsilon$-bisimulations} of Desharnais \emph{et al.} \cite{AAPP}. While originally introduced for \emph{labeled Markov processes} \cite{LCBLMP,BfLMP}, $\varepsilon$-bisimulations were later adapted to other models like LMCs \cite{RBBTEAPC,ABM} or Segala's \emph{probabilistic automata} \cite{MVRDRTS,CDBPA}. 
	
	\begin{definition}[\cite{AAPP,RBBTEAPC}]\label{Definition: Epsilon-Bisimulation}
		A reflexive\footnote{In contrast to \cite{AAPP,RBBTEAPC} we require reflexivity of $\varepsilon$-bisimulations. This is a rather natural assumption (a state should always simulate itself) that does not affect $\sim_{\varepsilon}$.} and symmetric relation $R \subseteq S \times S$ is an \emph{$\varepsilon$-bisimulation} if for all $(s,t) \in R$ and any $A \subseteq S$ it holds that 
		\begin{align*}
			(\text{\emph{i}}) \, \, l(s) = l(t) \quad \text{ and } \quad (\text{\emph{ii}}) \, \, \prob(s)(A) \leq \prob(t)(R(A)) + \varepsilon. 
		\end{align*}
		States $s,t$ are \emph{$\varepsilon$-bisimilar}, denoted $s \sim_{\varepsilon} t$, if there is an $\varepsilon$-bisimulation $R$ with $(s,t) \in R$. LMCs $\mathcal{M}, \mathcal{N}$ are $\varepsilon$-bisimilar, denoted $\mathcal{M} \sim_\varepsilon \mathcal{N}$, if $s_{init}^\mathcal{M} \sim_\varepsilon s_{init}^\mathcal{N}$ in $\mathcal{M} \oplus \mathcal{N}$.
	\end{definition}
	
	Intuitively, $s \sim_\varepsilon t$ if both states can mimic the other's transition probabilities to any $A \subseteq S$ by transitioning to the (potentially bigger) set ${\sim_\varepsilon}(A)$ with a probability that is smaller by at most $\varepsilon$ than the original one. The parameter $\varepsilon$ describes how much the behavior of related states may differ: for $\varepsilon$ close to $1$ more states can  be related, while for $\varepsilon \approx 0$ related states behave almost equivalently. In the extreme case of $\varepsilon = 0$, we have ${\sim_0} = {\sim}$ \cite{AAPP, RBBTEAPC}. 
	
	Instead of being transitive, $\varepsilon$-bisimulations are \emph{additive} in their tolerances: $s \sim_{\varepsilon_1} t$ and $t \sim_{\varepsilon_2} u$ implies $s \sim_{\varepsilon'} u$ for some $0 \leq \varepsilon' \leq \min\{1,\varepsilon_1 + \varepsilon_2\}$ \cite{AAPP}.
	As the next example suggests, transitivity is not always desirable for $\varepsilon$-bisimulations if $\varepsilon > 0$. 
	
	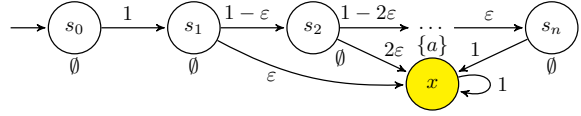
\begin{figure}[t]
		\centering
		\resizebox{!}{0.067\textheight}{
			\begin{tikzpicture}[->,>=stealth',shorten >=1pt,auto, semithick]
				\tikzstyle{every state} = [text = black]
				\node[state] (s0) {$s_0$};
				\node[state] (s1) [right of = s0, node distance = 2cm] {$s_1$};
				\node[state] (s2) [right of = s1, node distance = 2cm] {$s_2$};
				\node (temp) [right of = s2, node distance = 2cm] {\dots};
				\node[state] (sn) [right of = temp, node distance = 2cm] {$s_n$};
				\node[state] (x) [below of = temp, node distance = 0.95cm, fill = yellow] {$x$};
				
				\node (sinit) [left of = s0, node distance = 1.2cm] {};
				
				\path 
				(sinit) edge (s0)
				(s0) edge node [above, pos = 0.4] {$1$} (s1)
				(s1) edge node [above, pos = 0.4] {$1-\varepsilon$} (s2) 
				(s1) edge [bend right = 15] node [left, xshift = -0.4cm] {$\varepsilon$} (x)
				
				(s2) edge node [above, pos = 0.4] {$1-2\varepsilon$} (temp)
				(s2) edge node [right, xshift = 0.3cm, pos = 0.3] {$2\varepsilon$} (x)
				
				(temp) edge node [above] {$\varepsilon$} (sn)
				(sn) edge node [left, xshift = -0.3cm, pos = 0.3] {$1$} (x)
				
				(x) edge [loop right] node {$1$} (x)
				;
				
				\node[below of = s0, node distance = 0.65cm] {$\emptyset$};
				\node[below of = s1, node distance = 0.65cm] {$\emptyset$};
				\node[below right of = s2, node distance = 0.65cm, yshift = -0.05cm] {$\emptyset$};
				\node[below of = sn, node distance = 0.65cm] {$\emptyset$};
				\node[above of = x, node distance = 0.65cm] {$\{a\}$};
		\end{tikzpicture}}
		\caption{The LMC used in \Cref{Example: Epsilon-bisimulation not transitive}.}
		\label{Figure: Example epsilon-bisimulation not transitive}
	\end{figure}

	\begin{example}\label{Example: Epsilon-bisimulation not transitive}
		Let $\varepsilon = \frac{1}{n}$ for $n \geq 1$ and consider the LMC of \Cref{Figure: Example epsilon-bisimulation not transitive}. There, the reflexive and symmetric closure of $R = \{(s_i, s_{i+1}) \mid 0 \leq i \leq n-1\}$ is an $\varepsilon$-bisimulation. 
		Hence, $s_0$ and $s_n$ are related by a chain of $\varepsilon$-bisimilar states, even though they behave completely different: $s_0$ transitions to the $\{a\}$-labeled state $x$ with probability $0$, $s_n$ does so with probability $1$. 
	\end{example}
	
	There are different ways to characterize $\varepsilon$-bisimulations. For example, Desharnais \textit{et al.} \cite{AAPP} provide a characterization in terms of the values of maximum flows in specific flow networks. Their result is well-suited for algorithmic purposes, but is restricted to finite models. Another way to express condition (ii) of \Cref{Definition: Epsilon-Bisimulation} is via the existence of \emph{weight functions} $\Delta \colon S \to Distr(S)$ that describe how to split the successor probabilities of related states. This approach is used in, e.g., \cite{CDBPA,ABM,ABMfull,FMOSCSMDP}. The following lemma describes one possible way to define such weight functions in order to characterize $\varepsilon$-bisimulations. 
	
	\begin{restatable}{lemma}{LemmaEpsilonBisim}
		\label{lem:epsilon-bisimulation-distribution}
		A reflexive and symmetric relation $R \subseteq S \times S$ that only relates states with the same label is an $\varepsilon$-bisimulation iff for all $(s,t) \in R$ there is a map $\Delta\colon \mathit{Succ}(s) \to \mathit{Distr}(\mathit{Succ}(t))$ such that 
		\begin{enumerate}
			\item
			for all $t^\prime \in\mathit{Succ}(t) $ we have $\prob(t)(t^\prime)= \sum_{s^\prime\in \mathit{Succ}(s) } \prob(s)(s^\prime)\cdot\Delta(s^\prime)(t^\prime) $, and
			\item 
			$
			\sum_{s^\prime\in \mathit{Succ}(s)} \prob(s)(s^\prime)\cdot \Delta(s^\prime)(R(s^\prime) \cap \mathit{Succ}(t)) \geq 1 -\varepsilon
			$.
		\end{enumerate}
		
	\end{restatable}
	
	Intuitively, \Cref{lem:epsilon-bisimulation-distribution} tells us that, if $s \sim_\varepsilon t$, the successors $s^\prime$ of $s$ can be mapped to distributions $\Delta(s^\prime)$, i.e., convex combinations, of successors of $t$. More precisely, it shows that (i) if we move from $s$ to a successor $s^\prime$ with probability $\prob(s)(s^\prime)$ and, afterwards, from $s^\prime$ to a successor $t^\prime$ of $t$ with probability $\Delta(s^\prime)(t^\prime)$, then we reach $t^\prime$ with probability $\prob(t)(t^\prime)$, and that (ii) the overall probability that the states $s^\prime$ and $t^\prime$ are $\varepsilon$-bisimilar is at least $1-\varepsilon$.
	
	\smallskip
	
	A second notion of approximate probabilistic bisimulation are $\varepsilon$-APBs, which stands short for \emph{approximate probabilistic bisimulations with precision $\varepsilon$} \cite{RPCTLMC,AMBPBGSSMP, PMCLMPFAB}. In contrast to $\varepsilon$-bisimulations, where the differences in transition probabilities of related states are bounded w.r.t. all subsets $A \subseteq S$, an $\varepsilon$-APB $R$ only requires a difference of at most $\varepsilon$ for the probabilities of related states to transition to $R$-closed subsets of $S$.
	\begin{definition}[\cite{RPCTLMC}]\label{Definition: Epsilon-APB}
		A reflexive and symmetric relation $R \subseteq S \times S$ is an \emph{$\varepsilon$-APB} if for all $(s,t) \in R$ and any $R$-closed set $A \subseteq S$ it holds that 
		\begin{align*} 
			(\text{\emph{i}}) \, \, l(s) = l(t) \quad \text{ and } \quad (\text{\emph{ii}}) \,  \left \vert \prob(s)(A) - \prob(t)(A) \right \vert \leq \varepsilon. 
		\end{align*}
		We write $s \equiv_\varepsilon t$ if $s$ and $t$ are related by any $\varepsilon$-APB, and $\mathcal{M} \equiv_\varepsilon \mathcal{N}$ if $s_{init}^\mathcal{M} \equiv_\varepsilon s_{init}^\mathcal{N}$ in $\mathcal{M} \oplus \mathcal{N}$.
	\end{definition}
	
	Like $\sim_\varepsilon$, $\varepsilon$-APBs are additive in their tolerances, and we have ${\equiv_0} = {\sim} = {\sim_0}$ \cite{AAPP,RPCTLMC}. 
	
	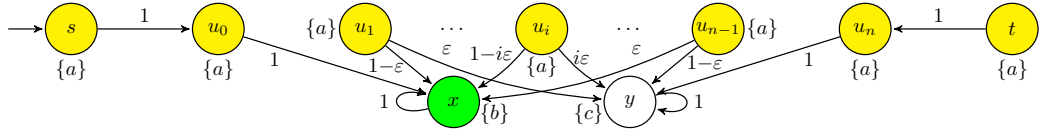
\begin{figure}[t]
		\centering
		\resizebox{!}{0.077\textheight}{
			\begin{tikzpicture}[->,>=stealth',shorten >=1pt,auto, semithick]
				\tikzstyle{every state} = [text = black]
				\node[state] (s) [fill = yellow] {$s$}; 
				\node[state] (u0) [fill = yellow, right of = s, node distance = 2.5cm] {$u_0$};
				\node[state] (u1) [fill = yellow, right of = u0, node distance = 2.5cm] {$u_1$};
				\node[] (utemp) [right of = u1, node distance = 1.5cm] {\dots};
				\node[state] (ui) [right of = utemp, node distance = 1.5cm, fill = yellow] {$u_i$};
				\node[] (utemp2) [right of = ui, node distance = 1.5cm] {\dots};
				\node[state] (uN1) [right of = utemp2, node distance = 1.5cm, fill = yellow, inner sep = 0cm] {$u_{n-1}$};
				\node[state] (uN) [right of = uN1, node distance = 2.5cm, fill = yellow] {$u_n$};
				\node[state] (t) [right of = uN, node distance = 2.5cm, fill = yellow] {$t$};
				\node[state] (x) [below of = utemp, node distance = 1.25cm, fill = green] {$x$};
				\node[state] (y) [below of = utemp2, node distance = 1.25cm] {$y$};
				
				\node[] (sinit) [left of = s, node distance = 1.2cm] {};
				
				\path
				(sinit) edge (s)
				(s) edge node [above] {$1$} (u0)
				(u0) edge node [left, pos = 0.3, xshift = -0.2cm, yshift = -0.1cm] {$1$} (x)
				(u1) edge node [left, pos = 0.45, yshift = -0.05cm, xshift = 0.08cm] {$1{-}\varepsilon$} (x)
				(u1) edge [bend right = 10] node [right, pos = 0.15, xshift = 0.2cm, yshift = 0.1cm] {$\varepsilon$} (y)
				(ui) edge [bend left = 10] node [left, pos = 0.1] {$1 {-} i\varepsilon$} (x)
				(ui) edge [bend right = 10] node [right, pos = 0.15] {$i \varepsilon$} (y)
				(uN1) edge [bend left = 10] node [left, pos = 0.15, xshift = -0.2cm, yshift = 0.1cm] {$\varepsilon$} (x)
				(uN1) edge node [right, pos = 0.45, xshift = 0.05cm] {$1{-}\varepsilon$} (y)
				(uN) edge node [right, pos = 0.3, xshift = 0.2cm, yshift = -0.1cm] {$1$} (y)
				(x) edge [loop left] node [] {$1$} (x)
				(y) edge [loop right] node [] {$1$} (y)
				(t) edge node [above] {$1$} (uN)
				
				;
				
				\node [below of = s, node distance = 0.7cm] {$\{a\}$};
				\node [below of = t, node distance = 0.7cm] {$\{a\}$};
				\node [below of = u0, node distance = 0.7cm] {$\{a\}$};
				\node [left of = u1, node distance = 0.75cm] {$\{a\}$};
				\node [below of = ui, node distance = 0.65cm] {$\{a\}$};
				\node [right of = uN1, node distance = 0.75cm] {$\{a\}$};
				\node [below of = uN, node distance = 0.7cm] {$\{a\}$};
				\node [left of = y, node distance = 0.7cm, yshift = -0.2cm] {$\{c\}$};
				\node [right of = x, node distance = 0.7cm, yshift = -0.2cm] {$\{b\}$};
		\end{tikzpicture}}
		\caption{The LMC used in \Cref{Example: Epsilon-APBs}, adapted from \cite{RBBTEAPC}.}
		\label{Figure: LMC not all epsilon-APB are epsilon-bisim}
	\end{figure}
	
	\begin{example}[\cite{RBBTEAPC}]\label{Example: Epsilon-APBs}
		Let $\varepsilon \in (0, 1]$ and $n = \left \lceil \frac{1}{\varepsilon} \right \rceil \in \mathbb{N}$. Consider $\mathcal{M}$ as in \Cref{Figure: LMC not all epsilon-APB are epsilon-bisim}, and let $R$ be the reflexive and symmetric closure of $\{(s,t),(x,x), (y,y)\} \cup \{(u_i, u_{i+1}) \mid 0 \leq i \leq n-1\}$.
		The $R$-closed sets in $\mathcal{M}$ are $\{s,t\}, \{x\}, \{y\}, \{u_i \mid 0 \leq i \leq n\}$ and their unions. For all $(p,q) \in R$ and $R$-closed sets $A$ it holds that $\vert \prob(p)(A) - \prob(q)(A) \vert \leq \varepsilon$, 
		so $R$ is an $\varepsilon$-APB. 
	\end{example}
	
	\Cref{Example: Epsilon-APBs} illustrates that the use of $\varepsilon$-APBs as a notion that relates states with almost equivalent behavior is questionable: even though states $s$ and $t$ in \Cref{Figure: LMC not all epsilon-APB are epsilon-bisim} are related by $\equiv_\varepsilon$, they behave completely different. This is caused by the set $\{u_0, \dots, u_n\}$ of (unreachable) states being $R$-closed, which in turn allows to relate $s$ and $t$ by the relation $R$ from the example. Such an anomaly cannot occur for $\sim_\varepsilon$, and we in fact have $s \nsim_\varepsilon t$ in \Cref{Figure: LMC not all epsilon-APB are epsilon-bisim} for every $\varepsilon \in (0,1)$. In particular, this shows that $\sim_\varepsilon$ can be strictly finer than $\equiv_\varepsilon$ for $\varepsilon \in (0,1)$.
	
	Lastly, we introduce \emph{up-to-$(n, \varepsilon)$-bisimulations} \cite{AAPP,SAAPBPCTL}, which are relations that require the behaviors of related states to differ by at most $\varepsilon$ for at least $n$ steps.
	\begin{definition}[\cite{AAPP,SAAPBPCTL}]\label{Definition: Up-To-(nepsilon)-Bisimulation}
		The \emph{up-to-$(n, \varepsilon)$-bisimulation} ${\sim_{\varepsilon}^n} \subseteq S \times S$ is inductively defined on $n$ via $s \sim_{\varepsilon}^0 t$ for all $s, t \in S$ and, for $n \geq 0$, $s \sim_{\varepsilon}^{n+1} t$ iff for all $A \subseteq S$
		\begin{align*}
			\text{\emph{(i)}} \, \, l(s) = l(t), \hspace{0.2cm} \text{\emph{(ii)}} \, \, \prob(s)(A) \leq \prob(t)({\sim_{\varepsilon}^n}(A)) + \varepsilon \hspace{0.2cm} \text{ and } \hspace{0.2cm} \text{\emph{(iii)}}  \, \, \prob(t)(A) \leq \prob(s)({\sim_{\varepsilon}^n}(A)) + \varepsilon.
		\end{align*} 
	\end{definition}
	
	States $s, t$ are \emph{$(n, \varepsilon)$-bisimilar} if $s \sim_{\varepsilon}^{n} t$, and the notion is lifted to LMCs as usual. Similar to $\sim_\varepsilon$ and $\equiv_\varepsilon$, $\sim_{\varepsilon}^n$ is reflexive and symmetric, but not transitive. Instead, it is additive in the tolerances and monotonic in $n$ and $\varepsilon$, i.e., for $n \geq n'$ and $\varepsilon \leq \varepsilon'$, $s \sim_{\varepsilon}^n t$ implies $s \sim_{\varepsilon'}^{n'} t$ \cite{SAAPBPCTL}.
	
	It is clear that $s \sim_\varepsilon^n t$ for a fixed $n$ does not necessarily imply $s \sim_\varepsilon t$ or $s \equiv_\varepsilon t$, as $(n, \varepsilon)$-bisimilarity only restricts the behavior of related states for $n$ steps. However, considering the limit $n \to \infty$ makes $\sim_\varepsilon$ and $\sim_\varepsilon^n$ coincide, i.e., $s \sim_{\varepsilon} t $ iff $s \sim_{\varepsilon}^n t$ for all $n \in \mathbb{N}$ \cite{AAPP}.
	
	We now make precise the relationship between $\varepsilon$-APBs and up-to-$(n, \varepsilon)$-bisimulations.
	\begin{restatable}{proposition}{TheoremRelationshipEpsilonAPBAndUpToBisimilarity}\label{Theorem: Relationship epsilon-APB and up-to-bisimilarity with precise n}
		If $\varepsilon \in (0,1)$, $s \equiv_\varepsilon t$ implies $s \sim_{\varepsilon}^n t$ if $n \leq 2$, but not necessarily if $n \geq 3$.
	\end{restatable}

	
	\subsection{Preservation of Logical Properties}\label{Section: Properties of Known Relations}
	A key application of exact probabilistic bisimilarity $\sim$ is the use of quotients $\mathcal{Q} = \xfrac{\mathcal{M}}{\sim}$ to speed up $\mathrm{PCTL}^*$ model checking \cite{BMMSUPMC,PMCL}. As abstractions built by grouping states related by approximate probabilistic bisimulations can be smaller than $\mathcal{Q}$ \cite{RPCTLMC}, these notions might prove useful to combat the \emph{state space explosion problem} of model checking \cite{BMMSUPMC,PoMC,PMCL}. It is hence of interest to see which logical properties these relations preserve.
	
	We start by considering $\sim_\varepsilon$. As shown by Bian and Abate \cite{RBBTEAPC}, $\varepsilon$-bisimilarity induces bounds on the absolute difference of satisfaction probabilities of \emph{finite horizon} properties, i.e., of properties that only depend on traces of finite length, in related states. 
	\begin{theorem}[\cite{RBBTEAPC}]\label{Theorem: Main Theorem RBBTEAPC}
		Let $s \sim_{\varepsilon} t$, $k \in \mathbb{N}$ and $T \subseteq (2^{AP})^{k+1}$ a set of traces of length $k + 1$. Then $\vert \mathrm{Pr}_s(T) - \mathrm{Pr}_t(T) \vert \leq 1 - (1 - \varepsilon)^k$.
	\end{theorem}
	
	Since any finite horizon LTL formula coincides with a set of traces of finite length, \Cref{Theorem: Main Theorem RBBTEAPC} in particular bounds the satisfaction probabilities of such formulas in $\varepsilon$-bisimilar states. Furthermore, as argued in \cite{RBBTEAPC} and the next example, this bound is tight.

	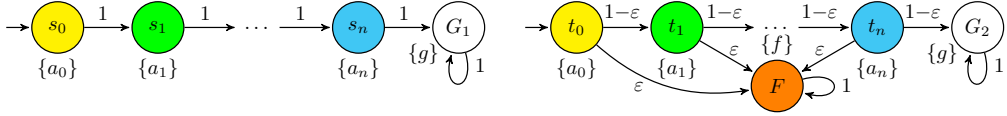
\begin{figure}[t]
		\centering
		
		\resizebox{!}{0.068\textheight}{
			\begin{tikzpicture}[->,>=stealth',shorten >=1pt,auto, semithick]
				\tikzstyle{every state} = [text = black]
				
				\node[state] (s0) [fill = yellow] {$s_0$}; 
				\node[state] (s1) [right of = s0, node distance = 1.7cm, fill = green] {$s_1$};
				\node[] (temp) [right of = s1, node distance = 1.7cm] {\dots};
				\node[state] (sn) [fill = cyan!60, right of = temp, node distance = 1.7cm] {$s_n$};
				\node[state] (g1) [right of = sn, node distance = 1.7cm] {$G_1$};
				
				\node[state] (t0) [fill = yellow, right of = g1, node distance = 2cm] {$t_0$}; 
				\node[state] (t1) [right of = t0, node distance = 1.7cm, fill = green] {$t_1$};
				\node[] (temp2) [right of = t1, node distance = 1.7cm] {\dots};
				\node[state] (tn) [fill = cyan!60, right of = temp2, node distance = 1.7cm] {$t_n$};
				\node[state] (trap) [fill = orange, below of = temp2, node distance = 1cm] {$F$};
				\node[state] (g2) [right of = tn, node distance = 1.7cm] {$G_2$};

				\node(sinit) [left of = s0, node distance = 1cm] {};
				\node(tinit) [left of = t0, node distance = 1cm] {};
				
				\path 
				(sinit) edge (s0)
				(s0) edge node [above, pos = 0.4] {$1$} (s1)
				(s1) edge node [above, pos = 0.4] {$1$} (temp)
				(temp) edge node [above, pos = 0.4] {$1$} (sn)
				(sn) edge node [above, pos = 0.4] {$1$} (g1)
				(g1) edge [loop below] node [right, pos = 0.12] {$1$} (g1)
				
				(tinit) edge (t0)
				(t0) edge node [above, pos = 0.4] {$1{-}\varepsilon$} (t1)
				(t0) edge [bend right = 25] node [below, pos = 0.3] {$\varepsilon$} (trap)
				(t1) edge node [above, pos = 0.4] {$1{-}\varepsilon$} (temp2)
				(t1) edge [] node [right, pos = 0.3, xshift = 0.1cm] {$\varepsilon$} (trap)
				(temp2) edge node [above, pos = 0.4] {$1{-}\varepsilon$} (tn)
				(tn) edge node [above, pos = 0.4] {$1{-}\varepsilon$} (g2)
				(tn) edge [] node [left, pos = 0.3, xshift = -0.1cm] {$\varepsilon$} (trap)
				(g2) edge [loop below] node [right, pos = 0.12] {$1$} (g2)
				(trap) edge [loop right] node [right] {$1$} (trap)
				;

				\node [below of = s0, node distance = 0.7cm] {$\{a_0\}$}; 
				\node [below of = s1, node distance = 0.7cm] {$\{a_1\}$}; 
				\node [below of = sn, node distance = 0.7cm] {$\{a_n\}$}; 
				\node [below of = g1, node distance = 0.45cm, xshift=-0.6cm] {$\{g\}$}; 
				\node [below of = t0, node distance = 0.7cm] {$\{a_0\}$}; 
				\node [below of = t1, node distance = 0.7cm] {$\{a_1\}$};  
				\node [below of = tn, node distance = 0.7cm] {$\{a_n\}$}; 
				\node [below of = g2, node distance = 0.45cm, xshift=-0.6cm] {$\{g\}$};  
				\node [above of = trap, node distance = 0.7cm] {$\{f\}$}; 
		\end{tikzpicture}}
		\caption{An LMC in which $s_0 \sim_\varepsilon t_0$ and $\vert \mathrm{Pr}_{s_0}(\lozenge^{\leq n+1} g) - \mathrm{Pr}_{t_0}(\lozenge^{\leq n+1} g) \vert =  1 - (1-\varepsilon)^{n+1}$.}
		\label{Figure: LMCs for maximal difference in reachability probabilities in Branching or Weak Epsilon Bisimilar States}
	\end{figure}
	
	\begin{example}\label{Example: Maximal difference in Reachability Probabilities in Branching or Weak Epsilon Bisimilar States}
		Consider \Cref{Figure: LMCs for maximal difference in reachability probabilities in Branching or Weak Epsilon Bisimilar States}. For $i \in \{0,\dots,n\}$, let $l(s_i) = l(t_i) = a_i$ for pairwise distinct $a_i$, $l(G_1) = g = l(G_2)$ and $l(F) = f$ for some $f \neq g$.
		Then $s_0 \sim_\varepsilon t_0$ and the upper bound of \Cref{Theorem: Main Theorem RBBTEAPC} is met exactly: $\vert \mathrm{Pr}_{s_0}(\lozenge^{\leq n +1} g) - \mathrm{Pr}_{t_0}(\lozenge^{\leq n +1} g) \vert= 1 - (1-\varepsilon)^{n+1}$.
	\end{example}
	
	A disadvantage of the bound provided in \Cref{Theorem: Main Theorem RBBTEAPC} is, however, that it rapidly converges to $1$ for increasing $k$ and is thus not suitable when reasoning about long (or infinite) time horizons. In fact, it is the case that---without further assumptions---even simple unbounded reachability probabilities in $\varepsilon$-bisimilar states can strongly deviate. 
	
	\begin{figure}[t]
		\centering
		\resizebox{!}{0.051\textheight}{
			\begin{tikzpicture}[->,>=stealth',shorten >=1pt,auto, semithick]
				\tikzstyle{every state} = [text = black]
				\node[state, fill=yellow] (s0) {$s_0$};
				\node[state] (s2) [right of = s0, node distance = 2cm, fill = yellow] {$s_2$};
				\node[state] (s1) [left of = s0, node distance = 2cm, fill = yellow] {$s_1$};
				\node[state] (s3) [right of = s2, node distance = 2cm, fill = green] {$s_3$};
				
				\node (sinit) [above left of = s0, node distance = 1.2cm] {};
				
				\path 
				(sinit) edge (s0)
				(s0) edge node [above] {$\frac{1}{2}$} (s1)
				(s0) edge node [above] {$\frac{1}{2}$} (s2)
				(s2) edge node [above] {$\varepsilon$} (s3)
				(s2) edge [in=20,out=50,loop] node[above, yshift = -0.08cm] {$1{-}\varepsilon$} (s2)
				(s1) edge [loop left] node {$1$} (s1)
				(s3) edge [loop right] node {$1$} (s3);
				
				\node[above of = s0, node distance = 0.7cm] {$\{a\}$};
				\node[above of = s1, node distance = 0.7cm] {$\{a\}$};
				\node[above of = s2, node distance = 0.7cm] {$\{a\}$};
				\node[above of = s3, node distance = 0.7cm] {$\{g\}$};
		\end{tikzpicture}}
		
		\caption{The LMC used in \Cref{Example: cex thm:unbounded_reach}.}
		\label{Figure: Example cex thm:unbounded_reach}
	\end{figure}
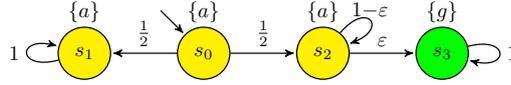
	
	\begin{example}\label{Example: cex thm:unbounded_reach}
		Let $\varepsilon \geq 0$. The states $s_0$, $s_1$, and $s_2$ in \Cref{Figure: Example cex thm:unbounded_reach} are pairwise $\varepsilon$-bisimilar.
		However, if $\varepsilon > 0$, we have $\mathrm{Pr}_{s_0}(\lozenge g) = \frac{1}{2}$, $\mathrm{Pr}_{s_1}(\lozenge g) = 0$, and $\mathrm{Pr}_{s_2}(\lozenge g) = 1$.
	\end{example}
	
	The difference in reachability probabilities observed in the last example is caused by $\sim_\varepsilon$ relating states that are able to reach a goal state $g$ with positive probability to those that can not reach $g$ at all. One way to avoid this issue is to require that states from which $g$ is not reachable are labeled with a distinct label $f$. The existence of such a label $f$ is a rather natural assumption, as a typical preprocessing step when computing reachability probabilities is to identify the states from which no goal state is reachable, i.e., to identify the states we assume to be labeled with $f$ \cite{PoMC}.
	A result in the spirit of \Cref{Theorem: Main Theorem RBBTEAPC} that deals with \emph{unbounded} reachability properties can then be obtained as follows. 
	
	\begin{restatable}[\cite{FMOSCSMDP,RDPTLCSS}]{theorem}{TheoremUnboundedReach}
		\label{thm:unbounded_reach}
		Let some states in $\mathcal{M}$ be labeled with $g$, and let exactly the states that cannot reach a $g$-labeled state be labeled with $f$.
		Further, let $s \sim_{\varepsilon} t$, and let $N$ be the random variable that counts the number of steps until reaching a $g$- or $f$-labeled state.
		Then,
		\begin{align*}
			\vert \mathrm{Pr}_s(\lozenge g) - \mathrm{Pr}_t(\lozenge g) \vert \leq \varepsilon \cdot \mathbb{E}_s(N).
		\end{align*}
	\end{restatable}
	
	\begin{remark}
		A result similar to \Cref{thm:unbounded_reach} is derived by Haesaert \textit{et al.} in \cite{FMOSCSMDP,RDPTLCSS} in the context of policy synthesis in control theory. In fact, their result is more general, as it considers all properties that can be described as the language of a deterministic finite automaton. These properties include, among others, the \emph{syntactically co-safe} LTL formulas \cite{MCSP}, which form a fragment of LTL built according to the grammar 
		\begin{align*}
			\varphi \Coloneqq \mathit{true} \mid a \mid \lnot a \mid \varphi_1 \lor \varphi_2 \mid \varphi_1 \land \varphi_2 \mid \Next \varphi \mid \varphi_1 \Until \varphi_2,
		\end{align*}
		where $a \in AP$. As unbounded reachability $\lozenge g$ is a syntactically co-safe LTL formula, the results of \cite{FMOSCSMDP,RDPTLCSS} extend the bound in \Cref{thm:unbounded_reach} to a broader class of properties. 
	\end{remark}
	
	Next, we show that the bound described in \Cref{thm:unbounded_reach} is actually tight.
	\begin{figure}[t]
		\centering
		\resizebox{!}{0.052\textheight}{
			\begin{tikzpicture}[->,>=stealth',shorten >=1pt,auto, semithick]
				\tikzstyle{every state} = [text = black]
				
				\node[state] (s0) [fill = yellow] {$s$};
				\node[state] (s1) [fill = red!65, right of = s0, node distance = 2.5cm] {$r$};
				\node[state] (s2) [fill = green, left of = s0, node distance = 2.5cm] {$q$};
				\node (sinit) [above left of = s0, node distance = 1.2cm] {};

				\node[state] (s0p) [fill = yellow, right of = s2, node distance = 11cm] {$t$};
				\node[state] (s1p) [fill = red!65, right of = s0p, node distance = 2.5cm] {$r'$};
				\node[state] (s2p) [fill = green, left of = s0p, node distance = 2.5cm] {$q'$};
				\node (sinit2) [above left of = s0p, node distance = 1.2cm] {};

				\path 
				(sinit) edge (s0)
				(s0) edge node [above] {$\frac{p}{2}$} (s1)
				(s0) edge [in=20,out=50,loop] node [above, yshift = -0.1cm] {$1{-}p$} (s0)
				(s0) edge node [above] {$\frac{p}{2}$} (s2)
				(s1) edge [loop right] node [right] {$1$} (s1) 
				(s2) edge [loop left] node [left] {$1$} (s2)
				
				(sinit2) edge (s0p)
				(s0p) edge node [above, pos = 0.65] {$\frac{p}{2}{+}\varepsilon$} (s1p)
				(s0p) edge [in=20,out=50,loop] node [above, yshift = -0.1cm] {$1 {-} p$} (s0p)
				(s0p) edge node [above] {$\frac{p}{2}{-}\varepsilon$} (s2p)
				(s1p) edge [loop right] node [right] {$1$} (s1p) 
				(s2p) edge [loop left] node [left] {$1$} (s2p)
				;
				
				\node[above of = s0, node distance = 0.7cm] {$\{a\}$};
				\node[above of = s2, node distance = 0.7cm] {$\{g\}$};
				\node[above of = s1, node distance = 0.7cm] {$\{f\}$};
				
				\node[above of = s0p, node distance = 0.7cm] {$\{a\}$};
				\node[above of = s2p, node distance = 0.7cm] {$\{g\}$};
				\node[above of = s1p, node distance = 0.7cm] {$\{f\}$};
				
		\end{tikzpicture}}
		\caption{The LMC used in  \Cref{ex:unbounded_reach}. The states $s$ and $t$ are $\varepsilon$-bisimilar.}
		\label{Figure: Example unbounded reachability}
	\end{figure}
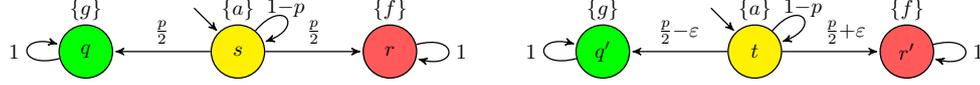
	
	\begin{example}
		\label{ex:unbounded_reach}
		Let $p\in (0,1)$, $\varepsilon< \frac{p}{2}$ and
		consider \Cref{Figure: Example unbounded reachability}, where $s \sim_\varepsilon t$.
		There,
		$\mathrm{Pr}_s (\lozenge   g) = \frac{1}{2}, \mathrm{Pr}_t (\lozenge   g) =  \frac{1}{2} - \frac{\varepsilon}{p}$ and $\mathbb{E}_s(N) = \mathbb{E}_t(N)  = \frac{1}{p}$.
		Hence, the bound in \Cref{thm:unbounded_reach} is met exactly: $\vert \mathrm{Pr}_s(\lozenge g) - \mathrm{Pr}_t(\lozenge g) \vert = \frac{\varepsilon}{p} = \varepsilon \cdot \mathbb{E}_s(N) = \varepsilon \cdot \mathbb{E}_t(N)$. 
	\end{example}
	
	Regarding $\varepsilon$-APBs and up-to-$(n, \varepsilon)$-bisimilarity, some preservation results w.r.t. the (approximate or robust) satisfaction of PCTL state-formulas can be found in the literature. Since, as we have seen in \Cref{Section: Definitions and Interrelation of Known Notions}, $s \sim_\varepsilon t$ implies both $s \equiv_\varepsilon t$ and $s \sim_\varepsilon^n t$ for any $n \in \mathbb{N}$ and any two states $s,t$, the following results also hold for $\varepsilon$-bisimilar states.
	
	An important property of $\varepsilon$-APBs is that related states satisfy the same \emph{$\varepsilon$-robust PCTL state formulas} $\Phi_{\mathrm{robust}}$, i.e., that $s \equiv_\varepsilon t$ implies that $s \vDash \Phi_{\mathrm{robust}}$ iff $t \vDash \Phi_{\mathrm{robust}}$, where $\vDash$ is the usual PCTL \emph{satisfaction relation} \cite{PoMC}. Intuitively, $\Phi_{\mathrm{robust}}$ is $\varepsilon$-robust if for all subformulas $\phi$ of $\Phi_{\mathrm{robust}}$ and all $s \in S$ either a strengthened version of $\phi$, obtained by making $\phi$'s probability thresholds harder to meet, holds in $s$, or even relaxing $\phi$'s probability thresholds is not sufficient to ensure that $s$ satisfies $\phi$. For details, see \cite{RPCTLMC}. 
	
	Furthermore, it was shown in \cite{SAAPBPCTL} that $(n, \varepsilon)$-bisimilar states approximately satisfy the same bounded PCTL state formulas. The fragment of PCTL considered does not allow unbounded until, and requires all until operator appearing in a formula to have the same time bound $k \in \mathbb{N}$. Under these assumptions, the precision of the approximation of satisfaction probabilities between $(n, \varepsilon)$-bisimilar states is proved to depend linearly on the parameters $n$ and $\varepsilon$, as well as the common step bound $k$ of the until operators. For details, see \cite{SAAPBPCTL}.

	
	\section{\texorpdfstring{$\varepsilon$-Perturbed Bisimulation}{Epsilon-Perturbed Bisimulation}}\label{Section: Epsilon-Quotients}
	In this section we consider finite LMCs. In \cite{ABM}, Kiefer and Tang define the notion of \emph{$\varepsilon$-quotients} for $\varepsilon \ge 0$.
	Their goal is to construct, from a given \emph{perturbed} LMC $\mathcal{M}'$, an abstraction that is as close as possible to the exact bisimulation quotient of an unknown, unperturbed LMC $\mathcal{M}$ corresponding to $\mathcal{M}'$. This inspires us to introduce \emph{$\varepsilon$-perturbed bisimulations}, which relate two LMCs iff they can be made probabilistically bisimilar by small perturbations of their transition probabilities. Since we require $\varepsilon$-perturbed bisimulations to be equivalences, these relations are well-suited for the construction of quotients of a given model.
	
	Like the $\varepsilon$-quotients of \cite{ABM}, we base our definition on \emph{$\varepsilon$-perturbations} of LMCs.  
	\begin{definition}[\cite{ABM}]\label{Definition: Epsilon-Perturbation of LMCs}
		$\mathcal{M}' = (S, \prob', s_{init}, l)$ is an \emph{$\varepsilon$-perturbation} of
		$\mathcal{M} = (S, \prob, s_{init}, l)$ if $\Vert \prob(s) - \prob'(s) \Vert_1 \leq \varepsilon$ for all $s \in S$. 
	\end{definition}
	
	$\mathcal{M}$ and any of its $\varepsilon$-perturbations $\mathcal{M}'$ have the same state space and labeling, and we often write $S' = \{s' \mid s \in S\}$ for the state space of $\mathcal{M}'$. Hence, $\mathcal{M}$ and $\mathcal{M}'$ only differ in their transition distribution functions. 
	However, $\mathcal{M}'$ does not need to preserve the structure of $\mathcal{M}$, i.e., there can be transitions in $\mathcal{M}$ that have probability 0 in $\mathcal{M}'$ and vice versa. As the next lemma shows, the total probability mass of these transitions cannot exceed $\frac{\varepsilon}{2}$.
	\begin{restatable}{lemma}{LemmaDifferenceOfEpsilonPerturbationProbabilities}\label{Lemma: Difference of Epsilon-Perturbation Probabilities}
		For all $s \in S$ and $A \subseteq S$ it holds that $\vert \prob(s)(A) - \prob'(s')(A') \vert \leq \frac{\varepsilon}{2}$.
	\end{restatable}
	
	We now define the novel notion of \emph{$\varepsilon$-perturbed bisimulation}.
	\begin{definition}\label{Definition: Perturbed Epsilon Bisimulation}
		An equivalence $R \subseteq S \times S$ is called an \emph{$\varepsilon$-perturbed bisimulation} on $\mathcal{M}$ if there is an $\varepsilon$-perturbation $\mathcal{M}'$ of $\mathcal{M}$ such that $R$ is a bisimulation on $\mathcal{M}'$. Two states $s,t \in S$ are \emph{$\varepsilon$-perturbed bisimilar}, denoted $s \simeq_\varepsilon t$, if they are related by some $\varepsilon$-perturbed bisimulation. Given LMCs $\mathcal{M}$ and $\mathcal{N}$, then $\mathcal{M} \simeq_\varepsilon \mathcal{N}$ if $s_{init}^\mathcal{M} \simeq_\varepsilon s_{init}^\mathcal{N}$ in $\mathcal{M} \oplus \mathcal{N}$.
	\end{definition}
	
	In the terminology of \cite{ABM}, $\mathcal{M} \simeq_\varepsilon \mathcal{N}$ iff $\mathcal{M}$ and $\mathcal{N}$ have bisimilar $\varepsilon$-perturbations iff there are bisimilar $\varepsilon$-quotients of $\mathcal{M}$ and $\mathcal{N}$. If all states of $\mathcal{M}$ and $\mathcal{N}$ are reachable, even the stronger characterization $\mathcal{M} \simeq_\varepsilon \mathcal{N}$ iff $\mathcal{M}$ and $\mathcal{N}$ have \emph{isomorphic} $\varepsilon$-perturbations iff there are \emph{isomorphic} $\varepsilon$-quotients of $\mathcal{M}$ and $\mathcal{N}$ holds. Since the unique $0$-perturbation of any LMC is the LMC itself, $\mathcal{M} \simeq_0 \mathcal{N}$ iff $\mathcal{M} \sim \mathcal{N}$. Moreover, $\simeq_\varepsilon$ is symmetric and reflexive, but not transitive, which implies that $\simeq_\varepsilon$ is not necessarily an $\varepsilon$-perturbed bisimulation itself. 
	
	Let $s \sim_\varepsilon^* t$ denote that states $s$ and $t$ are related by a \emph{transitive} $\varepsilon$-bisimulation. We remark that both $\simeq_\varepsilon$ and $\sim_\varepsilon^*$ are definitions in the spirit of a notion called \emph{$\varepsilon$-lumpability} (or \emph{quasi-lumpability}), which describes that a LMC can be made exactly lumpable w.r.t. a given equivalence by slight changes (up to $\varepsilon$ in each value) of its transition probabilities \cite{EOLFMC,BQLMC,CRQLMC}. In contrast to the non-transitive case, any transitive $\varepsilon$-APB is also an $\varepsilon$-bisimulation. 
	
	\begin{figure}[t]
		\centering
		\resizebox{!}{0.07\textheight}{\begin{tikzpicture}[->,>=stealth',shorten >=1pt,auto, semithick]
				\tikzstyle{every state} = [text = black]
				
				\node[state] (s) [] {$s$}; 
				\node[] (temp1) [right of = s, node distance = 2cm] {};
				\node[state] (t) [right of = temp1, node distance = 2cm] {$t$};
				\node[] (temp2) [right of = t, node distance = 2cm] {};
				\node[state] (u) [right of = temp2, node distance = 2cm] {$u$};
				\node[state] (x) [below of = temp1, node distance = 1.1cm, fill = green] {$x$};
				\node[state] (y) [below of = temp2, node distance = 1.1cm, fill = yellow] {$y$}; 
				\node (sinit) [right of = t, node distance = 1.2cm] {};
				
				\path 
				(sinit) edge (t)
				(s) edge [bend right = 10] node [left, xshift = -0.1cm] {$0.5{-}\varepsilon$} (x)
				(s) edge node [above, pos = 0.2, xshift = 0.1cm] {$0.5{+}\varepsilon$} (y)
				
				(t) edge [bend left = 10] node [above, pos =0.3, xshift = -0.2cm] {$0.5$} (x)
				(t) edge [bend right = 10] node [above, pos =0.3, xshift = 0.2cm] {$0.5$} (y)
				
				(u) edge node [above, pos =0.2, xshift = -0.2cm] {$0.5{+}\varepsilon$} (x)
				(u) edge [bend left = 10] node [right, xshift = 0.1cm] {$0.5{-}\varepsilon$} (y)
				
				(x) edge [loop left  ]node [left] {$1$} (x)
				(y) edge [loop right ]node [right] {$1$} (y)
				;
				
				\node [left of = s, node distance = 0.65cm] {$\emptyset$}; 
				\node [left of = t, node distance = 0.65cm] {$\emptyset$}; 
				\node [right of = u, node distance = 0.65cm] {$\emptyset$}; 
				\node [right of = x, node distance = 0.75cm, yshift = -0.2cm] {$\{a\}$}; 
				\node [left of = y, node distance = 0.75cm, yshift = -0.2cm] {$\{b\}$}; 
		\end{tikzpicture}}
		\caption{An LMC in which there is no unique maximal transitive $\varepsilon$-bisimulation.}
		\label{Figure: Maximal Branching Bisimulation is not unique}
	\end{figure}
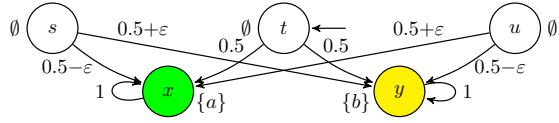
	
	The requirement of transitivity comes with the downside that there is not always a \emph{unique} largest transitive $\varepsilon$-bisimulation: in \Cref{Figure: Maximal Branching Bisimulation is not unique}, no transitive $\varepsilon$-bisimulation $R$ can contain both $(s,t)$ and $(t,u)$, as otherwise also $(s,u) \in R$ must hold. However, $s \sim_\varepsilon^* t$ and $t \sim_\varepsilon^* u$ as $R_1 = \{\{s,t\}, \{u\}, \{x\}, \{y\} \}$ and $R_2 = \{\{s\}, \{t,u\}, \{x\}, \{y\}\}$ are transitive $\varepsilon$-bisimulations. Hence, the union of all transitive $\varepsilon$-bisimulations in a given model is thus not always a transitive $\varepsilon$-bisimulation itself. This is different than in the non-transitive case, where $\sim_\varepsilon$ is always an $\varepsilon$-bisimulation \cite{AAPP}. Since $s \simeq_\varepsilon t$ and $t \simeq_\varepsilon u$ but $s \not \simeq_\varepsilon u$ in \Cref{Figure: Maximal Branching Bisimulation is not unique}, it follows that there is also not always a unique largest $\varepsilon$-perturbed bisimulation.
	
	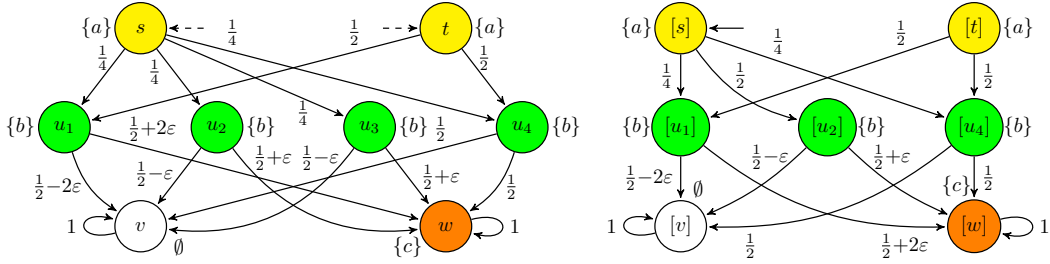
\begin{figure}[t]
		\centering
		\resizebox{!}{0.157\textheight}{\begin{tikzpicture}[->,>=stealth',shorten >=1pt,auto, semithick]
				\tikzstyle{every state} = [text = black]
				
				\node[state] (s0) [fill = yellow] {$s$}; 
				\node[] (tempu) [below of = s0, node distance = 1.7cm] {};
				\node[state] (u1) [left of = tempu, node distance = 1.3cm, fill = green] {$u_1$};
				\node[state] (u2) [right of = tempu, node distance = 1.3cm, fill = green] {$u_2$};
				\node[state] (u3) [right of = tempu, node distance = 3.9cm, fill = green] {$u_3$};
				\node[state] (u4) [right of = tempu, node distance = 6.5cm, fill = green] {$u_4$};
				\node (tempvw1) [right of = u2, node distance = 1.3cm] {};
				\node (tempvw) [below of = tempvw1, node distance = 1.7cm] {};
				\node[state] (v) [left of = tempvw, node distance = 2.6cm] {$v$}; 
				\node[state] (w) [right of = tempvw, node distance = 2.6cm, fill = orange] {$w$};
				\node (tempt) [right of = u3, node distance = 1.3cm] {};
				\node[state] (t) [above of = tempt, node distance = 1.7cm, fill = yellow] {$t$};
				\node (sinit1) [right of = s0, node distance = 1.2cm] {};
				\node (sinit2) [left of = t, node distance = 1.2cm] {};
				
				\path[dashed]
				(sinit1) edge (s0)
				(sinit2) edge (t);
				
				\path 
				(s0) edge node [above] {$\frac{1}{4}$} (u1)
				(s0) edge node [left, pos = 0.4, xshift = -0.1cm] {$\frac{1}{4}$} (u2)
				(s0) edge node [below, pos =0.7, xshift = 0.2cm] {$\frac{1}{4}$} (u3)
				(s0) edge node [above, pos = 0.2] {$\frac{1}{4}$} (u4)
				
				(t) edge [] node [above, pos = 0.2] {$\frac{1}{2}$} (u1)
				(t) edge [] node [above] {$\frac{1}{2}$} (u4)
				
				(v) edge [loop left] node {$1$} (v)
				(w) edge [loop right] node {$1$} (w)
				
				(u1) edge [bend right = 20] node [left] {$\frac{1}{2}{-}2\varepsilon$} (v)
				(u1) edge node [above, pos = 0.17, xshift = 0.1cm] {$\frac{1}{2}{+}2\varepsilon$} (w)
				
				(u2) edge node [left, pos = 0.45] {$\frac{1}{2}{-}\varepsilon$} (v)
				(u2) edge [bend right = 30] node [right, pos = 0.05, xshift = 0.1cm] {$\frac{1}{2}{+}\varepsilon$} (w)
				
				(u3) edge [bend left = 30] node [left, pos = 0.05] {$\frac{1}{2}{-}\varepsilon$} (v)
				(u3) edge node [right, pos = 0.5, xshift = 0.1cm] {$\frac{1}{2}{+}\varepsilon$} (w)
				
				(u4) edge node [above, pos = 0.17] {$\frac{1}{2}$} (v)
				(u4) edge [bend left = 20] node [right] {$\frac{1}{2}$} (w)
				
				;
				
				\node [left of = s0, node distance = 0.75cm] {$\{a\}$}; 
				\node [right of = t, node distance = 0.75cm] {$\{a\}$}; 
				\node [left of = u1, node distance = 0.75cm] {$\{b\}$}; 
				\node [right of = u2, node distance = 0.75cm] {$\{b\}$}; 
				\node [right of = u3, node distance = 0.75cm] {$\{b\}$}; 
				\node [right of = u4, node distance = 0.75cm] {$\{b\}$}; 
				\node [right of = v, node distance = 0.65cm, yshift = -0.35cm] {$\emptyset$}; 
				\node [left of = w, node distance = 0.65cm, yshift = -0.35cm] {$\{c\}$}; 
				
				\node[state] (s0) [fill = yellow, right of = t, node distance = 4cm] {$[s]$}; 
				\node[state] (u1) [below of = s0, node distance = 1.7cm, fill = green] {$[u_1]$};
				\node[state] (u2) [right of = u1, node distance = 2.5cm, fill = green] {$[u_2]$};
				\node[state] (u4) [right of = u2, node distance = 2.5cm, fill = green] {$[u_4]$};
				\node[state] (v) [below of = u1, node distance = 1.7cm] {$[v]$}; 
				\node[state] (w) [below of = u4, node distance = 1.7cm, fill = orange] {$[w]$};
				\node[state] (t) [above of = u4, node distance = 1.7cm, fill = yellow] {$[t]$};
				
				\node (initr) [right of = s0, node distance = 1.2cm] {};
				
				\path 
				(initr) edge (s0)
				(s0) edge node [left] {$\frac{1}{4}$} (u1)
				(s0) edge [bend right = 20] node [above] {$\frac{1}{2}$} (u2)
				(s0) edge [] node [above, pos = 0.3] {$\frac{1}{4}$} (u4)
				
				(t) edge [] node [above, pos = 0.2] {$\frac{1}{2}$} (u1)
				(t) edge [] node [right] {$\frac{1}{2}$} (u4)
				
				(v) edge [loop left] node {$1$} (v)
				(w) edge [loop right] node {$1$} (w)
				
				(u1) edge [] node [left] {$\frac{1}{2}{-}2\varepsilon$} (v)
				(u1) edge [bend right = 20] node [below, pos = 0.85] {$\frac{1}{2}{+}2\varepsilon$} (w)
				
				(u2) edge [bend left = 10 ]node [left, pos = 0.1] {$\frac{1}{2}{-}\varepsilon$} (v)
				(u2) edge [bend right = 10] node [right, pos = 0.1, xshift = 0.1cm] {$\frac{1}{2}{+}\varepsilon$} (w)
				
				(u4) edge [bend left = 20] node [below, pos = 0.85] {$\frac{1}{2}$} (v)
				(u4) edge node [right] {$\frac{1}{2}$} (w)
				
				;
				
				\node [left of = s0, node distance = 0.75cm] {$\{a\}$}; 
				\node [right of = t, node distance = 0.75cm] {$\{a\}$}; 
				\node [left of = u1, node distance = 0.75cm] {$\{b\}$}; 
				\node [right of = u2, node distance = 0.75cm] {$\{b\}$}; 
				\node [right of = u4, node distance = 0.75cm] {$\{b\}$}; 
				\node [above of = v, node distance = 0.65cm, xshift = 0.3cm] {$\emptyset$}; 
				\node [above of = w, node distance = 0.7cm, xshift = -0.3cm] {$\{c\}$}; 
				
		\end{tikzpicture}}
		\caption{Two LMCs $\mathcal{M}_s$ and $\mathcal{M}_t$ (left) with initial states $s$ and $t$, respectively, and $\mathcal{Q} = \xfrac{\mathcal{M}_s}{\sim}$ (right), demonstrating that $\simeq_\varepsilon$ and $\sim_\varepsilon^*$ can differentiate bisimilar models and are not additive.}
		\label{Figure: Example that simeq can differentiate bisimilar LMC}
	\end{figure}
	
	Now consider, for $\varepsilon < \frac{1}{4}$, the LMCs $\mathcal{M}_s$ and $\mathcal{M}_t$ on the left of \Cref{Figure: Example that simeq can differentiate bisimilar LMC}, with initial states $s$ and $t$, respectively. In both models, $\sim$ is the finest equivalence that contains $(u_2, u_3)$. Let $R_1$ be the finest equivalence that contains $(s,t), (u_1, u_2), (u_3, u_4)$, and let $R_2$ be the one that contains $(s,t), (u_1, u_3),(u_2, u_4)$. Both $R_1$ and $R_2$ are transitive $\varepsilon$-bisimulations, and since $u_1 \nsim_\varepsilon u_4$ no other transitive $\varepsilon$-bisimulation can contain $(s,t)$. Hence, no such relation contains $(u_2, u_3)$. 
	Let $\mathcal{Q} = \xfrac{\mathcal{M}_s}{\sim}$ be as on the right of the figure. 
	Then $\mathcal{M}_s \sim \mathcal{Q}$ and $\mathcal{M}_s \simeq_\varepsilon \mathcal{M}_t$ as, e.g., the $\varepsilon$-perturbations $\mathcal{M}_s'$ and $\mathcal{M}_t'$ that enforce $u_1' \sim u_2'$ and $u_3' \sim u_4'$ and are otherwise unchanged are bisimilar. However, there are no bisimilar $\varepsilon$-perturbations of $\mathcal{M}_t$ and $\mathcal{Q}$, i.e., $\mathcal{M}_t \not \simeq_\varepsilon \mathcal{Q}$. Since ${\simeq_0} = {\sim}$ this observation additionally yields that $\simeq_\varepsilon$ cannot be additive, as otherwise $\mathcal{M}_s \sim \mathcal{Q}$ and $\mathcal{M}_s \simeq_\varepsilon \mathcal{M}_t$ would have to imply $\mathcal{M}_t \simeq_\varepsilon \mathcal{Q}$. All in all, this leads to the following result.
	
	\begin{restatable}{proposition}{PropAnomalySimeq}\label{Proposition: Simeq can distinguish bisimilar states}
		The relation $\simeq_\varepsilon$ is not additive in the tolerances and can distinguish bisimilar LMCs in the following sense: there are LMCs $\mathcal{M}_1, \mathcal{M}_2$ and $\mathcal{N}$ such that $\mathcal{M}_1 \sim \mathcal{M}_2$ and $\mathcal{M}_1 \simeq_\varepsilon \mathcal{N}$, but $\mathcal{M}_2 \not \simeq_\varepsilon \mathcal{N}$.
	\end{restatable}
	
	This behavior of $\simeq_\varepsilon$ is in contrast to, e.g., $\sim_\varepsilon$, as $s_1 \sim s_2$ and $s_1 \sim_\varepsilon t$ always implies $s_2 \sim_\varepsilon t$. In particular, the non-additivity does not hinge on the existence of bisimilar states in the model. To see this consider, e.g., slight perturbations $\mathcal{M}_s'$ and $\mathcal{M}_t'$ of the LMCs on the left of \Cref{Figure: Example that simeq can differentiate bisimilar LMC}, where for some $\delta < \varepsilon$ we set $\prob(u_2)(v) = \frac{1}{2}-\varepsilon -\delta$ and $\prob(u_2)(w) = \frac{1}{2} + \varepsilon + \delta$, and leave the rest of the models unchanged. Then $u_2 \nsim u_3$ in $\mathcal{M}_s'$ and $\mathcal{M}_t'$, but still $\mathcal{M}'_s \simeq_\delta \mathcal{Q}$ and $\mathcal{M}'_s \simeq_\varepsilon \mathcal{M}'_t$ while $\mathcal{M}'_t \not \simeq_{\varepsilon + \delta} \mathcal{Q}$, where $\mathcal{Q}$ is again the (unperturbed) LMC on the right of the figure. Similar results hold for $\sim_\varepsilon^*$, as $\mathcal{M}_s \sim_\varepsilon^* \mathcal{M}_t$ and $\mathcal{M}_s \sim \mathcal{Q}$, but $\mathcal{M}_t \nsim_\varepsilon^* \mathcal{Q}$.
	
	We now discuss how $\simeq_{\varepsilon}$ relates to $\sim_\varepsilon$ and $\sim_\varepsilon^*$, starting with the direction from left to right. From \cite{ABMfull} it follows directly that $\mathcal{M} \simeq_\varepsilon \mathcal{N}$ implies $\mathcal{M} \sim_\varepsilon \mathcal{N}$. As we show next, the claim also holds when considering the stronger requirement of \emph{transitive} $\varepsilon$-bisimilarity.
	
	\begin{restatable}{lemma}{PerturbedEpsilonBisimImpliesEpsilonBisim}\label{Corollary: Perturbed Epsilon Bisim implies Epsilon Bisim}
		$\mathcal{M} \simeq_\varepsilon \mathcal{N}$ implies $\mathcal{M} \sim_\varepsilon^* \mathcal{N}$.
	\end{restatable}
	
	It is thus possible to transfer known results for $\sim_\varepsilon$ like, e.g., the preservation of approximate satisfaction of bounded PCTL state formulas \cite{SAAPBPCTL}, the exact preservation of $\varepsilon$-robust PCTL \cite{RPCTLMC}, or the bounds on finite horizon \cite{RBBTEAPC} and syntactically co-safe \cite{FMOSCSMDP,RDPTLCSS} LTL satisfaction probabilities to $\varepsilon$-perturbed bisimilar LMCs.
	
	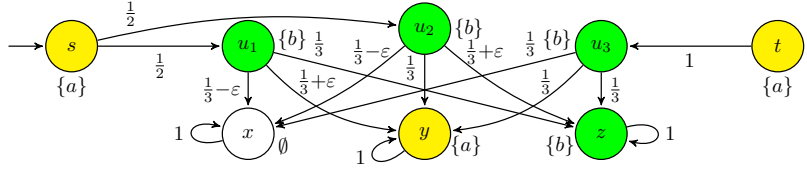
\begin{figure}[t]
		\centering
		\resizebox{!}{0.11\textheight}{\begin{tikzpicture}[->,>=stealth',shorten >=1pt,auto, semithick]
				\tikzstyle{every state} = [text = black]
				
				\node[state] (s) [fill = yellow] {$s$}; 
				\node[state] (u1) [right of = s, node distance = 3cm, fill = green] {$u_1$};
				\node[state] (u2) [right of = u1, node distance = 3cm, fill = green, yshift = 0.3cm] {$u_2$};
				\node[state] (u3) [right of = u1, node distance = 6cm, fill = green] {$u_3$};
				\node[state] (x) [below of = u1, node distance = 1.5cm] {$x$};
				\node[state] (y) [right of = x, node distance = 3cm, fill = yellow] {$y$};
				\node[state] (z) [right of = y, node distance = 3cm, fill = green] {$z$}; 
				\node[state] (t) [right of = u3, node distance = 3cm, fill = yellow] {$t$};
				\node (sinit) [left of = s, node distance = 1.2cm] {};
				
				\path 
				(sinit) edge (s)
				(s) edge node [below] {$\frac{1}{2}$} (u1)
				(s) edge [bend left = 10] node [above, pos = 0.1] {$\frac{1}{2}$} (u2)
				(t) edge node [below] {$1$} (u3)
				
				(u1) edge node [left] {$\frac{1}{3}{-}\varepsilon$} (x)
				(u1) edge [bend right = 20] node [right, pos = 0.2, yshift = 0.1cm] {$\frac{1}{3}{+}\varepsilon$} (y)
				(u1) edge node [above, pos = 0.15] {$\frac{1}{3}$} (z)
				
				(u2) edge [bend left = 10] node [left, pos = 0.05] {$\frac{1}{3}{-}\varepsilon$} (x)
				(u2) edge node [left, pos = 0.25] {$\frac{1}{3}$} (y)
				(u2) edge [bend right = 10] node [right, pos = 0.05, xshift = 0.1cm] {$\frac{1}{3}{+}\varepsilon$} (z)
				
				(u3) edge node [above, pos = 0.15] {$\frac{1}{3}$} (x)
				(u3) edge [bend left = 20] node [left, pos = 0.2, yshift = 0.1cm] {$\frac{1}{3}$} (y)
				(u3) edge node [right] {$\frac{1}{3}$} (z)
				
				(x) edge [loop left] node {$1$} (x)
				(y) edge [in=190,out=220,loop] node [left] {$1$} (y)
				(z) edge [loop right] node {$1$} (z)
				
				;
				\node [below of = s, node distance = 0.7cm] {$\{a\}$}; 
				\node [below of = t, node distance = 0.7cm] {$\{a\}$}; 
				\node [right of = u1, node distance = 0.75cm, yshift = 0.1cm] {$\{b\}$}; 
				\node [right of = u2, node distance = 0.75cm] {$\{b\}$}; 
				\node [left of = u3, node distance = 0.75cm, yshift = 0.1cm] {$\{b\}$}; 
				\node [right of = x, node distance = 0.6cm, yshift = -0.2cm] {$\emptyset$}; 
				\node [right of = y, node distance = 0.7cm, yshift = -0.2cm] {$\{a\}$}; 
				\node [left of = z, node distance = 0.7cm, yshift = -0.2cm] {$\{b\}$};
		\end{tikzpicture}}
		\caption{An LMC that demonstrates that $\simeq_\varepsilon$ is strictly finer than $\sim_\varepsilon^*$.}
		\label{Figure: Example that transitive epsilon bisim is not enough for simeq}
	\end{figure}
	
	Regarding the reverse implication, consider \Cref{Figure: Example that transitive epsilon bisim is not enough for simeq}. There, the finest equivalence that relates $u_1, u_2$ and $u_3$ and contains $(s,t)$ is a transitive $\varepsilon$-bisimulation. However, there is no $\varepsilon$-perturbation of the LMC in which $s$ and $t$ are bisimilar. Hence, $s \sim_\varepsilon^* t$, but $s \not \simeq_\varepsilon t$.
	
	\begin{lemma}\label{Corollary: Perturbed Finer Than Transitive}
		$\simeq_\varepsilon$ is strictly finer than $\sim_\varepsilon^*$.
	\end{lemma}
	
	In fact, $\varepsilon$-bisimilarity is not even guaranteed to imply $\delta$-perturbed bisimilarity if $\varepsilon \ll \delta$, or if the Markov chains in question are graph-isomorphic. 
	
	\begin{restatable}{theorem}{TheoremEpsilonBisimilarityDoesNotImplyExistenceOfCommonQuotient}\label{Theorem: Epsilon-Bisimulation does not imply existence of common 1/4-quotient}
		Let $\varepsilon \in \left (0, \frac{1}{4}\right]$. There are LMCs $\mathcal{M}$ and $\mathcal{N}$ with $\mathcal{M} \sim_\varepsilon \mathcal{N}$ but $\mathcal{M} \not \simeq_{\frac{1}{4}} \mathcal{N}$.
	\end{restatable} 
	
	\begin{figure}[t!]
		\centering
		\resizebox{!}{0.137\textheight}{
			\begin{tikzpicture}[->,>=stealth',shorten >=1pt,auto, semithick]
				\tikzstyle{every state} = [text = black]
				
				\node[state] (s) [fill = yellow] {$s$};
				\node[state] (s2) [fill = green, right of = s, node distance = 2.5cm] {$s_2$};
				\node[state] (s1) [above of = s2, node distance = 1.5cm, fill = green] {$s_{1}$};
				\node[state] (s3) [below of = s2, node distance = 1.5cm, fill = green] {$s_3$};
				\node[state] (x) [right of = s2, node distance = 2.5cm] {$x$};
				\node (sinit) [above of = s, node distance = 1cm] {};
				
				\path
				(sinit) edge (s)
				
				(s) edge node [above] {$\frac{4}{9}$} (s1)
				(s) edge node [above] {$\frac{4}{9}$} (s2)
				(s) edge node [below] {$\frac{1}{9}$} (s3)
				
				(s1) edge node [above,sloped] {$0.5 {+} \varepsilon$} (x)
				(s2) edge node [above, pos = 0.4] {$0.5 {-} \varepsilon$} (x)
				(s3) edge node [below,sloped] {$0.5 {-} 3\varepsilon$} (x)
				
				(s1) edge [loop right] node {$0.5 {-} \varepsilon$} (s1)
				(s2) edge [loop below] node [right, pos =0.1] {$0.5 {+} \varepsilon$} (s2)
				(s3) edge [loop right] node {$0.5 {+} 3\varepsilon$} (s3)
				
				(x) edge [loop above] node {$1$} (x)
				;
				
				\node[state] (t) [fill = yellow, right of = x, node distance = 2cm] {$t$};
				\node[state] (t2) [right of = t, node distance = 2.5cm, fill = green] {$t_2$};
				\node[state] (t1) [above of = t2, node distance = 1.5cm, fill = green] {$t_1$};
				\node[state] (t3) [below of = t2, node distance = 1.5cm, fill = green] {$t_3$};
				\node[state] (y) [right of = t2, node distance = 2.5cm] {$y$};
				\node (tinit) [above of = t, node distance = 1cm] {};
				
				\path
				(tinit) edge (t)
				
				(t) edge node [above] {$\frac{1}{9}$} (t1)
				(t) edge node [above] {$\frac{4}{9}$} (t2)
				(t) edge node [below] {$\frac{4}{9}$} (t3)
				
				(t1) edge node [above,sloped] {$0.5 {+} 2\varepsilon$} (y)
				(t2) edge node [above, pos = 0.4] {$0.5$} (y)
				(t3) edge node [below,sloped] {$0.5 {-} 2\varepsilon$} (y)
				
				(t1) edge [loop right] node {$0.5 {-} 2\varepsilon$} (t1)
				(t2) edge [loop below] node [right, pos =0.1] {$0.5$} (t2)
				(t3) edge [loop right] node {$0.5 {+} 2\varepsilon$} (t3)
				
				(y) edge [loop above] node {$1$} (y)
				;
				
				\node[below of = s, node distance = 0.7cm] {$\{a\}$};
				\node[left of = s1, node distance = 0.75cm] {$\{b\}$};
				\node[above of = s2, node distance = 0.7cm] {$\{b\}$};
				\node[left of = s3, node distance = 0.75cm] {$\{b\}$};
				\node[below of = x, node distance = 0.7cm] {$\{c\}$};
				
				\node[below of = t, node distance = 0.7cm] {$\{a\}$};
				\node[left of = t1, node distance = 0.75cm] {$\{b\}$};
				\node[above of = t2, node distance = 0.7cm] {$\{b\}$};
				\node[left of = t3, node distance = 0.75cm] {$\{b\}$};
				\node[below of = y, node distance = 0.7cm] {$\{c\}$};
		\end{tikzpicture}}
		\caption{The LMCs $\mathcal{M}_2$ (left) and $\mathcal{N}_2$ (right), as used in the proof of \Cref{Theorem: No common nepsilon quotient with graph isomorphism}.}
		\label{Figure: Example LMCs for case n = 2 for construction in Theorem: No common nepsilon quotient with graph isomorphism}
	\end{figure}
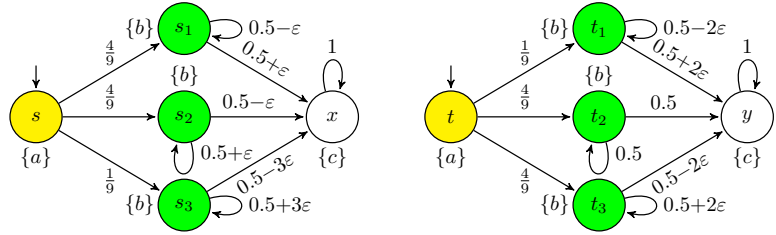
	
	\begin{restatable}{theorem}{NoCommonNEpsilonQuotientWithGraphIsomorphism}\label{Theorem: No common nepsilon quotient with graph isomorphism}
		There is a family $\mathcal{F} = \{(\mathcal{M}_n, \mathcal{N}_n) \mid n \in \mathbb{N}_{\geq 1}\}$ of pairs of finite LMCs such that, for all $n \in \mathbb{N}_{\geq 1}$ and $\varepsilon \in \left(0, \frac{1}{n \cdot (n+1)^2}\right]$, $\mathcal{M}_n$ and $\mathcal{N}_n$ are graph-isomorphic and $\varepsilon$-bisimilar, but $\mathcal{M}_n \not \simeq_\delta \mathcal{N}_n$ for any $\delta < n \varepsilon$.
	\end{restatable}
	\begin{proof}[Proof sketch]
		We sketch the case $n = 2$, with $\mathcal{M}_2$ and $\mathcal{N}_2$ as in \Cref{Figure: Example LMCs for case n = 2 for construction in Theorem: No common nepsilon quotient with graph isomorphism}, $\varepsilon \in \left(0, \frac{1}{18}\right]$ and $\sim_\varepsilon$ the symmetric and reflexive closure of $\{(s,t), (s_1, t_1), (s_1, t_2), (s_2, t_2), (s_2, t_3), (s_3, t_3), (x,y)\}$. 
		Any bisimilar perturbations $\mathcal{M}_2'$ and $\mathcal{N}_2'$ must ensure $s' \sim t'$. The smallest (w.r.t. the required tolerances) perturbations that achieve this make $s_2', s_3'$ and $t_3'$, as well as $s_1', t_1'$ and $t_2'$, bisimilar, and set the total probability mass from $s'$ (resp. $t'$) to reach these (sets of) state(s) to $\frac{1}{2}$ each. 
		But this requires a perturbation by at least $\delta = \frac{1}{9} \geq 2 \varepsilon$. 
	\end{proof}
	
	\begin{remark}
		\Cref{Theorem: Epsilon-Bisimulation does not imply existence of common 1/4-quotient,Theorem: No common nepsilon quotient with graph isomorphism} seem to resemble results of \cite{ABM,ABMfull}. There, an LMC is presented in which a specific order of merging $\varepsilon$-bisimilar states results in an approximate quotient that requires tolerance $\geq \frac{1}{4}$, and a family of LMCs is provided \cite[Thm. 12]{ABMfull} in which merging $\varepsilon$-bisimilar states yields an approximate quotient that requires tolerance $\geq n \varepsilon$. Our results differ in that we consider the existence of bisimilar $\varepsilon$-perturbations of two LMCs, and in that we show that no suitable \emph{smaller} tolerance exists.
	\end{remark}

	The observation that $\simeq_\varepsilon$ is strictly finer than $\sim_\varepsilon$ (and even $\sim_\varepsilon^*$) raises the question whether there are logical properties which are preserved under $\simeq_\varepsilon$, but not necessarily under $\sim_\varepsilon^{(*)}$. It is future work to make this precise. Here, we note that the bound for reachability probabilities from \Cref{thm:unbounded_reach} remains tight under $\simeq_\varepsilon$: the LMCs $\mathcal{M}$ and $\mathcal{N}$ in  \Cref{Figure: Example unbounded reachability} satisfy $\mathcal{M} \simeq_\varepsilon \mathcal{N}$, but the bounds are tight by \Cref{ex:unbounded_reach}. 
	
	The following theorem characterizes $\simeq_\varepsilon$ in terms of transitive $\varepsilon$-bisimulations that satisfy an additional \emph{centroid property} specified as in \Cref{centroid condition} below.
	
	\begin{restatable}{theorem}{CharacterizationPertubedEpsBisimTransitiveEpsBisim}\label{Theorem: Characterisation Perturbed Epsilon Bisimilar and Transitive Epsilon Bisimulation}
		The following statements are equivalent: 
		\begin{itemize}
			\item [\emph{(i)}] $\mathcal{M} \simeq_\varepsilon \mathcal{N}$. 
			\item [\emph{(ii)}] There is an $\varepsilon$-perturbation of $\mathcal{M} \oplus \mathcal{N}$ in which $s_{init}^\mathcal{M} \sim s_{init}^\mathcal{N}$. 
			\item [\emph{(iii)}] There is a transitive $\varepsilon$-bisimulation $R$ on $\mathcal{M} \oplus \mathcal{N}$ with $(s_{init}^\mathcal{M}, s_{init}^\mathcal{N}) \in R$ such that for each $A \in \xfrac{S}{R}$, where $S$ is the disjoint union of $S^\mathcal{M}$ and $S^\mathcal{N}$, there is a $\prob_A^* \in Distr(\xfrac{S}{R})$ with 
			\begin{align}
				\vert \prob(s)(C) - \prob_A^*(C) \vert \leq \frac{\varepsilon}{2} \text{ for all} s \in A \text{ and all } R \text{-closed sets } C. \label{centroid condition}
			\end{align}
		\end{itemize}
	\end{restatable}
	
	From the next lemma it follows immediately that, for a given equivalence $R \subseteq S \times S$, the centroid property in \Cref{centroid condition} can be checked efficiently. 
	
	\begin{restatable}{lemma}{TauStarConditionEquivalence}\label{Lemma: Tau-Star Condition Equivalence}
		For a finite set $X$ and $\mu_1, \dots, \mu_k \in Distr(X)$, the following are equivalent: 
		\begin{itemize}
			\item [\emph{(i)}] There exists $\mu^* \in Distr(X)$ with $\vert \mu_l(B) - \mu^*(B) \vert \leq \frac{\varepsilon}{2}$ for all $l \in \{1,\dots, k\}$ and $B \subseteq X$.
			\item [\emph{(ii)}]There exists $\mu \in Distr(X)$ with $\Vert \mu_l - \mu \Vert_1 \leq \varepsilon$ for all $l \in \{1, \dots, k\}$.
			\item [\emph{(iii)}] The following linear constraint system over non-negative variables $\delta_{l, i}$ and $x_i$ for $l \in \{1,\dots, k\}$ and $i \in X$ is solvable: 
			\begin{align*}
				\sum_{i \in X} x_i = 1 \quad \text{and} \quad  x_i - \mu_l(i) \leq \delta_{l, i} \quad \text{and} \quad \mu_l(i) - x_i \leq \delta_{l,i} \quad \text{and} \quad  \sum_{i \in X} \delta_{l,i} \leq \varepsilon.
			\end{align*}
		\end{itemize}
		The equivalence to \emph{(iii)} further implies that $\mu^* = \mu$ can be computed in polynomial time.
	\end{restatable}
	
	However, as we show next, for given $\mathcal{M}, \mathcal{N}$ and $\varepsilon$ it is \textsf{NP}-complete to decide if $\mathcal{M} \simeq_\varepsilon \mathcal{N}$ and if $\mathcal{M} \sim_\varepsilon^* \mathcal{N}$. This stands in contrast to the polynomial time computability of $\sim_\varepsilon$ \cite{AAPP}, which is possible in $\mathcal{O}(\vert S \vert^7)$ by iteratively solving maximum flow problems á la \cite{PSfPP,PTATPBS}. Our proofs are inspired by \cite[Thm. 1]{ABMfull}, which proves that deciding if a LMC has an $\varepsilon$-quotient with a fixed number of states is \textsf{NP}-complete. 
	
	\begin{restatable}{theorem}{TheoremNpCompleteCommonEpsilonQuotient}\label{Theorem: NP-completeness of common epsilon-quotient}
		For given finite LMCs $\mathcal{M}$ and $\mathcal{N}$ and given $\varepsilon \in (0,1]$, it is \emph{\textsf{NP}}-complete to decide if \emph{(i)} $\mathcal{M} \simeq_\varepsilon \mathcal{N}$ and to decide if \emph{(ii)} $\mathcal{M} \sim_\varepsilon^* \mathcal{N}$.
	\end{restatable}
	
	Nevertheless, one can check in polynomial time if a \emph{given} equivalence $R$ is a transitive $\varepsilon$-bisimulation or an $\varepsilon$-perturbed bisimulation. Since constructing quotients w.r.t. these relations by collapsing equivalence classes into single states can be done efficiently as well, the notions are therefore still suitable for constructing abstractions in practical applications.
	
	\begin{restatable}{proposition}{PropPolyTime}
		Given an equivalence $R$, one can decide in polynomial time if \emph{(i)} $R$ is a transitive $\varepsilon$-bisimulation and if \emph{(ii)} $R$ is an $\varepsilon$-perturbed bisimulation. 
	\end{restatable}

	
	\section{\texorpdfstring{Branching and Weak $\varepsilon$-Bisimulation}{Branching and Weak Epsilon-Bisimulation}}\label{Section: Weak and Branching Epsilon Bisimulation}
	We now introduce approximate versions of branching and weak probabilistic bisimulation. A similar approach has been discussed sporadically in the context of noninterference under the term ``weak bisimulation with precision $\varepsilon$'' \cite{TFAANP,SAPNRP,PAAAPN,QANPS,MCPS,NAWPB}. While our notion of branching $\varepsilon$-bisimilarity is a branching variant of transitive $\varepsilon$-bisimilarity $\sim_\varepsilon^*$, the weak $\varepsilon$-bisimilarity we propose is a weak variant of $\sim_\varepsilon$. Hence, the former is tailored to the construction of quotients of a given model, while the latter is closer to classic process relations. 
	\begin{definition}\label{Definition: Branching Epsilon Bisimulation}
		An equivalence $R \subseteq S \times S$ is a \emph{branching $\varepsilon$-bisimulation} if for all $(s,t) \in R$ and all $R$-closed sets $A \subseteq S$ it holds that
		\begin{align*}
			\text{\emph{(i)} } l(s) = l(t) \qquad \text{ and } \qquad \text{\emph{(ii)} } \vert \mathrm{Pr}_s([s]_R \Until A) - \mathrm{Pr}_t([t]_R \Until A) \vert \leq \varepsilon.
		\end{align*}
		We call $s,t \in S$ \emph{branching $\varepsilon$-bisimilar}, written $s \approx^b_\varepsilon t$, if they are related by a branching $\varepsilon$-bisimulation. LMCs $\mathcal{M}$ and $\mathcal{N}$ are \emph{branching $\varepsilon$-bisimilar}, written $\mathcal{M} \approx_\varepsilon^b \mathcal{N}$, if $s_{init}^\mathcal{M} \approx_\varepsilon^b s_{init}^\mathcal{N}$ in $\mathcal{M} \oplus \mathcal{N}$.
	\end{definition}
	
	We require branching $\varepsilon$-bisimulations to be equivalences, as their goal is to abstract from stutter steps inside a state's equivalence class. Because of transitivity, \Cref{Definition: Branching Epsilon Bisimulation} can also be formulated in the style of \Cref{Definition: Epsilon-Bisimulation} and should thus not be understood as an explicit extension of \Cref{Definition: Epsilon-APB}. With the same arguments as for $\sim_\varepsilon^*$ and $\simeq_\varepsilon$, transitivity causes that there may not be a unique maximal branching $\varepsilon$-bisimulation, that $\approx_\varepsilon^b$ is not additive in the tolerances, and that it can differentiate bisimilar models: the first claim follows from $s \approx_\varepsilon^b t$ and $t \approx_\varepsilon^b u$ but $s \not \approx_\varepsilon^b u$ in \Cref{Figure: Maximal Branching Bisimulation is not unique}, the others from ${\sim_\varepsilon^*} = {\approx_\varepsilon^b}$ in \Cref{Figure: Example that simeq can differentiate bisimilar LMC}.
	
	\begin{definition}
		A reflexive and symmetric relation $R \subseteq S \times S$ is a \emph{weak $\varepsilon$-bisimulation} if for all $(s,t) \in R$ and all $A \subseteq S$ it holds that 
		\begin{align*}
			\text{\emph{(i)} } l(s) = l(t) \qquad \text{ and } \qquad \text{\emph{(ii) }} \mathrm{Pr}_s(L(s) \Until A) \leq \mathrm{Pr}_t(L(t) \Until R(A)) + \varepsilon.
		\end{align*}
		We call $s,t \in S$ \emph{weakly $\varepsilon$-bisimilar}, written $s \approx^w_\varepsilon t$, if they are related by a weak $\varepsilon$-bisimulation. LMCs $\mathcal{M}$ and $\mathcal{N}$ are \emph{weakly $\varepsilon$-bisimilar}, written $\mathcal{M} \approx_\varepsilon^w \mathcal{N}$, if $s_{init}^\mathcal{M} \approx_\varepsilon^w s_{init}^\mathcal{N}$ in $\mathcal{M} \oplus \mathcal{N}$.
	\end{definition}
	
	In contrast to branching $\varepsilon$-bisimulations, we do not require transitivity for weak $\varepsilon$-bisimulations. As it turns out, $\approx_\varepsilon^w$ is instead additive in the tolerances.
	\begin{restatable}{lemma}{LemAdditivityWeakEpsilonBisim}
		$s \approx_\varepsilon^w t$ and $t \approx_\delta^w u$ implies $s \approx_{\varepsilon + \delta}^w u$. 	
	\end{restatable}
	
	Further, $\approx^w_0$ and $\approx^b_0$ coincide with $\approx^w$ and $\approx^b$, respectively, so our notions are conservative extensions of their exact counterparts. In particular, as ${\approx^w} = {\approx^b}$ for LMCs \cite{WBFPP}, it follows that ${\approx_0^w} = {\approx_0^b}$. For $\varepsilon > 0$ the notions can, however, become incomparable. This is different compared to the nonprobabilistic case, where $\approx^b$ is strictly finer than $\approx^w$ \cite{BTABS}.
	
	\begin{figure}[t]
		\centering
		\resizebox{!}{0.14\textheight}{\begin{tikzpicture}[->,>=stealth',shorten >=1pt,auto, semithick]
				\tikzstyle{every state} = [text = black]
				
				\node[state] (s1) [fill = yellow] {$s$}; 
				\node[state] (s2) [above of = s1, fill = yellow, node distance = 2.25cm] {$s_1$};
				\node[state] (x) [right of = s1, node distance = 2.75cm] {$x$};
				\node[state] (y) [fill = green, above of = x, node distance = 2.25cm] {$y$};
				\node[state] (t1) [fill = yellow, right of = x, node distance = 2.75cm] {$t$};
				\node[state] (t2) [fill = yellow, above of = t1, node distance = 2.25cm] {$t_1$};
				\node (temp1) [above of = t1, node distance = 1.25cm] {};
				\node (initl) [left of = s, node distance = 1.2cm] {};
				
				\path 
				(initl) edge (s)
				(s1) edge node [above] {$\frac{1}{2}$} (x)
				(s1) edge node [left] {$\frac{1}{2}$} (s2)
				(s2) edge node [above, xshift = 0.1cm] {$\frac{1}{4}$} (x)
				(s2) edge node [above] {$\frac{3}{4}$} (y)
				(t1) edge node [above] {$\frac{1}{2}{+}\varepsilon$} (x)
				(t1) edge node [right] {$\frac{1}{2}{-}\varepsilon$} (t2)
				(t2) edge node [above] {$\frac{3}{4} {-} \varepsilon$} (y)
				(t2) edge node [above, xshift = -0.1cm] {$\frac{1}{4} {+} \varepsilon$} (x)
				(y) edge [loop below] node [right, pos = 0.2] {$1$} (y)
				(x) edge [loop above] node [left, pos = 0.2] {$1$} (x)
				;

				\node [right of = s1, node distance = 0.7cm, yshift = -0.3cm] {$\{a\}$};
				\node [left of = s2, node distance = 0.75cm] {$\{a\}$};
				\node [right of = x, node distance = 0.7cm, yshift = -0.3cm] {$\emptyset$};
				\node [left of = y, node distance = 0.65cm, yshift = -0.4cm] {$\{b\}$};
				\node [right of = t1, node distance = 0.75cm] {$\{a\}$};
				\node [right of = t2, node distance = 0.75cm] {$\{a\}$};
				
				\node[] (temp) [right of = temp1, node distance = 3.5cm] {};
				\node[state] (s) [fill = yellow, above of = temp, node distance = 0.75cm] {$s$}; 
				\node[state] (t) [fill = yellow, below of = temp, node distance = 0.75cm] {$t$};
				\node[state] (v) [right of = temp, node distance = 3cm, fill = yellow] {$v$};
				\node[state] (u) [above of = v, node distance = 1.4cm, fill = yellow] {$u$};
				\node[state] (w) [below of = v, node distance = 1.4cm, fill = yellow] {$w$};
				\node[state] (x) [right of = s, node distance = 6cm] {$x$};
				\node[state] (y) [right of = t, node distance = 6cm, fill = green] {$y$};
				\node (initr) [above of = s, node distance = 1.2cm] {};
				
				\path 
				(initr) edge (s)
				(s) edge node [above] {$\frac{1}{2}$} (u)
				(s) edge node [below, pos = 0.1] {$\frac{1}{2}$} (v)
				(t) edge node [above, pos = 0.1] {$\frac{1}{2}$} (v)
				(t) edge node [below] {$\frac{1}{2}$} (w)
				
				(u) edge node [above] {$\frac{1}{2}{+}\varepsilon$} (x)
				(u) edge [bend left = 10] node [below, pos = 0.08] {$\frac{1}{2}{-}\varepsilon$} (y)
				
				(w) edge node [below] {$\frac{1}{2}{+}\varepsilon$} (y)
				(w) edge [bend right = 10] node [above, pos = 0.08] {$\frac{1}{2}{-}\varepsilon$} (x)
				
				(v) edge [bend right = 10] node [above, pos = 0.1] {$\frac{1}{2}$} (x)
				(v) edge [bend left = 10] node [below, pos = 0.1] {$\frac{1}{2}$} (y)
				
				(x) edge [loop right] node {$1$} (x)
				(y) edge [loop right] node {$1$} (y)
				;		
				
				\node [left of = s, node distance = 0.75cm] {$\{a\}$}; 
				\node [left of = t, node distance = 0.75cm] {$\{a\}$}; 
				\node [right of = u, node distance = 0.75cm, yshift = 0.2cm] {$\{a\}$}; 
				\node [right of = w, node distance = 0.75cm, yshift = -0.2cm] {$\{a\}$}; 
				\node [above of = v, node distance = 0.65cm, xshift = -0.2cm] {$\{a\}$}; 
				\node [above of = x, node distance = 0.7cm] {$\emptyset$}; 
				\node [below of = y, node distance = 0.7cm] {$\{b\}$}; 
				
		\end{tikzpicture}}
		\caption{The LMCs used in the proof of \Cref{Proposition: Branching and Weak are incomparable}.}
		\label{Figure: Combined LMCs for proof of Proposition: Branching and Weak are incomparable}
	\end{figure}
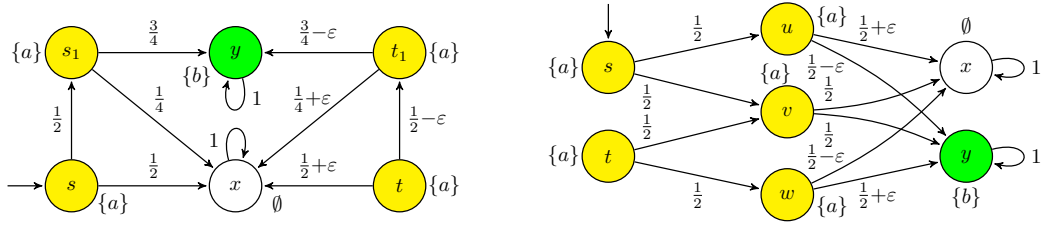
	
	\begin{restatable}{proposition}{ProbBranchingDoesNotImplyWeak}\label{Proposition: Branching and Weak are incomparable}
		For $0 < \varepsilon < \frac{1}{4}$, $s \approx^b_\varepsilon t \nRightarrow s \approx^w_\varepsilon t$ and $s \approx^w_\varepsilon t \nRightarrow s \approx_\varepsilon^b t$.
	\end{restatable}
	\begin{proof}
		Let $\varepsilon \in (0, \frac{1}{4})$ and consider \Cref{Figure: Combined LMCs for proof of Proposition: Branching and Weak are incomparable}. In the left LMC, $s \approx_\varepsilon^b t$, as the largest branching $\varepsilon$-bisimulation is induced by the equivalence classes $\{\{s, t\}, \{s_1, t_1\}, \{x\}, \{y\}\}$ and, in particular, $s \not \approx_\varepsilon^b s_1$ and $t \not \approx_\varepsilon^b t_1$. However, $s \not \approx_\varepsilon^w t$ as $\mathrm{Pr}_t(L(t) \Until \{x\}) = \frac{5}{8} + \frac{5}{4} \varepsilon - \varepsilon^2 > \frac{5}{8} + \varepsilon = \mathrm{Pr}_s(L(s) \Until \{x\}) + \varepsilon$. Furthermore, in the right LMC, $s \approx_\varepsilon^w t$ while $s \not \approx_\varepsilon^b t$ since any branching $\varepsilon$-bisimulation $R$ that contains $(s,t)$ must also contain $(u,w)$ due to transitivity, which is not possible as, e.g., $\vert \mathrm{Pr}_u([u]_R \Until [x]_R) - \mathrm{Pr}_w([w]_R \Until [x]_R) \vert > \varepsilon$.
	\end{proof}
	
	The major difference between $\approx_\varepsilon^w$, $\approx_\varepsilon^b$ and $\sim_\varepsilon, \equiv_\varepsilon$ is that the former can abstract from (some) stutter steps. Consequently, if no stuttering is possible, i.e., when $\prob(s)(L(s)) = 0$ for all $s \in S$, we have ${\sim_\varepsilon^*} = {\approx_\varepsilon^b}$ and ${\sim_\varepsilon} = {\approx_\varepsilon^w}$.
	Otherwise, the notions become incomparable.
	
	\begin{figure}[t]
		\centering
		\resizebox{!}{0.1\textheight}{\begin{tikzpicture}[->,>=stealth',shorten >=1pt,auto, semithick]
				\tikzstyle{every state} = [text = black]
				
				\node[state] (s) [fill = yellow] {$s$}; 
				\node[] (temp) [right of = s, node distance = 3cm] {};
				\node[state] (x1) [above of = temp, fill = green, node distance = 0.6cm] {$x_1$};
				\node[state] (x2) [below of = temp, node distance = 0.6cm] {$x_2$};
				\node[state] (t) [fill = yellow, right of = s, node distance = 6cm] {$t$};
				\node(sinit) [above of = s, node distance = 1.1cm] {};
				
				\path 
				(sinit) edge (s)
				(s) edge[bend right = 8] node [above, pos = 0.3] {$\varepsilon_1$} (x1)
				(s) edge [bend left = 8]node [below, pos = 0.3] {$\varepsilon_2$} (x2)
				(s) edge [loop below] node [left, pos = 0.85] {$1{-}\varepsilon_1{-}\varepsilon_2$} (s)
				(x1) edge [loop left] node {$1$} (x1)
				(x2) edge [loop left] node {$1$} (x2)
				(t) edge [bend left = 8] node [above, pos = 0.3] {$\varepsilon_2$} (x1)
				(t) edge [bend right = 8] node [below, pos = 0.3] {$\varepsilon_1$} (x2)
				(t) edge [loop below] node [right, pos = 0.15] {$1{-}\varepsilon_1{-}\varepsilon_2$} (t)
				;
				
				\node [left of = s, node distance = 0.75cm] {$\{a\}$};
				\node [right of = t, node distance = 0.75cm] {$\{a\}$};
				\node [right of = x1, node distance = 0.8cm, yshift = 0.2cm] {$\{b\}$};
				\node [right of = x2, node distance = 0.7cm, yshift = -0.2cm] {$\emptyset$};

				\node[state] (sr) [fill = yellow, right of = t, node distance = 2.75cm] {$s$}; 
				\node[state] (s1r) [right of = sr, node distance = 1.7cm, fill = yellow] {$s_1$};
				\node[state] (xr) [right of = s1r, node distance = 1.7cm, fill = green] {$x$}; 
				\node[state] (yr) [below of = xr, node distance = 1.2cm] {$y$};
				\node[state] (t2r) [right of = xr, node distance = 1.7cm, fill = yellow] {$t_1$}; 
				\node[state] (t1r) [right of = t2r, node distance = 1.7cm, fill = yellow] {$t$};
				\node (initr) [left of = sr, node distance = 1.2cm] {};
				
				\path 
				(initr) edge (sr)
				(sr) edge node [above, pos = 0.4] {$1{-}\varepsilon$}  (s1r)
				(sr) edge [loop above] node [left, pos =0.2] {$\frac{\varepsilon}{2}$} (sr)
				(sr) edge [bend right = 20] node [below, pos = 0.2] {$\frac{\varepsilon}{2}$} (yr)
				
				(s1r) edge node [above, pos = 0.4] {$1{-}\varepsilon$}  (xr)
				(s1r) edge [loop above] node [left, pos =0.2] {$\frac{\varepsilon}{2}$} (s1r)
				(s1r) edge [bend right = 20] node [above, pos = 0.5] {$\frac{\varepsilon}{2}$} (yr)
				
				(t1r) edge node [above, pos = 0.3] {$1$}  (t2r)
				
				(t2r) edge node [above, pos = 0.3] {$1$}  (xr)
				
				(yr) edge [loop right] node {$1$} (yr)
				(xr) edge [loop above] node [left, pos = 0.2]{$1$} (xr)
				;
				
				\node [below of = sr, node distance = 0.7cm] {$\{a\}$};
				\node [below of = s1r, node distance = 0.7cm] {$\{a\}$}; 
				\node [below of = t1r, node distance = 0.7cm] {$\{a\}$}; 
				\node [below of = t2r, node distance = 0.7cm] {$\{a\}$}; 
				\node [right of = xr, node distance = 0.6cm, yshift = 0.4cm] {$\{b\}$}; 
				\node [right of = yr, node distance = 0.55cm, yshift = 0.4cm] {$\emptyset$};
		\end{tikzpicture}}
		\caption{The LMCs used in the proof of (i) of \Cref{Proposition: Weak and Branching Epsilon-Bisimulation and Epsilon-Bisimulation are incomparable}.}
		\label{Figure: LMCs for showing that eps-bisim does not imply branching or weak}
	\end{figure}
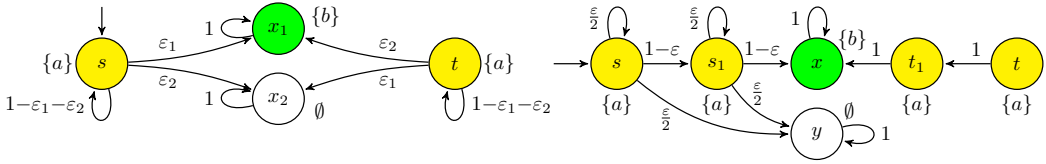
	
	\begin{restatable}{proposition}{WeakAndBanchingIncompWithEpsAPBAndBisim}\label{Proposition: Weak and Branching Epsilon-Bisimulation and Epsilon-Bisimulation are incomparable}
		Let ${\approx_\varepsilon} \in \{{\approx^b_\varepsilon}, {\approx_\varepsilon^w}\}.$ Then there are LMCs with states $s,t \in S$ such that
		\emph{(i)} $s \sim_\varepsilon t$ and $s \equiv_\varepsilon t$ but $s \not \approx_\varepsilon t$, and
		\emph{(ii)} $s \approx_\varepsilon t$ but $s \nsim_\varepsilon t$ and $s \not\equiv_\varepsilon t$. Hence, $\approx_\varepsilon$ and $\sim_\varepsilon, \equiv_\varepsilon$ are incomparable.
		Furthermore, \emph{(i)} and \emph{(ii)} also hold for $\sim_\varepsilon^*$ and $\equiv_\varepsilon^*$ instead of $\sim_\varepsilon$ and $\equiv_\varepsilon$.
	\end{restatable}
	
	\begin{proof}
		To show (i) we do a case distinction on $\approx_\varepsilon$. 
		If ${\approx_\varepsilon} = {\approx^b_\varepsilon}$, consider the LMC on the left of \Cref{Figure: LMCs for showing that eps-bisim does not imply branching or weak} where $\varepsilon_1, \varepsilon_2 \in (0,1)$, $\varepsilon_1 \neq \varepsilon_2$, $\varepsilon_1 + \varepsilon_2 < 1$, and $\varepsilon = \vert \varepsilon_1 - \varepsilon_2 \vert$. 
		In this model, both $s \sim_\varepsilon t$ and $s \equiv_\varepsilon t$. However, for any equivalence $R$ that only relates states with the same label, $\vert \mathrm{Pr}_s([s]_R \Until \{x_1\}) - \mathrm{Pr}_t([t]_R \Until \{x_1\}) \vert = \frac{\vert \varepsilon_1 - \varepsilon_2\vert }{\varepsilon_1 + \varepsilon_2} \overset{\varepsilon_1 + \varepsilon_2 < 1}{>} \vert \varepsilon_1 - \varepsilon_2 \vert = \varepsilon$, so $s \not \approx_\varepsilon^b t$.
		
		If ${\approx_\varepsilon} = {\approx^w_\varepsilon}$, consider the right of \Cref{Figure: LMCs for showing that eps-bisim does not imply branching or weak} with $\varepsilon \in (0,1)$. There, $s \sim_\varepsilon t$ and $s \equiv_\varepsilon t$. However, $\mathrm{Pr}_t(L(t) \Until \{x\}) = 1 > \frac{4(1-\varepsilon)^2}{(2-\varepsilon)^2} = \mathrm{Pr}_s(L(t) \Until \{x\})$ for all $\varepsilon \in (0,1)$, so $s \not \approx_\varepsilon^w t$.
		
		The second claim follows when considering an LMC with three states, say $s,t$ and $x$, with initial state $s$ and $l(s) = l(t) \neq l(x)$ as well as $\prob(s)(t) = \prob(t)(x) = \prob(x)(x) = 1$. There, $s \approx^b_\varepsilon t$ and $s \approx^w_\varepsilon t$ for any $\varepsilon$, but neither $s \equiv_\varepsilon t$ nor $s \sim_\varepsilon t$. 
		
		The claims are shown analogously when replacing $\sim_\varepsilon$ and $\equiv_\varepsilon$ with $\sim_\varepsilon^*$ resp. $\equiv_\varepsilon^*$.
	\end{proof}
	
	Note that the anomaly of $\equiv_\varepsilon$ described in \Cref{Example: Epsilon-APBs} does not occur for branching $\varepsilon$-bisimilarity, as here transitivity would enforce $u_i \approx_\varepsilon^b u_j$ for all $i,j$ in \Cref{Figure: LMC not all epsilon-APB are epsilon-bisim} if $s \approx_\varepsilon^b t$. 
	
	The next lemma bounds the probabilities of states related by $\approx_\varepsilon^b$ or $\approx_\varepsilon^w$ to stutter forever.
	
	\begin{restatable}{lemma}{EndlessStutteringBranchingBisim}\label{Lemma: Properties Endless Stuttering}
		Let $\varepsilon \in [0,1]$ and let $R$ be a branching $\varepsilon$-bisimulation. 
		\begin{enumerate}
			\item If $(s,t) \in R$ and $C = [s]_R = [t]_R$ then $\vert \mathrm{Pr}_s(\Box C) - \mathrm{Pr}_t(\Box C) \vert \leq \varepsilon$.
			\item If $\mathcal{M}$ is finite and $(s,t) \in R$ then, for any $C\in \xfrac{S}{R}$, either \emph{(i)} $\mathrm{Pr}_s(\Box C) = 0$ or \\\emph{(ii)} $\mathrm{Pr}_s(\Box C) \geq 1 - \varepsilon$ for all $s \in C$.
			\item If $s \approx_\varepsilon^w t$ and $b = l(s) = l(t)$ then $\vert \mathrm{Pr}_s(\Box b) - \mathrm{Pr}_t(\Box b) \vert \leq \varepsilon$.
		\end{enumerate}
		
	\end{restatable}
	
	Since $\approx_\varepsilon^b$ and $\approx_\varepsilon^w$ cannot differentiate single steps from steps after an arbitrary (but finite) amount of stuttering, they do not preserve any next-step probabilities. Furthermore, in \Cref{Figure: LMCs for maximal difference in reachability probabilities in Branching or Weak Epsilon Bisimilar States} both $s \approx_\varepsilon^b t$ and $s \approx_\varepsilon^w t$, so by \Cref{Example: Maximal difference in Reachability Probabilities in Branching or Weak Epsilon Bisimilar States} we cannot expect a better bound for finite horizon satisfaction probabilities in related states than the one from \cite{RBBTEAPC} stated in \Cref{Theorem: Main Theorem RBBTEAPC}. We can, however, extend \Cref{thm:unbounded_reach} to states related by $\approx_\varepsilon^b$ and $\approx_\varepsilon^w$. 
	
	Given an equivalence $R$ on a finite LMC $\mathcal{M}$, let $div_R \subseteq \xfrac{S}{R}$ be the set of \emph{divergent} $R$-equivalence classes, i.e., $C \in div_R$ iff $\mathrm{Pr}_s(\Box C) \geq 1 - \varepsilon$ for all $s \in C$. 
	We construct from $\mathcal{M}$ an LMC $\mathcal{M}_R$ and an equivalence $R^b$ on $\mathcal{M}_R$ with $R \subseteq R^b$. Intuitively, $\mathcal{M}_R$ is obtained from $\mathcal{M}$ by redirecting the probabilities $\mathrm{Pr}_s(\Box C)$ for $C = [s]_R$ to fresh ``divergence states'' $s_C$. 
	
	\begin{definition}\label{Definition: MR}
		Given a finite LMC $\mathcal{M}$ and an equivalence $R$ that only relates states with the same label, let $\mathcal{M}_R = (S_R, \prob_R, s_{init}, l_R)$ with 
		\begin{itemize}
			\item $S_R = S \cup \{s_C \mid C \in div_R\}$ where the $s_C$ are fresh, pairwise different states
			\item $l_R(s) = l(s)$ if $s \in S$ and $l(s_C) = l(s)$ for some $s \in C$
			\item for $s \in S$ and $C = [s]_R$, the values of the distribution $\prob_R(s)$ are defined by
			\begin{align*}
				\prob_R(s)(t) = \begin{cases}
					\mathrm{Pr}_s(C \Until t), &\text{if } t \in S \setminus C \\
					\mathrm{Pr}_s(\Box C), &\text{if } s \neq s_C \text{ and } t = s_C \\
					1, &\text{if} s = t = s_C  \\ 
					0, &\text{otherwise}
				\end{cases}.
			\end{align*}
		\end{itemize}
	\end{definition}
	
	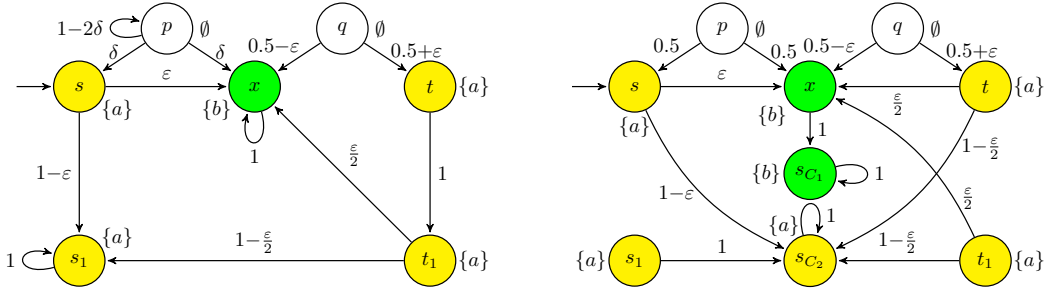
\begin{figure}[t]
		\centering
		\resizebox{!}{0.17\textheight}
		{\begin{tikzpicture}[->,>=stealth',shorten >=1pt,auto, semithick]
				\tikzstyle{every state} = [text = black]
				
				\node[state] (s) [fill = yellow] {$s$}; 
				\node (temp1) [right of = s, node distance = 1.5cm] {}; 
				\node[state] (x) [right of = temp1, node distance = 1.5cm, fill = green] {$x$}; 
				\node (temp2) [right of = x, node distance = 1.5cm] {};
				\node[state] (t) [right of = temp2, node distance = 1.5cm, fill = yellow] {$t$}; 
				\node[state] (s1) [below of = s, node distance = 3cm, fill = yellow] {$s_1$};
				\node[state] (t1) [below of = t, node distance = 3cm, fill = yellow] {$t_1$}; 
				\node[state] (p) [above of = temp1, node distance = 1cm] {$p$}; 
				\node[state] (q) [above of = temp2, node distance = 1cm] {$q$};
				\node (initl) [left of = s, node distance = 1.2cm] {};
				
				\path 
				(initl) edge (s)
				(s) edge node [left] {$1 {-} \varepsilon$} (s1)
				(s) edge node [above] {$\varepsilon$} (x)
				(s1) edge [loop left] node {$1$} (s1)
				(x) edge [loop below] node {$1$} (x)
				(t) edge node [right] {$1$} (t1)
				(t1) edge node [above] {$1{-}\frac{\varepsilon}{2}$} (s1)
				(t1) edge node [above, xshift = 0.2cm] {$\frac{\varepsilon}{2}$} (x)
				(p) edge [loop left] node [left] {$1{-}2\delta$} (p)
				(p) edge node [left, pos = 0.3, xshift = -0.1cm] {$\delta$} (s)
				(p) edge node [right, pos = 0.3, xshift = 0.1cm] {$\delta$} (x)
				(q) edge node [left, pos = 0.3, yshift = 0.1cm] {$0.5{-}\varepsilon$}(x)
				(q) edge node [right, pos =0.3, xshift = 0.1cm] {$0.5{+}\varepsilon$} (t)
				;
				
				\node [right of = s, node distance = 0.7cm, yshift = -0.4cm, xshift = -0.05cm] {$\{a\}$}; 
				\node [right of = t, node distance = 0.75cm] {$\{a\}$}; 
				\node [right of = s1, node distance = 0.7cm, yshift=0.4cm, xshift = -0.05cm] {$\{a\}$}; 
				\node [right of = t1, node distance = 0.75cm] {$\{a\}$}; 
				\node [right of = q, node distance = 0.65cm] {$\emptyset$}; 
				\node [right of = p, node distance = 0.65cm] {$\emptyset$};
				\node [left of = x, node distance = 0.7cm, yshift = -0.4cm, xshift = 0.05cm] {$\{b\}$}; 
				
				\node[state] (s) [fill = yellow, right of = t, node distance = 3.5cm] {$s$}; 
				\node (temp1) [right of = s, node distance = 1.5cm] {}; 
				\node[state] (x) [right of = temp1, node distance = 1.5cm, fill = green] {$x$}; 
				\node[state] (xdiv) [below of = x, node distance = 1.5cm, fill = green] {$s_{C_1}$};
				\node (temp2) [right of = x, node distance = 1.5cm] {};
				\node[state] (t) [right of = temp2, node distance = 1.5cm, fill = yellow] {$t$}; 
				\node[state] (s1) [below of = s, node distance = 3cm, fill = yellow] {$s_1$};
				\node[state] (t1) [below of = t, node distance = 3cm, fill = yellow] {$t_1$}; 
				\node[state] (p) [above of = temp1, node distance = 1cm] {$p$}; 
				\node[state] (q) [above of = temp2, node distance = 1cm] {$q$};
				\node[state] (sdiv) [right of = s1, node distance = 3cm, fill = yellow] {$s_{C_2}$};
				\node (initr) [left of = s, node distance = 1.2cm] {};
				
				\path 
				(initr) edge (s)
				(s) edge node [above] {$\varepsilon$} (x)
				(s) edge [bend right = 20] node [left] {$1{-}\varepsilon$} (sdiv)
				(s1) edge node [above] {$1$} (sdiv)
				(x) edge node [right] {$1$} (xdiv)
				(t) edge node [below] {$\frac{\varepsilon}{2}$} (x)
				(t) edge [bend left = 20] node [right, pos = 0.2] {$1{-}\frac{\varepsilon}{2}$} (sdiv)
				(t1) edge node [above] {$1{-}\frac{\varepsilon}{2}$} (sdiv)
				(t1) edge [bend right = 20] node [right, pos = 0.2] {$\frac{\varepsilon}{2}$} (x)
				(p) edge node [left, pos = 0.3, yshift = 0.1cm] {$0.5$} (s)
				(p) edge node [right, pos = 0.3, xshift = 0.1cm] {$0.5$} (x)
				(q) edge node [left, pos = 0.3, yshift = 0.1cm] {$0.5{-}\varepsilon$}(x)
				(q) edge node [right, pos =0.3, xshift = 0.1cm] {$0.5{+}\varepsilon$} (t)
				(sdiv) edge [loop above] node [right] [pos = 0.8] {$1$} (sdiv)
				(xdiv) edge [loop right] node {$1$} (xdiv)
				;
				
				\node [below of = s, node distance = 0.7cm] {$\{a\}$}; 
				\node [right of = t, node distance = 0.75cm] {$\{a\}$}; 
				\node [left of = s1, node distance = 0.75cm] {$\{a\}$}; 
				\node [above of = sdiv, node distance = 0.6cm, xshift = -0.45cm] {$\{a\}$}; 
				\node [right of = t1, node distance = 0.75cm] {$\{a\}$}; 
				\node [right of = q, node distance = 0.65cm] {$\emptyset$}; 
				\node [right of = p, node distance = 0.65cm] {$\emptyset$}; 
				\node [left of = x, node distance = 0.75cm, yshift = -0.5cm, xshift = 0.1cm] {$\{b\}$}; 
				\node [left of = xdiv, node distance = 0.75cm] {$\{b\}$}; 
		\end{tikzpicture}}
		\caption{An LMC $\mathcal{M}$ (left) and its transformation $\mathcal{M}_R$ (right) w.r.t. the branching $\varepsilon$-bisimulation $R$ with equivalence classes $\{s, t, s_1, t_1\}, \{p, q\}, \{x\}$ for $0 < \varepsilon < 1-2\delta$ and $div_R = \{\{s, t, s_1, t_1\}, \{x\}\}$.}
		\label{Figure: Example Transformation to MR}
	\end{figure}
	
	An example for the transformation from $\mathcal{M}$ to $\mathcal{M}_R$ can be found in \Cref{Figure: Example Transformation to MR}. 
	We now show the connection between branching $\varepsilon$-bisimulations $R$ on  finite LMCs $\mathcal{M}$ and transitive $\varepsilon$-bisimulations on their transformations $\mathcal{M}_R$. 
	\begin{restatable}{lemma}{EpsilonBisimFromBranchingBisimOnTranfsformation}\label{Lemma: Derivation of Epsilon-Bisimulation in MR}
		Let $\mathcal{M}$ be finite, $R$ an equivalence relating only equally labeled states, and $R^b$ the finest equivalence on $S_R$ with $R \subseteq R^b$ and $(s, s_C) \in R^b$ for all $C \in div_R$ and $s \in C$. Then $R$ is a branching $\varepsilon$-bisimulation on $\mathcal{M}$ iff $R^b$ is a transitive $\varepsilon$-bisimulation on $\mathcal{M}_R$.
	\end{restatable}
	
	It is clear from the definition of $\mathcal{M}_R$ that for every $C \in \xfrac{S}{R}$ and all $s \in S$ with $s \notin C$ we have $\mathrm{Pr}_s^\mathcal{M}([s]_R \Until C) = \prob_R(s)(C)$. Hence, \Cref{Lemma: Derivation of Epsilon-Bisimulation in MR} allows us to transfer \Cref{thm:unbounded_reach} to states $s \approx_\varepsilon^b t$, since they are $\varepsilon$-bisimilar in $\mathcal{M}_R$. As in $\mathcal{M}_R$ any transition from $s$ to a $u \in S$ represents an equivalence class change in $\mathcal{M}$, the random variable $N^b$ can be modified to count the number of \emph{equivalence class changes} on paths to a $g$- or $f$-labeled state. 
	
	\begin{corollary}\label{Corollary: Unbounded Reach for Branching Epsilon Bisim}
		Let $\mathcal{M}$ be finite, let some states in $\mathcal{M}$ be labeled with $g$, and let exactly the states that cannot reach a $g$-labeled state be labeled with $f$. Further, let $s \approx_\varepsilon^b t$, and let $N^b$ denote the random variable that counts the number of equivalence class changes until a $g$- or $f$-labeled state is reached. Then $\vert \mathrm{Pr}_s(\lozenge g) - \mathrm{Pr}_t(\lozenge g) \vert \leq \varepsilon \cdot \mathbb{E}_s(N^b).$
	\end{corollary}
	
	Furthermore, it is possible to extend \Cref{thm:unbounded_reach} to weakly $\varepsilon$-bisimilar states.  
	\begin{restatable}{proposition}{BoundWeakBisim}\label{Corollary: Bound for weak bisim}
		Let $\mathcal{M}$, $f$ and $g$ be as in \Cref{Corollary: Unbounded Reach for Branching Epsilon Bisim}, let $s \approx_\varepsilon^w t$ and let $N^w$ denote the random variable that counts the number of label changes until a $g$- or $f$-labeled state is reached. Then $\vert \mathrm{Pr}_s(\lozenge g) - \mathrm{Pr}_t(\lozenge g) \vert \leq \varepsilon \cdot \mathbb{E}_s(N^w).$
	\end{restatable}
	\begin{proof}[Proof sketch]
		Let $\mathcal{L} = \{b \in 2^{AP} \mid \exists \, s \in S \colon \mathrm{Pr}_s(\Box b) > 0\}$. From $\mathcal{M}$ we construct an LMC $\mathcal{M}^w$, almost similar to $\mathcal{M}_R$ in \Cref{Definition: MR}. The main differences are that we introduce fresh states $s_b$ for all $b \in \mathcal{L}$, and that we set $\prob^w(s)(t) = \mathrm{Pr}_s(L(s) \Until t)$ for all $s,t \in S$ with $l(s) \neq l(t)$ as well as $\prob^w(s)(s_b) = \mathrm{Pr}_s(\Box b)$ if $l(s) = b \in \mathcal{L}$. Because for any weak $\varepsilon$-bisimulation $R$ the finest reflexive and symmetric relation $R^w$ on $\mathcal{M}^w$ with $R \subseteq R^w$ and $(s, s_b) \in R^w$ iff $b = l(s)$ and $\mathrm{Pr}_s(\Box b) \geq 1 - \varepsilon$ is an $\varepsilon$-bisimulation on $\mathcal{M}^w$, the result follows from \Cref{thm:unbounded_reach}. 
	\end{proof}
	
	As the LMCs in \Cref{Figure: Example unbounded reachability} are both branching $\frac{\varepsilon}{p}$-bisimilar and weak $\frac{\varepsilon}{p}$-bisimilar, and since in these models $\mathbb{E}_s(N^b) = \mathbb{E}_s(N^w) = 1$, the bounds are again tight by \Cref{ex:unbounded_reach}. 
	
	We finish this section by analyzing the complexity of deciding if two given states $s,t$ are branching $\varepsilon$-bisimilar, i.e., if $s \approx_\varepsilon^b t$. The analogous problem for $\approx_\varepsilon^w$ is left open.  
	\begin{restatable}{theorem}{ThmNPCompleteBranching}\label{Theorem: Deciding Branching Bisimilarity is NP-complete}
		Given a finite $\mathcal{M}$, $s, t \in S$, and $\varepsilon \in (0,1]$, deciding if $s \approx_\varepsilon^b t$ is  \emph{\textsf{NP}}-complete.
	\end{restatable}

	
	\section{Conclusion and Future Work} \label{Section: Conclusion}
	We investigated several new types of approximate probabilistic bisimulation and showed how they interrelate, as well as how they are connected to notions from the literature like, e.g., $\sim_\varepsilon$ and $\equiv_\varepsilon$ (see \Cref{Figure: Visualization of Results}). These connections in turn allowed the transfer of known preservation results for logical formulas between the different notions, which we extended by tight bounds on the absolute difference of unbounded reachability probabilities in weak and branching $\varepsilon$-bisimilar states. Additionally, we established complexity results for most of our relations.
	
	The results of \Cref{Section: Epsilon-Quotients} indicate that $\varepsilon$-perturbed bisimilarity $\simeq_\varepsilon$ and transitive $\varepsilon$-bisimilarity $\sim_\varepsilon^*$ show some anomalies (lack of additivity, the possibility to differentiate bisimilar models and the fact that they themselves are not necessarily an $\varepsilon$-perturbed resp. a transitive $\varepsilon$-bisimulation) when viewed as process relations. However, both relations can be interesting for algorithmic purposes as they permit efficient quotienting techniques: given a transitive $\varepsilon$-bisimulation $R$ (with or without the centroid property) on an LMC $\mathcal{M}$, one can build in polynomial time a quotient LMC that arises from $\mathcal{M}$ by collapsing all $R$-equivalence classes into single states.
	The quotient under an $\varepsilon$-perturbed bisimulation $R_1$ enjoys the property that every state $s$ and its $R_1$-equivalence class $[s]_{R_1}$ are $\frac{\varepsilon}{2}$-bisimilar \cite{ABM}, while for the quotients under a transitive $\varepsilon$-bisimulation $R_2$ that lacks the centroid property we can only guarantee $s \sim_\varepsilon [s]_{R_2}$. On the other hand, transitive $\varepsilon$-bisimulations can identify more states and hence can induce smaller quotients. 
	
	Similarly, the transitivity of branching $\varepsilon$-bisimulations causes the same anomalies as for $\simeq_\varepsilon$ and $\sim_\varepsilon^*$. However, checking if a given equivalence is a branching $\varepsilon$-bisimulation and constructing a corresponding quotient is again possible in polynomial time. Hence, investigating the potential of transitive (or branching) $\varepsilon$-bisimulations as abstraction techniques for an approximate analysis of LMCs in practice is an interesting future research direction.
	
	Other open questions include the search for a characterization of logical formulas that distinguish ${\sim_\varepsilon}, {\sim_\varepsilon^*},{\approx^w_\varepsilon}, {\approx^b_\varepsilon}$ and $\simeq_\varepsilon$, and how our results relate to bisimilarity distances \cite{PBD}.
	
	\bibliography{references.bib}{}
	
	\newpage 
	\appendix
	\section{Proofs of Section 2}\label{Appendix: Proofs of Section 2}
	For the sake of completeness, we prove the characterization of approximate probabilistic bisimulation from the preliminaries which is also used in, e.g., \cite{ALMP,AAPP,RPCTLMC,RBBTEAPC}. \unskip
	\begin{lemma}
		Let $\mathcal{M}$ be a finitely branching LMC, and let $R \subseteq S \times S$ be an equivalence. Then $R$ is a probabilistic bisimulation if and only if, for all $(s,t) \in R$ and all $R$-closed subsets $A \subseteq R$, it holds that $l(s) = l(t)$ and $\prob(s)(A) = \prob(t)(A)$.
	\end{lemma}
	\begin{proof}
		For the direction from left to right, let $R \subseteq S \times S$ be a bisimulation and let $(s,t) \in R$. By the definition of bisimulations, $l(s) = l(t)$. Since $R$ is an equivalence, any $R$-closed set $A \subseteq S$ is a disjoint union of countably many equivalence classes $(C_i)_{i \in I}$ for an index set $I \subseteq \mathbb{N}$, i.e., $A = \uplus_{i \in I} C_i$ \cite{RBBTEAPC}. Furthermore, as $R$ is a bisimulation, $(s,t) \in R$ implies $\prob(s)(C) = \prob(t)(C)$ for all equivalence classes of $R$, and hence it holds that
		\begin{align*}
			\prob(s)(A) = \prob(s)(\uplus_{i \in I} C_i) = \sum_{i \in I}\prob(s)(C_i) = \sum_{i \in I} \prob(t)(C_i) = \prob(t)(\uplus_{i \in I} C_i) = \prob(t)(A). 
		\end{align*}
		
		For the reverse direction, let $R$ be an equivalence on $S$ such that for all $(s,t) \in R$ it holds that (i) $l(s) = l(t)$, and that (ii) for all $R$-closed sets $A \subseteq S$ we have $\prob(s)(A) = \prob(t)(A)$. 
		
		Further, let $C$ be an equivalence class of $R$. Then $C$ is $R$-closed, as for every $t \in R(C)$ there is a $s \in C$ with $(s,t) \in R$, and because $C$ is an equivalence class of $R$ we get $t \in C$. 
		
		Now let $(s,t) \in R$. Then $l(s) = l(t)$ by condition (i), and since any equivalence class $C$ is $R$-closed it follows from 
		condition (ii) that $\prob(s)(C) = \prob(t)(C)$ for all of them. Hence, $R$ is a probabilistic bisimulation.
	\end{proof}
	
	\section{Proofs of Section 3}\label{Appendix: Proofs of Section 3}

	\LemmaEpsilonBisim*
	
	For the proof of the lemma, we make use of the following technical statement that is shown in \cite{bollobas1975representation}.
	An atomless measure space is a measure space in which any measurable set $A$ with positive measure has a measurable subset $B$ 
	with smaller positive measure.
	
	\begin{lemma}[\cite{bollobas1975representation}]
		\label{lem:atomless}
		Let $(Z,\mathcal{F},\mu)$ be an atomless measure space. Let $J$ be an index set, $(X_{j})_{j\in J}$ be a family of measurable sets in $\mathcal{F}$ 
		of finite measure, and $(\lambda_j)_{j\in J}$ be a family of non-negative real numbers. The following two statements are equivalent:
		\begin{enumerate}
			\item
			There exists a family of measurable sets $(Y_j)_{j\in J}$ in $\mathcal{F}$ such that we have $Y_j\subseteq X_j$ and $\mu(Y_j)=\lambda_j$ for all $j\in J$ and $\mu(Y_j\cap Y_{j'})=0 $ for all $j,j'\in J$ with $j\not=j'$.
			\item
			For every finite subset $I\subseteq J$, we have $\mu\left(\bigcup_{i\in I} X_i \right) \geq \sum_{i\in I}\lambda_i$.
		\end{enumerate}
	\end{lemma}
	
	\begin{proof}[Proof of \Cref{lem:epsilon-bisimulation-distribution}]
		Let $R$ be a reflexive and symmetric relation that only relates states with the same label. We show the claim by proving both implications separately. 
		
		For the direction from left to right, assume that $R$ is an $\varepsilon$-bisimulation. We construct an atomless measure space  as in the proof of \cite[Theorem 4]{RBBTEAPC} (cf. the full version \cite{RBBTEAPCExtended} of \cite{RBBTEAPC}) on the set
		\[
		Z = \mathit{Succ}(t)\times [0,1] \cup \{\ast\}\times [0,\varepsilon],
		\]
		where $\ast \notin S$ is a fresh dummy element. 
		We let the $\sigma$-algebra $\mathcal{F}$ be the smallest $\sigma$-algebra on $Z$ containing 
		$\left( 2^{\mathit{Succ}(t)} \times \mathcal{B}([0,1]) \right) \cup  \left( \{\ast\}\times  \mathcal{B}([0,\varepsilon])  \right)$,
		where $\mathcal{B}(X)$ denotes the Borel $\sigma$-algebra on $X$.
		The measure $\mu$ is (uniquely) defined by
		\[
		\mu(D\times I) = \prob(t)(D) \cdot \lambda (I) \quad \text{ for all $D\times I \in 2^\mathit{Succ}(t) \times \mathcal{B}([0,1])$}
		\]
		and
		\[
		\mu(\{\ast\} \times I) = \lambda (I) \quad \text{ for all $ I \in  \mathcal{B}([0,\varepsilon])$},
		\]
		where $\lambda$ denotes the Lebesgue measure.
		
		In order to apply \Cref{lem:atomless}, we let 
		\[
		X_{s^\prime} = \left( R(s^\prime) \times [0,1] \right) \cup \left(  \{\ast\} \times [0,\varepsilon]  \right)
		\qquad \text{and} \qquad
		\lambda_{s^\prime} = \prob(s)(s^\prime)
		\]
		for all $s^\prime \in \mathit{Succ}(s)$.
		
		Observe that statement 2 of \Cref{lem:atomless} holds:
		For any (finite) subset $B\subseteq \mathit{Succ}(s)$, we have 
		\begin{align*}
			\mu\left(   \bigcup_{s^\prime \in B} X_{s^\prime} \right) & = \mu\left(  \left( \bigcup_{s^\prime \in B} R (s^\prime) \times [0,1] \right) \cup
			\left( \{\ast\}\times [0,\varepsilon]  \right)   \right) \\
			& = \prob(t)(R(B)) + \varepsilon \geq  \prob(s)(B) = \sum_{s^\prime \in B} \lambda_{s^\prime}
		\end{align*}
		where the inequality holds by $s\sim_{\varepsilon} t$ and the definition of $\sim_{\varepsilon}$.
		So, by \Cref{lem:atomless}, there are sets $Y_{s^\prime}\subseteq X_{s^\prime}$ for all $s^\prime \in \mathit{Succ}(s)$ 
		with $\mu(Y_{s^\prime} \cap Y_{s^{\prime\prime}}) = 0$ for all $s^\prime,s^{\prime\prime}\in \mathit{Succ}(s)$ with $s^\prime \not=  s^{\prime\prime}$ 
		and $\mu(Y_{s^\prime}) = \prob(s)(s^\prime) $ for all $s^\prime \in \mathit{Succ}(s)$.
		
		We use these sets to construct the map $\Delta$.
		First, observe that 
		\[
		\mu\left( \bigcup_{s^\prime\in \mathit{Succ}(s)} Y_{s^\prime} \right) = \sum_{s^\prime\in \mathit{Succ}(s)}  \lambda_{s^\prime} = 1.
		\]
		As $\mu(Z)=1+\varepsilon$ and $\mu( \mathit{Succ}(t) \times [0,1] )=1$, we furthermore have 
		\[
		\mu\left( \bigcup_{s^\prime\in \mathit{Succ}(s)} Y_{s^\prime} \setminus ( \mathit{Succ}(t) \times [0,1] ) \right)
		= \mu \left(  ( \mathit{Succ}(t) \times [0,1]  ) \setminus   \bigcup_{s^\prime\in \mathit{Succ}(s)} Y_{s^\prime}   \right) \leq \varepsilon.
		\tag{$\dagger$}
		\]
		
		Now, we construct sets $Y^\prime_{s^\prime}\subseteq \mathit{Succ}(t) \times [0,1]$ for all $s^\prime \in \mathit{Succ}(s)$ such that 
		\[
		Y_{s^\prime} \cap ( \mathit{Succ}(t) \times [0,1] )   \subseteq Y^\prime_{s^\prime} \tag{C1}
		\]
		and 
		\[
		\mu(Y^\prime_{s^\prime} )=\lambda_{s^\prime},
		\tag{C2}
		\] 
		and such that 
		\[
		\mu(Y^\prime_{s^\prime} \cap Y^\prime_{s^{\prime\prime}}) = 0
		\tag{C3}
		\]
		for all $s^\prime, s^{\prime\prime} \in \mathit{Succ}(s)$ with $s^\prime\not= s^{\prime\prime}$.
		To this end, we partition $ ( \mathit{Succ}(t) \times [0,1]  ) \setminus   \bigcup_{s^\prime\in \mathit{Succ}(s)} Y_{s^\prime} $
		into sets $(W_{s^\prime})_{s^\prime\in \mathit{Succ}(s)}$ with $\mu(W_{s^\prime}) = \mu(Y_{s^\prime} \setminus ( \mathit{Succ}(s) \times [0,1] ))$.
		This is possible by Equation ($\dagger$). We let
		\[
		Y^\prime_{s^\prime} = \left(Y_{s^\prime}\cap ( \mathit{Succ}(t) \times [0,1]  )  \right) \cup W_{s^\prime}
		\]
		for all $s^\prime\in\mathit{Succ}(s)$. These sets now satisfy conditions (C1)--(C3). Furthermore,
		$\bigcup_{s^\prime\in \mathit{Succ}(s)} Y^\prime_{s^\prime}$ is equal to $\mathit{Succ}(s) \times [0,1]$  up to measure $0$.
		
		Now, we define the distributions $\Delta(s^\prime)$ for all $s^\prime\in\mathit{Succ}(s)$ by
		\[
		\Delta(s^\prime) (t^\prime) = \frac{\mu(Y^{\prime}_{s^\prime} \cap (\{t^\prime\} \times [0,1]) ) }{\lambda_{s^\prime}}
		\]
		for all $t^\prime \in \mathit{Succ}(t)$.
		As $Y^\prime_{s^\prime}\subseteq \mathit{Succ}(t) \times [0,1]$ and by (C2), this indeed defines a distribution of $\mathit{Succ}(t)$.
		
		\begin{enumerate}
			\item
			First, we show that  statement 1 holds:
			For all $t^\prime \in \mathit{Succ}(t)$, we have
			\begin{align*}
				\prob(t)(t^\prime) &= \mu(\{t^\prime\} \times [0,1]) \\
				&= \mu\left( \bigcup_{s^\prime\in \mathit{Succ}(s)} Y^\prime_{s^\prime} \cap (\{t^\prime\} \times [0,1])    \right) \\
				&
				=\sum_{s^\prime\in \mathit{Succ}(s)} \mu(Y^\prime_{s^\prime} \cap (\{t^\prime\} \times [0,1])  ) \\
				&
				= \sum_{s^\prime\in \mathit{Succ}(s)} \lambda_{s^\prime}\cdot  \Delta(s^\prime) (t^\prime) \\
				&= 
				\sum_{s^\prime\in \mathit{Succ}(s)} \prob(s)(s^\prime) \cdot  \Delta(s^\prime) (t^\prime).
			\end{align*}
			
			\item
			For statement 2, 
			we have
			\begin{align*}
				& \sum_{s^\prime\in \mathit{Succ}(s)} \prob(s)(s^\prime)\cdot \Delta(s^\prime)(R(s^\prime) \cap \mathit{Succ}(t)) \\
				=& \sum_{s^\prime\in \mathit{Succ}(s)} \sum_{t^\prime \in R(s^\prime) \cap \mathit{Succ}(t)} \mu(Y^{\prime}_{s^\prime} \cap (\{t^\prime\} \times [0,1]) ) \\
				=& \sum_{s^\prime\in \mathit{Succ}(s)} \sum_{t^\prime \in R(s^\prime) \cap \mathit{Succ}(t)} \mu(Y_{s^\prime} \cap (\{t^\prime\} \times [0,1]) \cup W_{s^\prime} \cap (\{t^\prime\} \times [0,1]) ) \\
				\geq & \sum_{s^\prime\in \mathit{Succ}(s)} \sum_{t^\prime \in R(s^\prime) \cap \mathit{Succ}(t)} \mu(Y_{s^\prime} \cap (\{t^\prime\} \times [0,1]) ) \\
				\overset{Y_{s^\prime}\subseteq X_{s^\prime}}= & \sum_{s^\prime\in \mathit{Succ}(s)}  \mu(Y_{s^\prime} \cap (\mathit{Succ}(t) \times [0,1]) ) \\
				= &  \mu\left( \bigcup_{s^\prime \in \mathit{Succ}(s)} Y_{s^\prime} \cap (\mathit{Succ}(t) \times [0,1]) \right) \geq 1-\varepsilon. 
			\end{align*}
			The last inequality holds because 
			$
			\mu\left( \bigcup_{s^\prime \in \mathit{Succ}(s)} Y_{s^\prime}\right) = 1$
			and because   Equation ($\dagger$) states that
			\[
			\mu\left( \bigcup_{s^\prime\in \mathit{Succ}(s)} Y_{s^\prime} \setminus ( \mathit{Succ}(t) \times [0,1] ) \right) \leq \varepsilon. 
			\]
			\end{enumerate}
			
			We now turn our attention to the implication from right to left. To this end, assume that for every $(s,t) \in R$ there is some map $\Delta \colon Succ(s) \to Distr(Succ(t))$ such that statements 1 and 2 hold. Because we assume $R$ to be reflexive, symmetric and only relate states with the same label, the only thing left to show for $R$ to be an $\varepsilon$-bisimulation is condition (ii) in \Cref{Definition: Epsilon-Bisimulation}. Let $A \subseteq S$ be arbitrary. By statement 2 we have  
			\begin{align}
				1 - \varepsilon &\leq \sum_{s' \in Succ(s)}\prob(s)(s') \cdot \Delta (s') (R(s') \cap Succ(t)) \nonumber
				\\&= \sum_{s' \in Succ(s) \cap A} \prob(s)(s') \cdot \Delta (s') (R(s') \cap Succ(t)) \nonumber
				\\& \phantom{+++} + \sum_{s ' \in Succ(s) \setminus A} \prob(s)(s') \cdot \Delta (s') (R(s') \cap Succ(t)). \label{Eq: Char-1}
			\end{align}
			As $R(s') \subseteq R(A)$ for every $s' \in A$, the first sum is bounded from above by
			\begin{align} 
				&\sum_{s' \in Succ(s) \cap A} \prob(s)(s') \cdot \Delta (s') (R(s') \cap Succ(t)) \nonumber \\ &\leq \sum_{s' \in Succ(s) \cap A} \prob(s)(s') \cdot \Delta (s') (R(A) \cap Succ(t)), \label{Eq: Char-2}
			\end{align}
			while in the second sum we can replace $\Delta(s')(R(s') \cap Succ(t))$ by the maximal probability of $1$ to obtain 
			\begin{align}
				 \sum_{s ' \in Succ(s) \setminus A} \prob(s)(s') \cdot \Delta (s') (R(s') \cap Succ(t)) &\leq \sum_{s' \in Succ(s) \setminus A} \prob(s)(s') \nonumber \\&= \prob(s)(S \setminus A) \nonumber\\&= 1 - \prob(s)(A). \label{Eq: Char-3}
			\end{align}
			Plugging \Cref{Eq: Char-2,Eq: Char-3} into \Cref{Eq: Char-1} now yields 
			\begin{align*}
				1 - \varepsilon \leq \sum_{s' \in Succ(s) \cap A} \prob(s)(s') \cdot \Delta (s') (R(A) \cap Succ(t)) + 1 - \prob(s)(A),
			\end{align*}
			or equivalently
			\begin{align*}
				\prob(s)(A) &\leq \sum_{s' \in Succ(s) \cap A} \prob(s)(s') \cdot \Delta (s') (R(A) \cap Succ(t)) + \varepsilon \\
				&= \sum_{t' \in R(A) \cap Succ(t)} \sum_{s' \in Succ(s) \cap A} \prob(s)(s') \cdot \Delta(s')(t') + \varepsilon \\
				&\leq \sum_{t' \in R(A) \cap Succ(t)} \sum_{s' \in Succ(s)} \prob(s)(s') \cdot \Delta(s')(t') + \varepsilon
				\\&= \sum_{t' \in R(A) \cap Succ(t)} \prob(t)(t') + \varepsilon 
				\\&= \prob(t)(R(A)) + \varepsilon,
			\end{align*}
			where the second to last equality follows from statement 1. Because $A \subseteq S$ was chosen arbitrarily, this shows that for any $(s,t) \in R$ and all $A \subseteq S$ we have $\prob(s)(A) \leq \prob(t)(R(A)) + \varepsilon$. Hence, $R$ satisfies condition (ii) of \Cref{Definition: Epsilon-Bisimulation}, and it follows that $R$ is indeed an $\varepsilon$-bisimulation.
	\end{proof}
	
	\TheoremRelationshipEpsilonAPBAndUpToBisimilarity*
	\begin{proof}
		Let $s,t \in S$ be such that $s \equiv_\varepsilon t$. 
		For $n = 0$ the claim is trivial, as $ {\sim_{\varepsilon}^0} = S \times S$. 
		
		If $n = 1$ then ${\sim_{\varepsilon}^1} = \{(x,y) \in S \times S \mid l(x) = l(y)\}$. Since $\equiv_\varepsilon$ only relates states with the same label, it follows that for all $s,t$ with $s \equiv_\varepsilon t$ we have $s \sim_{\varepsilon}^1 t$. 
		
		Now let $n = 2$ and $A \subseteq S$. Then
		\begin{align*}
			{\sim_\varepsilon^1}(A) = \{y \in S \mid \exists \, x \in A \colon x \sim_\varepsilon^1 y\} = \{y \in S \mid \exists \, x \in A\colon l(x) = l(y)\}.
		\end{align*} 
		We start by showing that this set is $\equiv_\varepsilon$-closed. To see this, first observe that
		\begin{align*}
			{\equiv_\varepsilon}({\sim_\varepsilon^1}(A)) &= \{q \in S \mid \exists \, p \in {\sim_\varepsilon^1}(A)\colon p \equiv_\varepsilon q\} \\&= \{q \in S \mid \exists \, p \in S \colon \exists \, x \in A\colon l(x) = l(p) \land p \equiv_\varepsilon q \}.
		\end{align*}
		Hence, for $y \in {\equiv_\varepsilon}({\sim_\varepsilon^1}(A))$, there is a $p \in S$ and a $x \in A$ such that $l(x) = l(p)$ and $p \equiv_\varepsilon y$. But, as $\equiv_\varepsilon$ only relates states with the same label, this implies $l(x) = l(p) = l(y)$. Thus, there is an element in $A$, namely $x$, such that $x \sim_\varepsilon^1 y$, which shows that $y \in {\sim_\varepsilon^1}(A)$. By choice of $y$ this further implies ${\equiv_\varepsilon}({\sim_\varepsilon^1}(A)) \subseteq {\sim_\varepsilon^1}(A)$, so the set is indeed $\equiv_\varepsilon$-closed.
		
		It follows from $s \equiv_\varepsilon t$ that $	\vert \prob(s)({\sim_\varepsilon^1}(A)) - \prob(t)({\sim_\varepsilon^1}(A)) \vert \leq \varepsilon$
		which, using the reflexivity of $\sim_\varepsilon^1$ \cite{SAAPBPCTL} to establish the first inequality, yields
		\begin{align*}
			\prob(s)(A) \leq \prob(s)({\sim_\varepsilon^1}(A)) \leq \prob(t)({\sim_\varepsilon^1}(A)) + \varepsilon.
		\end{align*}
		With a similar argument, it can be shown that $\prob(t)(A) \leq \prob(s)({\sim_\varepsilon^1}(A)) + \varepsilon$. As $l(s) = l(t)$ by the definition of $\varepsilon$-APBs, this proves $s \sim_\varepsilon^2 t$.
		
		Regarding $n \geq 3$, we know from \Cref{Example: Epsilon-APBs} that the states $s, t$ in \Cref{Figure: LMC not all epsilon-APB are epsilon-bisim} satisfy $s \equiv_\varepsilon t$, but it is easy to see that $s \nsim_\varepsilon^n t$ for any $n \geq 3$.
	\end{proof}
	
	\TheoremUnboundedReach*
	\begin{proof}
		Let $\mathcal{M}$ and $N$ be as in the theorem, and let $s \sim_\varepsilon t$. 
		The result follows from the following auxiliary claim:  
		For all $k\in \mathbb{N}$, we have
		\begin{align}
			\vert \mathrm{Pr}_s(\lozenge^{\leq k} g) - \mathrm{Pr}_t(\lozenge^{\leq k} g) \vert \leq \varepsilon \cdot \sum_{i = 1}^k \mathrm{Pr}_s(N \geq i). \label{Eq: Auxiliary Claim}
		\end{align}
	
		\paragraph*{Proof of \Cref{Eq: Auxiliary Claim}:}
		\Cref{Eq: Auxiliary Claim} is shown by proving the two inequalities 
		\begin{align*}
			\mathrm{Pr}_t (\lozenge^{\leq k}  g) \geq \mathrm{Pr}_s (\lozenge^{\leq k}  g) - \sum_{i=1}^k \mathrm{Pr}_s(N\geq i)\cdot \varepsilon 
		\end{align*}
		and 
		\begin{align*}
			\mathrm{Pr}_t (\lozenge^{\leq k}  g) \leq \mathrm{Pr}_s (\lozenge^{\leq k}  g) + \sum_{i=1}^k \mathrm{Pr}_s(N\geq i)\cdot \varepsilon
		\end{align*}
		separately.
		\paragraph*{Lower bound:}
		First, we prove that for all $k \in \mathbb{N}$
		\[
		\mathrm{Pr}_{t} (\lozenge^{\leq k}  g) \geq \mathrm{Pr}_{s} (\lozenge^{\leq k}  g) - \sum_{i=1}^k \mathrm{Pr}_{s}(N\geq i)\cdot \varepsilon
		\]
		by induction on $k$.
		
		\textbf{Base case:}
		For $k=0$, there is nothing to prove as $\mathrm{Pr}_t (\lozenge^{\leq k}  g) = \mathrm{Pr}_s (\lozenge^{\leq k}  g)\in\{0,1\}$ in this case.
		
		\textbf{Induction hypothesis:}
		For the induction, we now suppose that for all $\varepsilon$-bisimilar states $s^\prime$ and $t^\prime$ in $\mathcal{M}$, we have
		\[
		\mathrm{Pr}_{t^\prime} (\lozenge^{\leq k}  g) \geq \mathrm{Pr}_{s^\prime} (\lozenge^{\leq k}  g) - \sum_{i=1}^k \mathrm{Pr}_{s^\prime}(N\geq i)\cdot \varepsilon. \tag{IH}
		\]
		
		\textbf{Induction step:}
		If $s$, and hence $t$, is labeled with $g$ or $f$, the statement is trivial. So, we assume that $s$ and $t$ are not labeled with $g$ or $f$.
		Let $\Delta\colon \mathit{Succ}(s) \to \mathit{Distr}(\mathit{Succ}(t))$ be as in \Cref{lem:epsilon-bisimulation-distribution}.
		We obtain
		\begin{align*}
			&\mathrm{Pr}_{t} (\lozenge^{\leq k+1}g) \\
			&= \sum_{t^\prime \in \mathit{Succ}(t)} \prob(t)(t^\prime) \cdot \mathrm{Pr}_{t^\prime} (\lozenge^{\leq k}g) \\
			&= \sum_{t^\prime \in  \mathit{Succ}(t)} \sum_{s^\prime\in \mathit{Succ}(s)}\prob(s)(s^\prime)\cdot\Delta(s^\prime)(t^\prime)
			\cdot \mathrm{Pr}_{t^\prime} (\lozenge^{\leq k}g) \\
			&= \sum_{s^\prime\in \mathit{Succ}(s)}\prob(s)(s^\prime) \cdot \sum_{t^\prime \in  \mathit{Succ}(t)} \Delta(s^\prime)(t^\prime)
			\cdot \mathrm{Pr}_{t^\prime} (\lozenge^{\leq k}g) \\
			&\geq \sum_{s^\prime\in \mathit{Succ}(s)}\prob(s)(s^\prime) \cdot \sum_{\substack{t^\prime \in  \mathit{Succ}(t),\\ s^\prime \sim_{\varepsilon} t^\prime}} \Delta(s^\prime)(t^\prime)
			\cdot \mathrm{Pr}_{t^\prime} (\lozenge^{\leq k}g) \\
			& \overset{(\text{IH})}{\geq} \sum_{s^\prime\in \mathit{Succ}(s)}\prob(s)(s^\prime) \cdot \sum_{\substack{t^\prime \in  \mathit{Succ}(t),\\ s^\prime \sim_{\varepsilon} t^\prime}} \Delta(s^\prime)(t^\prime)
			\cdot \left(\mathrm{Pr}_{s^\prime} (\lozenge^{\leq k}  g) - \sum_{i=1}^k \mathrm{Pr}_{s^\prime}(N\geq i)\cdot \varepsilon \right) \\
			& =  \sum_{s^\prime\in \mathit{Succ}(s)}\prob(s)(s^\prime)
			\cdot \left(\mathrm{Pr}_{s^\prime} (\lozenge^{\leq k}  g) - \sum_{i=1}^k \mathrm{Pr}_{s^\prime}(N\geq i)\cdot \varepsilon \right) \\
			& \phantom{=} -
			\sum_{s^\prime\in \mathit{Succ}(s)}\prob(s)(s^\prime) \cdot \left(1 - \sum_{\substack{t^\prime \in  \mathit{Succ}(t),\\ s^\prime \sim_{\varepsilon} t^\prime}} \Delta(s^\prime)(t^\prime) \right)
			\cdot \underbrace{\left(\mathrm{Pr}_{s^\prime} (\lozenge^{\leq k}  g) - \sum_{i=1}^k \mathrm{Pr}_{s^\prime}(N\geq i)\cdot \varepsilon \right)}_{\text{($=:c_{s^\prime}$)}} \\
			&\overset{c_{s^\prime}\leq 1\text{ for all $s^\prime$}}{ \geq}
			\sum_{s^\prime\in \mathit{Succ}(s)}\prob(s)(s^\prime)
			\cdot \left(\mathrm{Pr}_{s^\prime} (\lozenge^{\leq k}  g) - \sum_{i=1}^k \mathrm{Pr}_{s^\prime}(N\geq i)\cdot \varepsilon \right) \\
			& \phantom{======} -
			\sum_{s^\prime\in \mathit{Succ}(s)}\prob(s)(s^\prime) \cdot \left(1 - \sum_{\substack{t^\prime \in  \mathit{Succ}(t),\\ s^\prime \sim_{\varepsilon} t^\prime}} \Delta(s^\prime)(t^\prime) \right) \\
			& =
			\sum_{s^\prime\in \mathit{Succ}(s)}\prob(s)(s^\prime)
			\cdot \left(\mathrm{Pr}_{s^\prime} (\lozenge^{\leq k}  g) - \sum_{i=1}^k \mathrm{Pr}_{s^\prime}(N\geq i)\cdot \varepsilon \right) \\
			& \phantom{=} - 1 +
			\sum_{s^\prime\in \mathit{Succ}(s)}\prob(s)(s^\prime) \cdot  \sum_{\substack{t^\prime \in  \mathit{Succ}(t),\\ s^\prime \sim_{\varepsilon} t^\prime}} \Delta(s^\prime)(t^\prime)  \\
			& \overset{\text{Property of $\Delta$}}{\geq }
			\left( \sum_{s^\prime\in \mathit{Succ}(s)}\prob(s)(s^\prime)  
			\cdot \left(\mathrm{Pr}_{s^\prime} (\lozenge^{\leq k}  g) - \sum_{i=1}^k \mathrm{Pr}_{s^\prime}(N\geq i)\cdot \varepsilon \right) \right) - 1+ (1- \varepsilon) \\
			&=\mathrm{Pr}_{s} (\lozenge^{\leq k+1}  g) -  \left(\sum_{i=2}^{k+1} \mathrm{Pr}_{s}(N\geq i)\cdot \varepsilon  \right)
			- \varepsilon \\
			&= \mathrm{Pr}_{s} (\lozenge^{\leq k+1}  g) -  \left(\sum_{i=2}^{k+1} \mathrm{Pr}_{s}(N\geq i)\cdot \varepsilon  \right)
			- \mathrm{Pr}_{s}(N\geq 1)\cdot\varepsilon   \\
			&=\mathrm{Pr}_{s} (\lozenge^{\leq k+1}  g)-  \sum_{i=1}^{k+1} \mathrm{Pr}_{s}(N\geq i)\cdot \varepsilon. 
		\end{align*}
		
		\paragraph*{Upper bound:}
		Now, we prove that for all $k$
		\[
		\mathrm{Pr}_{t} (\lozenge^{\leq k}  g) \leq \mathrm{Pr}_{s} (\lozenge^{\leq k}  g) + \sum_{i=1}^k \mathrm{Pr}_{s}(N\geq i)\cdot \varepsilon
		\]
		by induction on $k$. We essentially repeat the same induction with minor modifications in the induction step.
		
		\textbf{Base case:}
		Again, for $k=0$, there is nothing to prove as $\mathrm{Pr}_t (\lozenge^{\leq 0}  g) = \mathrm{Pr}_s (\lozenge^{\leq 0}  g)\in\{0,1\}$.
		
		\textbf{Induction hypothesis:}
		For the induction, we now suppose that for all $\varepsilon$-bisimilar states $s^\prime$ and $t^\prime$ in $\mathcal{M}$, we have
		\[
		\mathrm{Pr}_{t^\prime} (\lozenge^{\leq k}  g) \leq \mathrm{Pr}_{s^\prime} (\lozenge^{\leq k}  g) + \sum_{i=1}^k \mathrm{Pr}_{s^\prime}(N\geq i)\cdot \varepsilon. \tag{IH}
		\]
		
		\textbf{Induction step:}
		If $s$, and hence $t$, is labeled with $g$ or $f$, the statement is trivial. So, we assume that $s$ and $t$ are not labeled with $g$ or $f$.
		Let $\Delta\colon \mathit{Succ}(s) \to \mathit{Distr}(\mathit{Succ}(t))$ be as in \Cref{lem:epsilon-bisimulation-distribution}.
		We obtain
		\begin{align*}
			&\mathrm{Pr}_{t} (\lozenge^{\leq k+1}g) \\
			&= \sum_{t^\prime \in \mathit{Succ}(t)} \prob(t)(t^\prime) \cdot \mathrm{Pr}_{t^\prime} (\lozenge^{\leq k}g) \\
			&= \sum_{t^\prime \in  \mathit{Succ}(t)} \sum_{s^\prime\in \mathit{Succ}(s)}\prob(s)(s^\prime)\cdot\Delta(s^\prime)(t^\prime)
			\cdot \mathrm{Pr}_{t^\prime} (\lozenge^{\leq k}g) \\
			&= \sum_{s^\prime\in \mathit{Succ}(s)}\prob(s)(s^\prime) \cdot \sum_{t^\prime \in  \mathit{Succ}(t)} \Delta(s^\prime)(t^\prime)
			\cdot \mathrm{Pr}_{t^\prime} (\lozenge^{\leq k}g) \\
			&\leq  \left(\sum_{s^\prime\in \mathit{Succ}(s)}\prob(s)(s^\prime) \cdot \sum_{\substack{t^\prime \in  \mathit{Succ}(t),\\ s^\prime \sim_{\varepsilon} t^\prime}} \Delta(s^\prime)(t^\prime)
			\cdot \mathrm{Pr}_{t^\prime} (\lozenge^{\leq k}g) \right)  \\
			& \phantom{===}  + \sum_{s^\prime\in \mathit{Succ}(s)}\prob(s)(s^\prime) \cdot \sum_{\substack{t^\prime \in  \mathit{Succ}(t),\\ s^\prime \nsim_{\varepsilon} t^\prime}} \Delta(s^\prime)(t^\prime) \cdot 1 \\
			&\overset{\text{Property of $\Delta$}}{\leq  } \left(\sum_{s^\prime\in \mathit{Succ}(s)}\prob(s)(s^\prime) \cdot \sum_{\substack{t^\prime \in  \mathit{Succ}(t),\\ s^\prime \sim_{\varepsilon} t^\prime}} \Delta(s^\prime)(t^\prime)
			\cdot \mathrm{Pr}_{t^\prime} (\lozenge^{\leq k}g) \right) +\varepsilon \\
			& \overset{(\text{IH})}{\leq} \left(\sum_{s^\prime\in \mathit{Succ}(s)}\prob(s)(s^\prime) \cdot \sum_{\substack{t^\prime \in  \mathit{Succ}(t),\\ s^\prime \sim_{\varepsilon} t^\prime}} \Delta(s^\prime)(t^\prime)
			\cdot \left(\mathrm{Pr}_{s^\prime} (\lozenge^{\leq k}  g) + \sum_{i=1}^k \mathrm{Pr}_{s^\prime}(N\geq i)\cdot \varepsilon \right) \right)+\varepsilon \\
			& \leq \left(\sum_{s^\prime\in \mathit{Succ}(s)}\prob(s)(s^\prime) \cdot \left(\mathrm{Pr}_{s^\prime} (\lozenge^{\leq k}  g) + \sum_{i=1}^k \mathrm{Pr}_{s^\prime}(N\geq i)\cdot \varepsilon \right) \right)+\varepsilon \\
			& = \left(\mathrm{Pr}_{s} (\lozenge^{\leq k+1}  g) + \sum_{i=2}^{k+1} \mathrm{Pr}_{s}(N\geq i)\cdot \varepsilon \right) +  \mathrm{Pr}_{s}(N\geq 1)\cdot\varepsilon \\
			&=\mathrm{Pr}_{s} (\lozenge^{\leq k+1}  g) + \sum_{i=1}^{k+1} \mathrm{Pr}_{s}(N\geq i)\cdot \varepsilon. 
		\end{align*}
		
		\paragraph*{Proof of \Cref{thm:unbounded_reach}:}
		As $\mathbb{E}(X) = \sum_{i=1}^\infty \mathrm{Pr}(X\geq i)$ for any random variable $X$ with values in $\mathbb{N}$ it follows that $\sum_{i=1}^\infty \mathrm{Pr}_s(N\geq i)  = \mathbb{E}_s(N).$ Hence, \Cref{thm:unbounded_reach} follows when taking the limit $k \to \infty$ in \Cref{Eq: Auxiliary Claim}.
	\end{proof}
	
	\section{Proofs of Section 4}\label{Appendix: Proofs of Section 4}
	\LemmaDifferenceOfEpsilonPerturbationProbabilities*
	\begin{proof}
		Let $s \in S$ and $A \subseteq S$. We show the claim by proving that
		\begin{align*}
			\prob(s)(A) \leq \prob'(s')(A') + \frac{\varepsilon}{2} \qquad \text{ and } \qquad \prob'(s')(A') \leq \prob(s)(A) + \frac{\varepsilon}{2}, 
		\end{align*}
		from which the lemma directly follows. As both inequalities can be shown similarly, we only focus on the first one to avoid repetition. 
		
		The claim is proved by contraposition. Hence, towards a contradiction, assume that $\prob(s)(A) > \prob'(s')(A') + \frac{\varepsilon}{2}$ or, equivalently, that $\prob(s)(A) - \prob'(s')(A') > \frac{\varepsilon}{2}$. First, note that
		\begin{align*}
			\Vert \prob(s) - \prob'(s') \Vert_1 &= \sum_{t \in A} \vert \prob(s)(t) - \prob'(s')(t') \vert + \sum_{t \in S \setminus A} \vert \prob(s)(t) - \prob'(s')(t') \vert.
		\end{align*}
		For the first summand, the assumption $\prob(s)(A) - \prob'(s')(A') > \frac{\varepsilon}{2}$ implies 
		\begin{align}
			\sum_{t \in A} \vert \prob(s)(t) - \prob'(s')(t') \vert \geq \sum_{t \in A} (\prob(s)(t) - \prob'(s')(t')) = \prob(s)(A) - \prob'(s')(A') > \frac{\varepsilon}{2},\label{Proof: Lemma - Difference of Epsilon-Perturbation Probabilities - S2}
		\end{align}
		while we can apply the (reversed) triangle inequality to the second term to get
		\begin{align*}
			\sum_{t \in S \setminus A} \vert \prob(s)(t) - \prob'(s')(t') \vert &\geq \big \vert \sum_{t \in S \setminus A} \left(\prob(s)(t) - \prob'(s')(t')  \right) \big \vert \\&= \vert \prob(s)(S \setminus A) - \prob'(s')(S' \setminus A')\vert.
		\end{align*}
		Since $\prob(s)(S) = \prob'(s')(S') = 1$, the assumption $\prob(s)(A) > \prob'(s')(A') + \frac{\varepsilon}{2}$ entails
		\begin{align*}
			\prob(s)(S \setminus A) &= \prob(s)(S) - \prob(s)(A) =
			\prob'(s')(S') - \prob(s)(A) \\
			&< \prob'(s') (S') - \prob'(s')(A') - \frac{\varepsilon}{2} 
			= \prob'(s')(S' \setminus A') - \frac{\varepsilon}{2}, 
		\end{align*}
		which is equivalent to 
		\begin{align*}
			\prob(s)(S \setminus A) - \prob'(s')(S' \setminus A') < - \frac{\varepsilon}{2}.
		\end{align*}
		Thus, $\vert \prob(s)(S \setminus A) - \prob'(s')(S' \setminus A')\vert > \frac{\varepsilon}{2}$, and together with \Cref{Proof: Lemma - Difference of Epsilon-Perturbation Probabilities - S2} we get
		\begin{align*}
			\Vert \prob(s) - \prob'(s') \Vert_1 > \frac{\varepsilon}{2} + \vert \prob(s)(S \setminus A) - \prob'(s')(S' \setminus A')\vert > \frac{\varepsilon}{2} + \frac{\varepsilon}{2} = \varepsilon.
		\end{align*}
		This contradicts $\mathcal{M}'$ being an $\varepsilon$-perturbation of $\mathcal{M}$. 
	\end{proof}

	\begin{lemma}\label{Lemma: Epsilon-Bisim = Epsilon-APB in equivalence case}
		Let $R$ be an equivalence that only relates equally labeled states. Then $R$ is an $\varepsilon$-APB iff it is an $\varepsilon$-bisimulation.
	\end{lemma}
	\begin{proof}
		It follows from \cite{RBBTEAPC} that every $\varepsilon$-bisimulation is an $\varepsilon$-APB, and the same holds in particular if the relation in question is transitive. Now, let $R$ be a transitive $\varepsilon$-APB, and let $(s,t) \in R$. Then for all $A \subseteq S$ we have that (i) $R(A)$ is $R$-closed and (ii) $A \subseteq R(A)$. Thus,
		\begin{align*}
			\prob(s)(A) \overset{\text{(ii)}}{\leq} \prob(s)(R(A)) \overset{\text{(i)} + (s,t) \in R}{\leq} \prob(t)(R(A)) + \varepsilon.
		\end{align*}
		Similarly, $\prob(t)(A) \leq \prob(s)(R(A)) + \varepsilon$. Hence, $R$ is a transitive $\varepsilon$-bisimulation.
	\end{proof}

	\PropAnomalySimeq*
	\begin{proof}
		The second claim follows as described before the proposition. Regarding the non-additivity, assume the opposite. Then $\mathcal{M}_1 \sim \mathcal{M}_2$ and $\mathcal{M}_1 \simeq_\varepsilon \mathcal{N}$ would always yield $\mathcal{M}_2 \simeq_\varepsilon \mathcal{N}$ since ${\sim} = {\simeq_0}$, which is a contradiction to the possibility of $\simeq_\varepsilon$ to distinguish bisimilar states. Hence, $\simeq_\varepsilon$ cannot be additive.
	\end{proof}

	\PerturbedEpsilonBisimImpliesEpsilonBisim*
	\begin{proof}
		As $\mathcal{M} \simeq_\varepsilon \mathcal{N}$ there are $\varepsilon$-perturbations $\mathcal{M}'$ and $\mathcal{N}'$ of $\mathcal{M}$ and $\mathcal{N}$, respectively, such that $\mathcal{M}' \sim \mathcal{N}'$.
		
		Denote by $\sim'$ the bisimilarity relation on $\mathcal{M}' \oplus \mathcal{N}'$. We now  show that $\sim'$ is a transitive $\varepsilon$-bisimulation on $\mathcal{M} \oplus \mathcal{N}$. To this end, denote by $S$ and $S'$ the state spaces and by $\prob$ and $\prob'$ the transition distribution functions of $\mathcal{M} \oplus \mathcal{N}$ and $\mathcal{M}' \oplus \mathcal{N}'$, respectively.
		
		Let $s,t \in S$ with $s' \sim' t'$. Further, let $C \subseteq S'$ be a $\sim'$-closed set. Then, as $\sim'$ is an equivalence, $C$ is of the form $C = C_1' \uplus \dots \uplus C_n'$ for disjoint equivalence classes $C_1', \dots, C_n'$ of $\sim'$. But then, as $s' \sim' t'$,
		\begin{align*}
			\vert \prob(s)(C) - \prob(t)(C) \vert &= \vert \prob(s)(C) - \prob'(s')(C') + \prob'(s')(C') - \prob(t)(C) \vert \\
			&= \vert \prob(s)(C) - \prob'(s')(C') + \prob'(t')(C') - \prob(t)(C) \vert \\
			&\leq \vert \prob(s)(C) - \prob'(s')(C') \vert + \vert \prob(t)(C) - \prob'(t')(C') \vert \\
			&\leq \frac{\varepsilon}{2} + \frac{\varepsilon}{2} = \varepsilon,
		\end{align*}
		where the last inequality follows from \Cref{Lemma: Difference of Epsilon-Perturbation Probabilities} and the fact that both $\mathcal{M}'$ and $\mathcal{N}'$ are $\varepsilon$-perturbations of $\mathcal{M}$ and $\mathcal{N}$, respectively. Hence, $\sim'$ is an $\varepsilon$-APB on $\mathcal{M} \oplus \mathcal{N}$ and, by \Cref{Lemma: Epsilon-Bisim = Epsilon-APB in equivalence case}, it is also an $\varepsilon$-bisimulation. Further, as ${s_{init}^\mathcal{M}}' \sim' {s_{init}^\mathcal{N}}'$ by assumption, it follows that $\mathcal{M} \sim_\varepsilon^* \mathcal{N}$.
	\end{proof}
	
	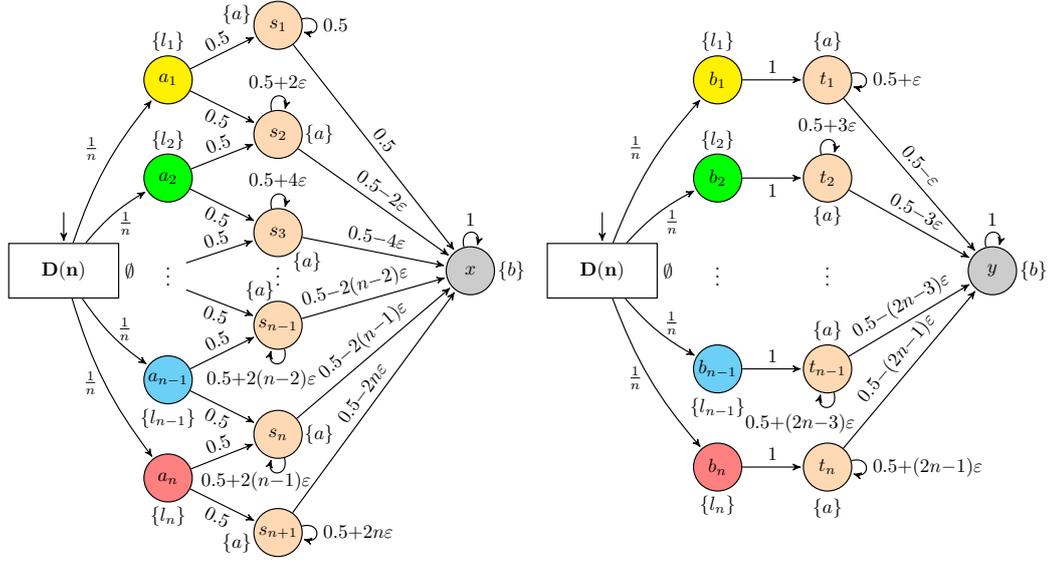
\begin{figure}[t!]
		\centering
		\resizebox{!}{0.335\textheight}{
			\begin{tikzpicture}[->,>=stealth',shorten >=1pt,auto, semithick]
				\tikzstyle{every state} = [text = black]
				\node[state] (a1) [fill = yellow] {$a_1$};
				\node[state] (a2) [fill = green, below of = a1, node distance = 1.8cm] {$a_2$};
				\node[] (lhsdots) [below of = a2, node distance = 1.7cm] {\vdots};
				\node[] (lhsdotstemp1) [above of = lhsdots, node distance = 0.1cm, xshift = 0.2cm] {};
				\node[] (lhsdotstemp2) [below of = lhsdots, node distance = 0.3cm, xshift = 0.2cm] {};
				\node[state] (an-1) [fill = cyan!50, below of = lhsdots, node distance = 2cm, inner sep = 0cm] {$a_{n-1}$};
				\node[state] (an) [fill = red!50, below of = an-1, node distance = 1.8cm] {$a_n$};
				
				\node[] (a1temp) [right of = a1, node distance = 2cm] {};
				\node[] (a2temp) [right of = a2, node distance = 2cm] {};
				\node[] (lhsdotstemp) [right of = lhsdots, node distance = 2cm] {\vdots};
				\node[] (an-1temp) [right of = an-1, node distance = 2cm] {};
				\node[] (antemp) [right of = an, node distance = 2cm] {};
				
				\node[state] (s1) [above of = a1temp, node distance = 1cm, fill = orange!30] {$s_1$};
				\node[state] (s2) [below of = a1temp, node distance = 1cm, fill = orange!30] {$s_2$};
				\node[state] (s3) [below of = a2temp, node distance = 1cm, fill = orange!30] {$s_3$};
				\node[state] (sn-1) [above of = an-1temp, node distance = 1cm, inner sep = 0cm, fill = orange!30] {$s_{n-1}$};
				\node[state] (sn) [below of = an-1temp, node distance = 1cm, fill = orange!30] {$s_{n}$};
				\node[state] (sn+1) [below of = antemp, node distance = 1cm, inner sep = 0cm, fill = orange!30] {$s_{n+1}$};
				
				\node[state] (x) [fill = gray!40, right of = lhsdotstemp, node distance = 3.5cm] {$x$};
				
				\node[state, draw, rectangle] (s) [left of = lhsdots, node distance = 1.9cm, minimum width = 2cm, minimum height = 1cm] {$\mathbf{D(n)}$};
				\node (initl) [above of = s, node distance = 1.2cm] {};
				
				\path 
				
				(initl) edge (s)
				(s) edge [bend left = 10] node {$\frac{1}{n}$} (a1)
				(s) edge [bend left = 10] node [below, pos = 0.7] {$\frac{1}{n}$} (a2)
				(s) edge [bend right = 10] node [above, pos = 0.7] {$\frac{1}{n}$} (an-1)
				(s) edge [bend right = 10] node  [left]{$\frac{1}{n}$} (an)
				(a1) edge node [above,sloped] {$0.5$} (s1)
				(a1) edge node [below,sloped] {$0.5$} (s2)
				(a2) edge node [above,sloped] {$0.5$} (s2)
				(a2) edge node [below,sloped] {$0.5$} (s3)
				(an-1) edge node [above,sloped] {$0.5$} (sn-1)
				(an-1) edge node [below,sloped] {$0.5$} (sn)
				(an) edge node [above,sloped] {$0.5$} (sn)
				(an) edge node [below,sloped] {$0.5$} (sn+1)
				(lhsdotstemp1) edge node [above,sloped] {$0.5$} (s3)
				(lhsdotstemp2) edge node [below,sloped] {$0.5$} (sn-1)
				
				(s1) edge node [above,sloped] {$0.5$} (x)
				(s2) edge node [above,sloped] {$0.5 {-} 2\varepsilon$} (x)
				(s3) edge node [above,sloped] {$0.5 {-} 4\varepsilon$} (x)
				(sn-1) edge node [pos=0.4,above,sloped] {$0.5 {-} 2(n{-}2)\varepsilon$} (x)
				(sn) edge node [above,sloped] {$0.5 {-} 2(n{-}1)\varepsilon$} (x)
				(sn+1) edge node [above,sloped] {$0.5 {-} 2n\varepsilon$} (x)
				
				(s1) edge [loop right, looseness = 4] node {$0.5$} (s1)
				(s2) edge [loop above, looseness = 4] node {$0.5 {+} 2\varepsilon$}  (s2)
				(s3) edge [loop above, looseness = 4] node {$0.5 {+} 4\varepsilon$}  (s3)
				
				(sn-1) edge [loop below, looseness = 4] node [xshift = -0.3cm] {$0.5 {+} 2(n{-}2)\varepsilon$} (sn-1)
				(sn) edge [loop below, looseness = 4] node [xshift = -0.4cm, yshift = 0.15cm] {$0.5 {+} 2(n{-}1)\varepsilon$} (sn)
				(sn+1) edge [loop right, looseness = 4] node {$0.5 {+} 2n\varepsilon$} (sn+1)
				
				(x) edge [loop above, looseness = 4] node {$1$} (x)
				;
				
				\node[state] (b1) [fill = yellow, right of = a1, node distance = 10cm] {$b_1$};
				\node[state] (b2) [fill = green, below of = b1, node distance = 1.8cm] {$b_2$};
				\node[] (rhsdots) [below of = b2, node distance = 1.7cm] {\vdots};
				\node[state] (bn-1) [fill = cyan!50, below of = rhsdots, node distance = 1.8cm, inner sep = 0cm] {$b_{n-1}$};
				\node[state] (bn) [fill = red!50, below of = bn-1, node distance = 1.8cm] {$b_n$};
				
				\node[] (rhsdotstemp) [right of = rhsdots, node distance = 2cm] {\vdots};
				
				\node[state] (t1) [right of = b1, node distance = 2cm, fill = orange!30] {$t_1$};
				\node[state] (t2) [right of = b2, node distance = 2cm, fill = orange!30] {$t_2$};
				\node[state] (tn-1) [right of = bn-1, node distance = 2cm, inner sep = 0cm, fill = orange!30] {$t_{n-1}$};
				\node[state] (tn) [right of = bn, node distance = 2cm, fill = orange!30] {$t_{n}$};
				
				\node[state] (y) [fill = gray!40, right of = rhsdotstemp, node distance = 3cm] {$y$};
				
				\node[state, draw, rectangle] (init) [left of = rhsdots, node distance = 2.1cm, minimum width = 2cm, minimum height = 1cm] {$\mathbf{D(n)}$};
				\node (initr) [above of = init, node distance = 1.2cm] {};
				
				\path 
				
				(initr) edge (init)
				(init) edge [bend left = 10] node {$\frac{1}{n}$} (b1)
				(init) edge [bend left = 10] node [below, pos = 0.7] {$\frac{1}{n}$} (b2)
				(init) edge [bend right = 10] node [above, pos = 0.7] {$\frac{1}{n}$} (bn-1)
				(init) edge [bend right = 10] node  [left]{$\frac{1}{n}$} (bn)
				(b1) edge node [above,sloped] {$1$} (t1)
				(b2) edge node [below,sloped] {$1$} (t2)
				(bn-1) edge node [above,sloped] {$1$} (tn-1)
				(bn) edge node [above,sloped] {$1$} (tn)
				
				(t1) edge node [above,sloped] {$0.5 {-} \varepsilon$} (y)
				(t2) edge node [above,sloped] {$0.5 {-} 3\varepsilon$} (y)
				(tn-1) edge node [above,sloped] {$0.5 {-} (2n{-}3)\varepsilon$} (y)
				(tn) edge node [above,sloped] {$0.5 {-} (2n{-}1)\varepsilon$} (y)
				
				(t1) edge [loop right, looseness = 4] node {$0.5 {+} \varepsilon$} (t1)
				(t2) edge [loop above, looseness = 4] node {$0.5 {+} 3\varepsilon$}  (t2)
				(tn-1) edge [loop below, looseness = 4] node [xshift = -0.5cm] {$0.5 {+} (2n{-}3)\varepsilon$} (tn-1)
				(tn) edge [loop right, looseness = 4] node {$0.5 {+} (2n{-}1)\varepsilon$} (tn)
				
				(y) edge [loop above, looseness = 4] node {$1$} (y)
				;
				
				\node[right of = s, node distance = 1.2cm] {$\emptyset$};
				\node[above of = a1, node distance = 0.7cm] {$\{l_1\}$};
				\node[above of = a2, node distance = 0.7cm] {$\{l_2\}$};
				\node[below of = an-1, node distance = 0.7cm] {$\{l_{n-1}\}$};
				\node[below of = an, node distance = 0.7cm] {$\{l_n\}$};
				\node[left of = s1, node distance = 0.75cm, yshift = 0.2cm] {$\{a\}$};
				\node[right of = s2, node distance = 0.75cm] {$\{a\}$};
				\node[below right of = s3, node distance = 0.75cm] {$\{a\}$};
				\node[above left of = sn-1, node distance = 0.75cm, yshift = 0.1cm, xshift = 0.2cm] {$\{a\}$};
				\node[right of = sn, node distance = 0.75cm] {$\{a\}$};
				\node[left of = sn+1, node distance = 0.75cm, yshift = -0.2cm] {$\{a\}$};
				\node[right of = x, node distance = 0.75cm] {$\{b\}$};
				
				\node[right of = init, node distance = 1.2cm] {$\emptyset$};
				\node[above of = b1, node distance = 0.7cm] {$\{l_1\}$};
				\node[above of = b2, node distance = 0.7cm] {$\{l_2\}$};
				\node[below of = bn-1, node distance = 0.7cm] {$\{l_{n-1}\}$};
				\node[below of = bn, node distance = 0.7cm] {$\{l_n\}$};
				\node[above of = t1, node distance = 0.7cm] {$\{a\}$};
				\node[below of = t2, node distance = 0.7cm] {$\{a\}$};
				\node[above of = tn-1, node distance = 0.7cm] {$\{a\}$};
				\node[below of = tn, node distance = 0.75cm] {$\{a\}$};
				\node[right of = y, node distance = 0.75cm] {$\{b\}$};
		\end{tikzpicture}}
		\caption{Visualization of LMCs $\mathcal{M}_{n}$ and $\mathcal{N}_{n}$ used in the proof of \Cref{Theorem: Epsilon-bisimilarity does not imply common approximate quotients with small tolerance}.}
		\label{Figure: LMC-family used in the proof of Theorem: Epsilon-bisimilarity does not imply common approximate quotients with small tolerance}
	\end{figure}
	
	\TheoremEpsilonBisimilarityDoesNotImplyExistenceOfCommonQuotient*
	This theorem is a consequence of the following, more general result.
	\begin{lemma}\label{Theorem: Epsilon-bisimilarity does not imply common approximate quotients with small tolerance}
		There is a family $\{(\mathcal{M}_{n}, \mathcal{N}_{n}) \mid n \in \mathbb{N}_{\geq 1} \}$ of pairs of finite LMCs such that for all $n \in \mathbb{N}_{\geq 1}$ we have $\mathcal{M}_{n} \sim_\varepsilon \mathcal{N}_{n}$ for $0 < \varepsilon \leq \frac{1}{4n}$, but $\mathcal{M}_n \not \simeq_\delta \mathcal{N}_n$ for any $\delta < 2n  \varepsilon$.
	\end{lemma}
	\begin{proof}
		Let $n \in \mathbb{N}_{\geq 1}$ and $\varepsilon \in \left (0, \frac{1}{4n}\right]$. We consider the family $\{(\mathcal{M}_{n}, \mathcal{N}_{n}) \mid n \in \mathbb{N}_{\geq 1}\}$ consisting of pairs of LMCs constructed as those depicted in \Cref{Figure: LMC-family used in the proof of Theorem: Epsilon-bisimilarity does not imply common approximate quotients with small tolerance}. 
		
		In such a pair, both LMCs $\mathcal{M}_n$ and $\mathcal{N}_n$ consist of two parts. The first one, abstractly denoted $\mathbf{D}(n)$ in the figure, contains the initial state of the respective LMC and consists of a Markov chain that, by solely doing coin flips, generates a uniform distribution over the values $\{1,\dots,n \}$. The case for $n = 6$ is the well-known Knuth-Yao-Dice \cite{KY76}, which can be generalized to uniformly sample over arbitrary nonnegative integers $n \in \mathbb{N}_{\geq 1}$ \cite{ODUGCFA}. We assume that none of the states in $\mathbf{D}(n)$ is labeled by any atomic proposition and that every inner node of the LMC contains at most $2$ outgoing edges, both with probability exactly $\frac{1}{2}$. In the terminology of trees, the model has precisely $n$ leaves, and from the $i$-th leaf a transition from $\mathbf{D}(n)$ to the state $a_i$ and $b_i$, respectively, occurs with probability $1$. 
		
		The second part of the respective chains is given as follows: for $\mathcal{M}_{n}$, it has states $\{a_1, \dots, a_n, s_1, \dots, s_{n+1}, x\}$ with transition probabilities 
		\begin{alignat*}{3}
			\prob(a_i)(s_i) &= \frac{1}{2} \qquad &&\text{and} \qquad  \prob(a_i)(s_{i+1}) &&= \frac{1}{2} \\
			\prob(s_i)(s_i) &= \frac{1}{2} + 2  (i-1) \varepsilon \qquad&& \text{and} \qquad \prob(s_i)(x) &&= \frac{1}{2} - 2(i-1)\varepsilon \\
			\prob(x)(x) &= 1,
		\end{alignat*}
		while its state space in $\mathcal{N}_{n}$ is $\{b_1, \dots, b_n, t_1, \dots, t_n, y\}$ and the transition probabilities are given as 
		\begin{align*}
			\prob(b_i)(t_i) &= 1\\
			\prob(t_i)(t_i) &= \frac{1}{2} + (2i -1) \varepsilon \qquad \text{and} \qquad \prob(t_i)(y) = \frac{1}{2} - (2i - 1) \varepsilon\\
			\prob(y)(y) &= 1.
		\end{align*}
		Furthermore, $AP$ consists of the pairwise distinct atomic propositions $l_1, \dots, l_n, a, b$. The state labels are $l(a_i) = l(b_i) = \{l_i\}$ and $l(s_i) = l(t_i) = \{a\}$ for $1 \leq i \leq n$ as well as $l(x) = l(y) = \{b\}$. Note that the states in $\mathbf{D}(n)$ all have empty labels, so they cannot be related to any of the states $a_i, b_i, s_i, t_i, x$ or $y$ by any relation that preserves the state labels.
		
		The relation $\sim_\varepsilon$ in $\mathcal{M}_n \oplus \mathcal{N}_n$ is the symmetric and reflexive closure of the relation that contains all pairs of corresponding states in the respective copies of $\mathbf{D}(n)$, as well as
		\begin{align*}
			\{(a_i, b_i) \mid 1 \leq i \leq n \} \cup \{(x, y)\} \cup \{(s_i, t_i), (s_{i+1}, t_i) \mid 1 \leq i \leq n \}.
		\end{align*}
		In particular, $s_{init}^{\mathcal{M}_n} \sim_\varepsilon s_{init}^{\mathcal{N}_n}$.

		Let $\delta \geq 0$ such that $\mathcal{M}_n \simeq_\delta \mathcal{N}_n$, and let $\mathcal{M}_{n}', \mathcal{N}_{n}'$ be the corresponding $\delta$-perturbations, i.e., $\mathcal{M}_n' \sim \mathcal{N}_n'$. We now show $\delta \geq 2n\varepsilon$. 
		
		To this end, we first observe that every two states $a_i, a_j$ and $a_i, b_j$ and $b_i, b_j$ with $i \neq j$ have different labels, while $l(a_i) = l(b_i)$ for every $1 \leq i \leq n$. Furthermore, $\mathcal{M}_n' \sim \mathcal{N}_n'$ iff ${s_{init}^{\mathcal{M}_n}}' \sim {s_{init}^{\mathcal{N}_n}}'$.
		However, this can only be the case if $a_i'\sim b_i'$ for all $i$. Otherwise, we would have to perturb the states in $\mathbf{D}(n)$ in such a way that, for some $i$, the states $a_i$ and $b_i$ are not reachable anymore when traversing $\mathbf{D}(n)$. But, as all transition probabilities from one state in $\mathbf{D}(n)$ to another one in $\mathbf{D}(n)$ are equal to $\frac{1}{2}$, and the probability to transition from the $i$-th leaf of $\mathbf{D}(n)$ to $a_i$ or $b_i$ equals $1$ for all $i$, this requires a perturbation by at least $1 > 2n\varepsilon$. Hence, the only possible way to ensure $\mathcal{M}_n' \sim \mathcal{N}_n'$ for some $\delta$-perturbations with $\delta < 2n\varepsilon$ is to enforce $a_i' \sim b_i'$ for all $i$. 
		
		Let $1 \leq i \leq n$. First assume that, in $\mathcal{M}_n'$ and $\mathcal{N}_n'$, we have $s_i' \nsim t_i'$ and $s_{i+1}' \nsim t_i'$. Then it holds that $\prob(a_i)([t_i']_{\sim}) = 0$ while $\prob(b_i)([t_i']_{\sim}) = 1$, and it is easy to see that a perturbation of at least $1 \geq 2n\varepsilon$ is required to make the states $a_i'$ and $b_i'$ bisimilar. Now, assume that $s_i' \sim t_i'$ and $s_{i+1}' \nsim t_i'$. Here, $\prob(a_i)([t_i']_{\sim}) = \frac{1}{2}$ while $\prob(b_i)([t_i']_{\sim}) = 1$, and thus a perturbation of at least $\frac{1}{2} \geq 2n\varepsilon$ is necessary to obtain $a_i' \sim b_i'$ (by assigning, for example, a probability of $\frac{3}{4}$ to the transition $a_i' \to s_i'$ and $b_i' \to t_i'$, as well as a total probability of $\frac{1}{4}$ to transitions $a_i' \to [s_{i+1}']_\sim$ and $b_i' \to [s_{i+1}']_\sim$, respectively). By an analogous argument we can show that if $s_i' \nsim t_i'$ and $s_{i+1}' \sim t_i'$, a perturbation of at least $2n \varepsilon$ is required as well. 
		
		Thus, we can w.l.o.g. assume that, if $\delta < 2n\varepsilon$, it has to holds that $s_i' \sim t_i' \sim s_{i+1}'$ for every $1 \leq i \leq n$. But then the transitivity of $\sim$ implies $s_1' \sim s_{n+1}'$, which by construction is only possible when perturbing by at least $2n\varepsilon$: in the resulting model, probabilities of $\frac{1}{2} + n \varepsilon$ are assigned to the transitions $s_1' \to s_1'$ and $s_{n+1}' \to s_{n+1}'$, while the transitions $s_1' \to x'$ and $s_{n+1}' \to x'$ are taken with probability $\frac{1}{2} -n\varepsilon$. 
		
		All in all, this shows that for $\mathcal{M}_n \simeq_\delta \mathcal{N}_n$ a tolerance of at least $\delta = 2n\varepsilon$ is required, or equivalently, that $\mathcal{M}_n \not \simeq_\delta \mathcal{N}_n$ for any $\delta < 2n\varepsilon$. 
	\end{proof}
	
	\begin{proof}[Proof of \Cref{Theorem: Epsilon-Bisimulation does not imply existence of common 1/4-quotient}]
		Let $\varepsilon \in \left( 0, \frac{1}{4} \right]$ and set $n = n(\varepsilon) = \left \lfloor \frac{1}{4 \varepsilon} \right \rfloor$. As $\left \lfloor \frac{1}{4\varepsilon} \right \rfloor \leq \frac{1}{4 \varepsilon}$, this choice of $n$ implies both $\varepsilon \leq \frac{1}{4n}$ and $\frac{1}{4\varepsilon} < n+1$. 
		
		From the first inequality, it follows by \Cref{Theorem: Epsilon-bisimilarity does not imply common approximate quotients with small tolerance} that $\mathcal{M}_n \sim_{\varepsilon} \mathcal{N}_n$ and $\mathcal{M}_n \not \simeq_\delta \mathcal{N}_n$ for any $\delta < 2 n \varepsilon$, where $\mathcal{M}_n$ and $\mathcal{N}_n$ are again as depicted in \Cref{Figure: LMC-family used in the proof of Theorem: Epsilon-bisimilarity does not imply common approximate quotients with small tolerance}.
		Furthermore, $\frac{1}{4 \varepsilon} < n+1$ implies $\frac{1}{2} \cdot \frac{n}{n+1} < \frac{1}{2} \cdot \frac{n}{1/(4\varepsilon)} = 2 n \varepsilon$, so in particular $\mathcal{M}_n \not \simeq_\delta \mathcal{N}_n$ for any $\delta \leq \frac{1}{2} \cdot \frac{n}{n+1}$. As $f(x) = \frac{1}{2} \cdot \frac{x}{x+1}$ is strictly increasing in $x > 0$, and because $n = \lfloor \frac{1}{4\varepsilon} \rfloor \geq 1$ for any $\varepsilon \in (0, \frac{1}{4}]$, it follows that for any such $\varepsilon$ and $n$ we have $\frac{1}{4} = \frac{1}{2} \cdot \frac{1}{1+1} = f(1) \leq f(n) = \frac{1}{2} \cdot \frac{n}{n+1} < 2n\varepsilon$. In combination, this yields $\mathcal{M}_n \not \simeq_\delta \mathcal{N}_n$ for any $\delta \leq \frac{1}{4}$.
	\end{proof}
	
	\begin{figure}[t!]
		\centering
		\resizebox{!}{0.25\textheight}{
			\begin{tikzpicture}[->,>=stealth',shorten >=1pt,auto, semithick]
				\tikzstyle{every state} = [text = black]
				\node[state] (s) [fill = yellow] {$s$};
				\node (stemp) [right of = s, node distance = 2.8cm] {};
				\node (stempl) [left of = stemp, node distance = 2cm, yshift = 0.1cm] {$\vdots$};
				\node (stempr) [right of = stemp, node distance = 2cm, yshift = 0.1cm] {$\vdots$};
				\node[state] (s1) [above of = stemp, node distance = 3cm, fill = green] {$s_1$};
				\node[state] (s2) [above of = stemp, node distance = 1.5cm, fill = green] {$s_2$};
				\node[state] (sn) [below of = stemp, node distance = 1.5cm, fill = green] {$s_{n}$};
				\node[state] (sn+1) [below of = stemp, node distance = 3cm, fill = green, inner sep = 0cm] {$s_{n+1}$};
				\node[state] (x) [right of = stemp, node distance = 2.8cm] {$x$};
				
				\node (sinit) [above of = s, node distance = 1cm] {};
				
				\path
				(sinit) edge (s)
				(s) edge[bend left = 20] node [above,sloped] {$\frac{n+2}{(n+1)^2}$} (s1)
				(s) edge node [above,sloped] {$\frac{n+2}{(n+1)^2}$} (s2)
				(s) edge node [below,sloped] {$\frac{n+2}{(n+1)^2}$} (sn)
				(s) edge [bend right = 20] node [below,sloped] {$\frac{1}{(n+1)^2}$} (sn+1)
				
				(s1) edge [bend left = 20] node [above,sloped] {$0.5 {+} \varepsilon(n{-}1)$} (x)
				(s2) edge node [above,sloped] {$0.5 {+} \varepsilon(n{-}3)$} (x)
				(sn) edge node [below,sloped] {$0.5 {+} \varepsilon({-}n {+} 1)$} (x)
				(sn+1) edge [bend right = 20] node [below,sloped, pos = 0.6] {$0.5 {+} \varepsilon({-}n{-}1)$} (x)
				
				(s1) edge [loop right] node {$0.5 {-} \varepsilon(n{-}1)$} (s1)
				(s2) edge [loop below] node {$0.5 {-} \varepsilon(n{-}3)$} (s2)
				(sn) edge [loop above] node {$0.5 {-} \varepsilon({-}n{+}1)$} (sn)
				(sn+1) edge [loop right] node {$0.5 {-} \varepsilon({-}n{-}1)$} (sn+1)
				
				(x) edge [loop above] node [above] {$1$} (x)
				;
				
				\node[state] (t) [fill = yellow, right of = x, node distance = 2cm] {$t$};
				\node (ttemp) [right of = t, node distance = 2.8cm] {};
				\node (ttempr) [right of = ttemp, node distance = 2cm, yshift = 0.1cm] {$\vdots$};
				\node (ttempl) [left of = ttemp, node distance = 2cm, yshift = 0.1cm] {$\vdots$};
				\node[state] (t1) [above of = ttemp, node distance = 3cm, fill = green] {$t_1$};
				\node[state] (t2) [above of = ttemp, node distance = 1.5cm, fill = green] {$t_2$};
				\node[state] (tn) [below of = ttemp, node distance = 1.5cm, fill = green] {$t_{n}$};
				\node[state] (tn+1) [below of = ttemp, node distance = 3cm, fill = green, inner sep = 0cm] {$t_{n+1}$};
				\node[state] (y) [right of = ttemp, node distance = 2.8cm] {$y$};
				\node (tinit) [above of = t, node distance = 1cm] {};
				
				\path
				(tinit) edge (t)
				(t) edge [bend left = 20] node [above,sloped] {$\frac{1}{(n+1)^2}$} (t1)
				(t) edge node [above,sloped] {$\frac{n+2}{(n+1)^2}$} (t2)
				(t) edge node [below,sloped] {$\frac{n+2}{(n+1)^2}$} (tn)
				(t) edge [bend right = 20] node [below,sloped] {$\frac{n+2}{(n+1)^2}$} (tn+1)
				
				(t1) edge [bend left = 20] node [above,sloped] {$0.5 {+} n\varepsilon$} (y)
				(t2) edge node [above,sloped] {$0.5 {+} \varepsilon(n{-}2)$} (y)
				(tn) edge node [below,sloped] {$0.5 {+} \varepsilon({-}n {+} 2)$} (y)
				(tn+1) edge [bend right = 20] node [below,sloped] {$0.5 {+} \varepsilon({-}n)$} (y)
				
				(t1) edge [loop right] node {$0.5 {-} n\varepsilon$} (t1)
				(t2) edge [loop below] node {$0.5 {-} \varepsilon(n{-}2)$} (t2)
				(tn) edge [loop above] node {$0.5 {-} \varepsilon({-}n{+}2)$} (tn)
				(tn+1) edge [loop right] node {$0.5 {-} \varepsilon({-}n)$} (tn+1)
				
				(y) edge [loop above] node[above] {$1$} (y)
				;
				
				\node[below of = s, node distance = 0.7cm] {$\{a\}$};
				\node[left of = s1, node distance = 0.75cm] {$\{b\}$};
				\node[above of = s2, node distance = 0.7cm] {$\{b\}$};
				\node[below of = sn, node distance = 0.7cm] {$\{b\}$};
				\node[left of = sn+1, node distance = 0.75cm] {$\{b\}$};
				\node[right of = x, node distance = 0.75cm] {$\{c\}$};

				\node[below of = t, node distance = 0.7cm] {$\{a\}$};
				\node[left of = t1, node distance = 0.75cm] {$\{b\}$};
				\node[above of = t2, node distance = 0.7cm] {$\{b\}$};
				\node[below of = tn, node distance = 0.7cm] {$\{b\}$};
				\node[left of = tn+1, node distance = 0.75cm] {$\{b\}$};
				\node[right of = y, node distance = 0.75cm] {$\{c\}$};
		\end{tikzpicture}}
		\caption{The LMCs $\mathcal{M}_n$ (left) and $\mathcal{N}_n$ (right) for $n \in \mathbb{N}_{\geq 1}$ used in the proof of \Cref{Theorem: No common nepsilon quotient with graph isomorphism}. }
		\label{Figure: LMCs for even case of Theorem: No common nepsilon quotient with graph isomorphism}
	\end{figure}
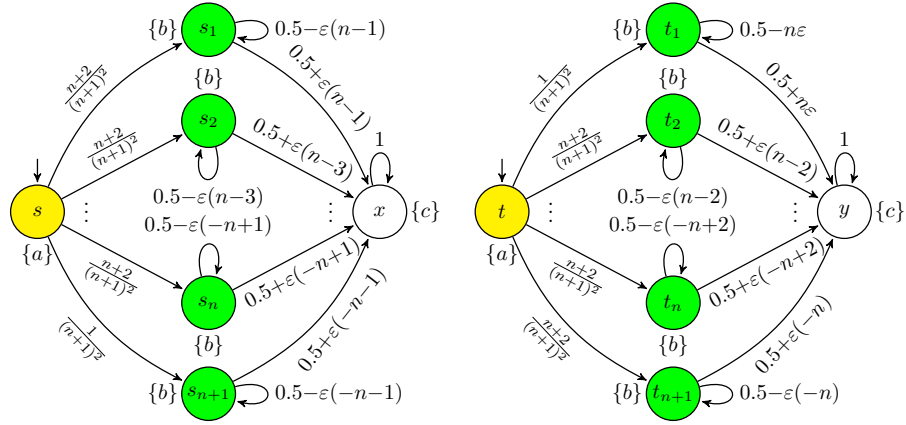
	
	\NoCommonNEpsilonQuotientWithGraphIsomorphism*
	\begin{proof}
		Let $n \in \mathbb{N}$ with $n \geq 1$. The family $\mathcal{F} = \{(\mathcal{M}_n, \mathcal{N}_n) \mid n \geq 1\}$ of LMCs that we consider is constructed as in \Cref{Figure: LMCs for even case of Theorem: No common nepsilon quotient with graph isomorphism}. All of these models are well-defined, as every transition probability is in $[0,1]$ by the choice of $\varepsilon$, and
		\begin{align*}
			\frac{1}{(n+1)^2} + n \cdot \frac{n+2}{(n+1)^2} = \frac{1 + n(n+2)}{(n+1)^2} = 1.
		\end{align*}
		To illustrate our construction, the LMCs $\mathcal{M}_2$ and $\mathcal{N}_2$ can be found in \Cref{Figure: Example LMCs for case n = 2 for construction in Theorem: No common nepsilon quotient with graph isomorphism}. By checking the conditions of $\varepsilon$-bisimilarity, it follows that always $\mathcal{M}_n \sim_\varepsilon \mathcal{N}_n$, and that the relations $\sim_{\varepsilon}$ are the reflexive and symmetric closures of 
		\begin{align*}
			\textcolor{white}{\cup}&\{(s,t), (x,y)\} 
			\cup \{(s_i, t_i) \mid i = 1,\dots,n + 1 \} \cup \{ s_i, t_{i+1}) \mid i = 1, \dots, n\}.
		\end{align*}
		In particular, we have $\vert {\sim_\varepsilon}(\{s_{n+1}\}) \vert = \vert {\sim_\varepsilon}(\{t_{1}\}) \vert = 2$, while $\vert {\sim_\varepsilon}(\{s_{i}\}) \vert = \vert {\sim_\varepsilon}(\{t_{j}\}) \vert = 3$ for all remaining choices of $i, j$. 
		
		Now assume that $\mathcal{M}_n \simeq_\delta \mathcal{N}_n$, and let $\mathcal{M}_n', \mathcal{N}_n'$ be the corresponding $\delta$-perturbations of $\mathcal{M}_n$ and $\mathcal{N}_n$, i.e., let $\mathcal{M}_n' \sim \mathcal{N}_n'$.
		
		For $\mathcal{M}_n' \sim \mathcal{N}_n'$ to hold, it follows by construction that we need to have $s' \sim t'$ and $x' \sim y'$ in $\mathcal{M}_n' \oplus \mathcal{N}_n'$. While the latter is obviously possible for all $\delta \geq 0$, the former requires for every $\sim$-class $C$ that $\prob'(s')(C') = \prob'(t')(C')$, and in particular that 
		\begin{align}
			\prob'(s')([s_{n+1}']_\sim) = \prob'(t')([s_{n+1}']_\sim). \label{Proof: Theorem: No common nepsilon quotient with graph isomorphism - Star 1}
		\end{align}
		By the choice of $\varepsilon$ it follows that if we perturb $\mathcal{M}_n$ such that the probability to move to $s_{n+1}' \in [s_{n+1}']$ becomes $0$, $\delta$ already has to be $\geq n \varepsilon$, which shows the claim. Hence, we can w.l.o.g. assume that $\prob'(s')([s_{n+1}']_\sim) > 0$, and therefore also that $\prob'(t')([s_{n+1}']_\sim) > 0$.
		The equivalence class in question is thus of the form
		\begin{align*}
			[s_{n+1}']_{\sim} = \{m_1', \dots, m_k', n_1', \dots, n_l' \mid m_1', \dots, m_k' \in S^{\mathcal{M}_n'}, n_1', \dots, n_l' \in S^{\mathcal{N}_n'}\},
		\end{align*}
		for $k, l \geq 1$. Further, the construction of $\mathcal{M}_n$ yields that the probability to move from $s$ to any set of states $A \subseteq S$ in $\mathcal{M}_n$ that contains $s_{n+1}$ is precisely 
		\begin{align*}
			\prob(s)(A) = \frac{1}{(n+1)^2} + (\vert A \vert -1) \frac{n+2}{(n+1)^2} = \frac{1 + (\vert A \vert - 1)(n+2)}{(n+1)^2}.
		\end{align*}
		Because we assume $\mathcal{M}_n'$ to be a $\delta$-perturbation of $\mathcal{M}_n$, it now follows from \Cref{Lemma: Difference of Epsilon-Perturbation Probabilities} that in $\mathcal{M}_n'$ we have $\prob'(s')([s_{n+1}']_\sim)  \in [lb(s'), ub(s')]$ for 
		\begin{align*}
			lb(s') &= \frac{1+(k-1)(n+2)}{(n+1)^2} - \frac{\delta}{2} \quad \text{ and } \quad
			ub(s') = \frac{1+(k-1)(n+2)}{(n+1)^2} + \frac{\delta}{2}.
		\end{align*}
		
		For the probability $\prob'(t')([s_{n+1}']_\sim)$, on the other hand, we have to consider two cases. If $t_1' \in [s_{n+1}']_\sim$ then $t_1' \sim s_{n+1}'$. But, by construction, making these two states bisimilar requires a perturbation by at least $\delta = 2n \varepsilon$ for both $\mathcal{M}_n$ and $\mathcal{N}_n$. Hence, in this case, $\delta \geq n \varepsilon$, which proves the claim. 
		
		Otherwise, if $t_1' \notin [s_{n+1}']_\sim$, then it follows again from the construction and \Cref{Lemma: Difference of Epsilon-Perturbation Probabilities} that $\prob'(t')([s_{n+1}']_\sim) \in [lb(t'), ub(t')]$, where the bounds are given as 
		\begin{align*}
			lb(t') = \frac{l(n+2)}{(n+1)^2} - \frac{\delta}{2} \quad \text{ and } \quad ub(t') = \frac{l(n+2)}{(n+1)^2} + \frac{\delta}{2}.
		\end{align*}
		
		In particular, we \Cref{Proof: Theorem: No common nepsilon quotient with graph isomorphism - Star 1} can only hold if there is a $p \in [lb(s'), ub(s')] \cap [lb(t'), ub(t')]$ with $p = \prob'(s')([s_{n+1}']_\sim) = \prob'(t')([s_{n+1}']_\sim)$.
		We utilize this fact to show $\delta \geq n\varepsilon$, by doing a case distinction on how $[lb(s'), ub(s')]$ and $[lb(t'), ub(t')]$ can intersect.
		\begin{enumerate}
			\item [] \textbf{Case 1. $lb(s') \leq lb(t') \leq ub(s') \leq ub(t')$.}
			The first and third inequality are equivalent to $1 \leq (l - k + 1)(n+2)$, and since $n+2 \geq 3$ and $l,k \in \mathbb{N}_{\geq 1}$ this implies $l - k + 1 \geq 1$. Moreover, the inequality $lb(t') \leq ub(s')$ is equivalent to $\frac{(n+2)(l-k+1) -1}{(n+1)^2} \leq \delta$. As $\varepsilon \in \left(0, \frac{1}{n(n+1)^2} \right]$, we can combine all of this to get
			\begin{align*}
				n \varepsilon \leq \frac{1}{(n+1)^2} \leq \frac{n+2 - 1}{(n+1)^2} \leq \frac{(n+2)(l-k+1) - 1}{(n+1)^2} \leq \delta.
			\end{align*}
			
			\item [] \textbf{Case 2. $lb(s') \leq lb(t') \leq ub(t') \leq ub(s')$.}
			Similar to Case 1, $lb(s') \leq lb(t')$ implies $l - k + 1\geq 1$, and thus $l \geq k$. On the other hand, $ub(t') \leq ub(s')$ yields $(l - k + 1)(n+2) \leq 1$, which, because $n +2 \geq 3$ and $l,k \in \mathbb{N}_{\geq 1}$, gives us $l-k+1 \leq 0$. But then $l + 1 \leq k$, so $l < k$, a contradiction to $l \geq k$. This case can therefore not occur. 
			
			\item [] \textbf{Case 3. $lb(t') \leq lb(s') \leq ub(s') \leq ub(t')$.}
			As for Case 2 we can reach a contradiction by observing that $lb(t') \leq lb(s')$ implies $l < k$ while $ub(s') \leq ub(t')$ yields $l \geq k$.
			
			\item [] \textbf{Case 4. $lb(t') \leq lb(s') \leq ub(t') \leq ub(s')$.}
			The first and third inequality are equivalent to $(n+2)(l-k+1) \leq 0$. Together with $n + 2 \geq 3$ and $l,k \in \mathbb{N}_{\geq 1}$ this yields $l-k+1 \leq 0$, or, equivalently, $k - l - 1 \geq 0$. Furthermore, $lb(s') \leq ub(t')$ can be transformed into $\frac{(n+2)(k-l-1) +1}{(n+1)^2} \leq \delta$. Similar to Case 1 we conclude 
			\begin{align*}
				n \varepsilon \leq \frac{1}{(n+1)^2} \leq \frac{(n+2)(k-l-1) + 1}{(n+1)^2} \leq \delta. 
			\end{align*}
		\end{enumerate}
		All in all, this shows that $\mathcal{M}_n \simeq_\delta \mathcal{N}_n$ only if $\delta \geq n \varepsilon$, as otherwise \Cref{Proof: Theorem: No common nepsilon quotient with graph isomorphism - Star 1} is violated.
	\end{proof}
	
	\CharacterizationPertubedEpsBisimTransitiveEpsBisim*
	\begin{proof}
		We set $S = S^\mathcal{M} \oplus S^\mathcal{N}$ and $\prob = \prob^{\mathcal{M} \oplus \mathcal{N}}$.
		
		\noindent
		\textbf{(i) $\Leftrightarrow$ (ii).} The equivalence is clear due to the $1$-to-$1$-correspondence of $\varepsilon$-perturbations of $\mathcal{M} \oplus \mathcal{N}$ and pairs of $\varepsilon$-perturbations of $\mathcal{M}$ and $\mathcal{N}$. 
		
		\noindent
		\textbf{(ii) $\Rightarrow$ (iii).} Let $\mathcal{L}$ be an $\varepsilon$-perturbation of $\mathcal{M} \oplus \mathcal{N}$ in which $s_{init}^\mathcal{M} \sim s_{init}^\mathcal{N}$. Following an argument similar to the one in the proof of \Cref{Corollary: Perturbed Epsilon Bisim implies Epsilon Bisim}, $\sim_\mathcal{L}$ on $\mathcal{L}$ is an $\varepsilon$-bisimulation on $\mathcal{M} \oplus \mathcal{N}$. Further, as $s_{init}^\mathcal{M} \sim_\mathcal{L} s_{init}^\mathcal{N}$ by assumption and $\sim_\mathcal{L}$ is clearly transitive, this relation satisfies all the requirements posed in (iii). 
		
		Now let $A \in \xfrac{S}{\sim_\mathcal{L}}$. The distribution $\prob_A^* \in Distr(\xfrac{S}{\sim_\mathcal{L}})$ with $\prob_A^*(B) = \prob_\mathcal{L}(u')(B')$ for some (or, equivalently, all) $u \in A$ satisfies the requirements of (iii), as by \Cref{Lemma: Difference of Epsilon-Perturbation Probabilities} it holds for all $\sim_\mathcal{L}$-closed $C \subseteq S$ that 
		\begin{align*}
			\vert \prob(s)(C) - \prob_A^*(C) \vert = \vert \prob(s)(C) - \prob_\mathcal{L}(s')(C') \vert \leq \frac{\varepsilon}{2}.
		\end{align*}
		\textbf{(iii) $\Rightarrow$ (ii).} We set $S = S_\mathcal{M} \oplus S_\mathcal{N}$. Let $R$ be a transitive $\varepsilon$-bisimulation on $\mathcal{M} \oplus \mathcal{N}$ such that $(s_{init}^\mathcal{M}, s_{init}^\mathcal{N}) \in R$, and such that for every $A \in \xfrac{S}{R}$ there is a distribution $\prob_A^*\in Distr(\xfrac{S}{R})$ which satisfies for all $R$-closed sets $C \subseteq S$ and all $s \in A$ that $\vert \prob(s)(C) - \prob_A^*(C) \vert \leq \frac{\varepsilon}{2}$.
		
		For $B \in \xfrac{S}{R}$, we define for fixed $A \in \xfrac{S}{R}$ and $s \in A$ the value $\Delta(s, B) = \prob(s)(B) - \prob_A^*(B)$, and set 
		\begin{align*}
			B^+_s = \{B \in \xfrac{S}{R} \mid \Delta(s, B) > 0\} \qquad \text{ and } \qquad B^-_s = \{B \in \xfrac{S}{R} \mid \Delta(s, B) < 0\}.
		\end{align*}
		Note that, if $B \in B^+_s$ we can pick, for all $t \in B$, values $0 \leq \delta^+(s,t) \leq \prob(s)(t)$ such that $\sum_{t \in B} \delta^+(s,t) = \Delta(s,B)$. Similarly, if $B \in B^-_s$, there are $0 \leq \delta^-(s,t) \leq 1-\prob(s)(t)$ for all $t \in B$ such that $\sum_{t \in B} \delta^-(s,t) = \vert \Delta(s,B) \vert$.
		
		Define a new LMC $\mathcal{L}$ on $S$, whose labeling is just like that of $\mathcal{M} \oplus \mathcal{N}$ and whose transition distributions $\prob_\mathcal{L}(s) \in Distr(S)$ for any $s \in S$ are, for $B \in \xfrac{S}{R}$ and $t \in B$, given by 
		\begin{align*}
			\prob_\mathcal{L}(s)(t) = \begin{cases}
				\prob(s)(t), & \text{if } \Delta(s,B) = 0 \\
				\prob(s)(t) - \delta^+(s, t) & \text{if } B \in B^+_s \\
				\prob(s)(t) + \delta^-(s,t) & \text{if} B \in B^-_s
			\end{cases}.
		\end{align*}
		
		Then $\prob_\mathcal{L}(s)$ is a distribution for all $s \in S$, as
		\begin{align*}
			&\sum_{t \in S} \prob_\mathcal{L}(s)(t) \\&= \sum_{B \in S/R} \prob_\mathcal{L}(s)(B) \\
			&= \sum_{\substack{B \in S/R,\\ \Delta(s,B) = 0}} \prob_\mathcal{L}(s)(B) + \sum_{B \in B^+_s} \prob_\mathcal{L}(s)(B) + \sum_{B \in B^-_s} \prob_\mathcal{L}(s)(B) \\
			&= \sum_{\substack{B \in S/R,\\ \Delta(s,B) = 0}} \prob(s)(B) + \sum_{B \in B^+_s} (\prob(s)(B) - \Delta(s,B)) + \sum_{B \in B^-_s} (\prob(s)(B) + \vert \Delta(s,B) \vert) \\
			&= \sum_{\substack{B \in S/R,\\ \Delta(s,B) = 0}} \prob_A^*(B) + \sum_{B \in B^+_s} \prob_A^*(B) + \sum_{B \in B^-_s} (\prob(s)(B) + \vert \prob(s)(B) - \prob_A^*(B) \vert) \\
			&= \sum_{\substack{B \in S/R,\\ \Delta(s,B) \geq 0}} \prob_A^*(B) + \sum_{B \in B^-_s} (\prob(s)(B) - \prob(s)(B) + \prob_A^*(B)) \\
			&=  \sum_{B \in S/R} \prob_A^*(B) = 1.
		\end{align*}
		Furthermore, $R$ is a probabilistic bisimulation on $\mathcal{L}$, as for all $s,t \in A$ and $B \in \xfrac{S}{R}$ the construction yields $\prob_\mathcal{L}(s)(B) = \prob_A^*(B) = \prob_\mathcal{L}(t)(B)$.  
		
		Now fix $s \in A$. Then 
		\begin{align*}
			1&= \sum_{B \in S/R} \prob(s)(B) = \sum_{B \in S/R} (\prob(s)(B) + \prob_A^*(B) - \prob_A^*(B)) \\&= \sum_{B \in S/R} \prob_A^*(B) + \sum_{B \in S/R} \Delta(s, B) = 1 + \sum_{B \in S/R} \Delta(s, B),
		\end{align*}
		so $\sum_{B \in S/R} \Delta(s, B) = 0$
		and hence $\sum_{B \in B^+_s} \Delta(s, B) = - \sum_{B \in B^-_s} \Delta(s, B)$, or, equivalently,
		\begin{align}
			\sum_{B\in B^+_s} \Delta(s,B) = \sum_{B \in B^-_s} \vert \Delta(s,B) \vert \label{Proof: CharacterizationPertubedEpsBisimTransitiveEpsBisim - Eq1}.
		\end{align}
		Therefore
		\begin{align*}
			\Vert \prob(s) - \prob_\mathcal{L}(s) \Vert_1 &= \sum_{t \in S} \vert \prob(s)(t) - \prob_\mathcal{L}(s)(t) \vert \\
			&= \sum_{B \in S/R} \sum_{t \in B} \vert \prob(s)(t) - \prob_\mathcal{L}(s)(t) \vert \\
			&= \underbrace{\sum_{\substack{B \in S/R, \\ \Delta(s, B) = 0}} \sum_{t \in B} \vert \prob(s)(t) - \prob_\mathcal{L}(s)(t) \vert}_{=0} + \sum_{\substack{B \in S/R, \\ \Delta(s, B) \neq 0}} \sum_{t \in B} \vert \prob(s)(t) - \prob_\mathcal{L}(s)(t) \vert \\
			&= \underbrace{\sum_{B \in B^+_s} \sum_{t \in B} \vert \prob(s)(t) - \prob_\mathcal{L}(s)(t) \vert}_{= \sum_{B \in B^+_s} \Delta(s, B)}  + \underbrace{\sum_{B \in B^-_s} \sum_{t \in B} \vert \prob(s)(t) - \prob_\mathcal{L}(s)(t) \vert}_{= \sum_{B \in B^-_s} \vert \Delta(s, B) \vert \overset{(\ref{Proof: CharacterizationPertubedEpsBisimTransitiveEpsBisim - Eq1})}{=} \sum_{B \in B^+_s} \Delta(s,B)}\\
			&= 2 \cdot \sum_{B \in B^+_s} \Delta(s,B) \\
			&= 2 \cdot \sum_{B \in B^+_s} (\prob(s)(B) - \prob_A^*(B)) \\
			&= 2 \cdot \left \vert \prob(s)\left(\bigcup_{B \in B^+_s}B\right) - \prob_A^*\left(\bigcup_{B \in B^+_s}B\right) \right \vert \\
			&\leq 2 \cdot \frac{\varepsilon}{2} = \varepsilon,
		\end{align*}
		where the inequality follows from the properties of $\prob_A^*$ and the fact that $\bigcup_{B \in B^+_s} B$ is, as a union of $R$-equivalence classes, $R$-closed. All in all, this shows that $\mathcal{L}$ is indeed an $\varepsilon$-perturbation of $\mathcal{M} \oplus \mathcal{N}$ in which, as $R$ is a bisimulation on this model and since $(s_{init}^\mathcal{M}, s_{init}^\mathcal{N}) \in R$ by assumption, it holds that $s_{init}^\mathcal{M} \sim s_{init}^\mathcal{N}$. 
	\end{proof}
	
	\TauStarConditionEquivalence*
	\begin{proof}
		\textbf{(i) $\Rightarrow$ (ii).} Let $\mu^* \in Distr(X)$ be as in (i), and let $l \in \{1, \dots, k\}$. We define $B_{l}^+ = \{i \in X \mid \mu_l(i) \geq \mu^*(i)\}$ and $B_{l}^- = \{i \in X \mid \mu_l(i) < \mu^*(i)\}$. Then $X = B_{l}^+ \uplus B_l^{-}$ and
		\begin{align*}
			\Vert \mu_l - \mu^* \Vert_1 &= \sum_{i \in X} \vert \mu_l(i) - \mu^*(i) \vert \\
			&= \sum_{i \in B_{l}^+} \vert \mu_l(i) - \mu^*(i) \vert + \sum_{i \in B_l^-} \vert \mu_l(i) - \mu^*(i) \vert \\
			&= \sum_{i \in B_{l}^+} (\mu_l(i) - \mu^*(i)) + \sum_{i \in B_l^-} (\mu^*(i) - \mu_l(i)) \\
			&= \underbrace{\mu_l(B_{l}^+) - \mu^*(B_{l}^+)}_{\geq 0} + \underbrace{\mu^*(B_{l}^-) - \mu_l(B_l^-)}_{\geq 0} \\
			&= \vert \mu_l(B_{l}^+) - \mu^*(B_{l}^+) \vert + \vert \mu_l(B_{l}^-) - \mu^*(B_l^-) \vert \\
			&\leq \frac{\varepsilon}{2}+ \frac{\varepsilon}{2} = \varepsilon
		\end{align*}
		where the last inequality follows from the fact that $B_{l}^+, B_l^- \subseteq X$ and the properties of $\mu^*$. As we have chosen $l \in \{1, \dots, k\}$ arbitrarily it follows that $\mu = \mu^*$ is a suitable choice for a distribution that satisfies the conditions in (ii).
		
		\textbf{(ii) $\Rightarrow$ (i).} This is just a special case of \Cref{Lemma: Difference of Epsilon-Perturbation Probabilities} and can be proved analogously. 
		
		\textbf{(ii) $\Rightarrow$ (iii).} Let $\mu \in Distr(X)$ be a distribution as in (ii). We have to show that there is a non-negative solution for the linear constraint system of (iii). To this end, set $x_i = \mu(i)$ and $\delta_{l, i} = \vert \mu_l(i) - \mu(i) \vert$ for all $i \in X, l \in \{1, \dots, k\}$. Then $x_i, \delta_{l, i} \geq 0$ for all $i, l,$ and: 
		\begin{itemize}
			\item $\sum_{i \in X} x_i = \sum_{i \in X} \mu(i) = 1$ as $\mu \in Distr(X)$, so the first constraint is satisfied. 
			\item For all $i, l,$ we have $\delta_{l, i} = \vert \mu_l(i) - \mu(i) \vert$ which implies $- \delta_{l, i} \leq \mu_l(i) - \mu(i) \leq \delta_{l, i}$. As $x_i = \mu(i)$ for all i it follows that $- \delta_{l, i} \leq \mu_l(i) - x_i \leq \delta_{l, i}$, so the second and third constraint are satisfied. 
			\item $\sum_{i \in X} \delta_{l, i} = \sum_{i \in X} \vert \mu_l(i) - \mu(i) \vert = \Vert \mu_l - \mu \Vert_1 \leq \varepsilon$ for all $i, l,$ by the assumption on $\mu$, so the fourth constraint is satisfied.
		\end{itemize}
		All in all this shows that the chosen assignment of values to the variables $x_i$ and $\delta_{l, i}$ does indeed constitute a non-negative solution of the linear constraint system.
		
		\textbf{(iii) $\Rightarrow$ (ii).} Let $\{x_i \mid i \in X\} \cup \{\delta_{l, i} \mid l \in \{1, \dots, k\}, i \in X\}$ be a non-negative solution of the linear constraint system of (iii). Define $\mu: X \to [0,1]$ via $\mu(i) = x_i$ for all $i \in X$. Then $\mu$ is a distribution as $\sum_{i \in X} \mu(i) = \sum_{i \in X} x_i = 1$ by the first constraint, and for $l \in \{1, \dots, k\}$,
		\begin{align*}
			\Vert \mu_l - \mu \Vert_1 = \sum_{i \in X} \vert \mu_l(i) - \mu(i) \vert = \sum_{i \in X} \vert \mu_l(i) - x_i \vert \overset{\text{Constr. } 2 + 3} = \sum_{i \in X} \delta_{l, i} \overset{\text{Constr. } 4}{\leq} \varepsilon.
		\end{align*}
		Thus, $\mu$ is a distribution satisfying the requirements of (ii).
	\end{proof}
	
	\TheoremNpCompleteCommonEpsilonQuotient*
	\begin{proof}
		We start with the first claim, i.e., we start by showing that it is \textsf{NP}-complete to decide if $\mathcal{M} \simeq_\varepsilon \mathcal{N}$ for given $\mathcal{M}, \mathcal{N}$ and $\varepsilon > 0$.
		The proof ideas take inspiration from the proof of \cite[Thm. 1]{ABM}, which can be found in the full version \cite{ABMfull} of \cite{ABM}.
		
		\smallskip 
		\noindent
		\textbf{\textsf{NP}-membership: } 
		Let $\mathcal{M}$ and $\mathcal{N}$ be two LMCs, and set $S = S^\mathcal{M} \oplus S^\mathcal{N}$. We consider the following nondeterministic polynomial time algorithm: 
		\begin{enumerate}
			\item Nondeterministically guess a partition $X$ of $S$ such that 
			\begin{enumerate}
				\item $X$ only relates states with the same label 
				\item For each block $B \in X$ it holds that $B \cap S_1 \neq \emptyset \neq B \cap S_2$
				\item There is a block $B_{init} \in X$ such that $s_{init}^{\mathcal{M}_1}, s_{init}^{\mathcal{M}_2} \in B_{init}$
			\end{enumerate}
			\item For every $B \in X$, apply \Cref{Lemma: Tau-Star Condition Equivalence} to $X$ and $\{\prob(s) \mid s \in B\}$, where the distributions are lifted to $X$ in the natural way, to check if there is a $\prob_B \in Distr(X)$ that satisfies condition (ii) of the lemma. If not, go back to the first step. 
		\end{enumerate}
		
		This algorithm has a solution iff $\mathcal{M} \simeq_\varepsilon \mathcal{N}$. The direction from right to left is clear, so we focus on the implication from left to right. Let $X, \{\prob_B \mid B \in X\}$ be a solution of the algorithm, and let $s \in B$ for some $B \in X$. By Lemma 13 of the full version \cite{ABMfull} of \cite{ABM} we can compute, in polynomial time, a distribution $\widehat{\prob}(s) \in Distr(S)$ such that $\widehat{\prob}(s)(C) = \prob_B(C)$ for all $C \in X$ and $\Vert \prob(s) - \widehat{\prob}(s) \Vert_1 = \sum_{C \in X} \vert \prob(s)(C) - \prob_B(C) \vert$. Define new LMCs $\mathcal{M}', \mathcal{N}'$ that differ from $\mathcal{M}$ and $\mathcal{N}$ only in the transition distribution functions, which are defined as $\prob'(s)(t) = \widehat{\prob}(s)(t)$. Then $\mathcal{M}'$ and $\mathcal{N}'$ are $\varepsilon$-perturbations of $\mathcal{M}$ and $\mathcal{N}$, respectively, as for any $s \in S$ it holds that
		\begin{align*}
			\Vert \prob(s) - \prob'(s') \Vert_1 = \Vert \prob(s) - \widehat{\prob}(s) \Vert_1 = \sum_{C \in X} \vert \prob(s)(C) - \prob_B(C) \vert \leq \varepsilon
		\end{align*}
		Furthermore, for $B, C \in X$, and $s,t \in B$,
		\begin{align*}
			\prob'(s')(C) = \widehat{\prob}(s)(C) = \prob_B(C) = \widehat{\prob}(t)(C) = \prob'(t')(C), 
		\end{align*}
		so the equivalence induced by $X$ is a probabilistic bisimulation on $\mathcal{M}' \oplus \mathcal{N}'$. Because there is a dedicated block $B_{init} \in X$ with $s_{init}^{\mathcal{M}}, s_{init}^{\mathcal{N}} \in B_{init}$ we have $\mathcal{M}' \sim \mathcal{N}'$, and hence $\mathcal{M} \simeq_\varepsilon \mathcal{N}$. 
		
		\begin{figure}[tb]
			\centering
			\resizebox{!}{0.135\textheight}{
				\begin{tikzpicture}[->,>=stealth',shorten >=1pt,auto, semithick]
					\tikzstyle{every state} = [text = black]
					\node[state] (s) [] {$s$};
					\node (stemp) [below of = s, node distance = 1.2cm] {\dots};
					\node[state] (s1) [left of = stemp, node distance = 2cm] {$s_1$};
					\node[state] (sn) [right of = stemp, node distance = 2cm] {$s_n$};
					\node[state] (sa) [below of = s1, node distance = 1.5cm] {$s_a$};
					\node[state] (sb) [below of = sn, node distance = 1.5cm, fill = green] {$s_b$};
					
					\node (sinit) [left of = s, node distance = 1.2cm] {};
					
					\path
					(sinit) edge (s)
					(s) edge node [pos = 0.4, left, yshift = 0.1cm] {$\frac{p_1}{T}$} (s1)
					(s) edge node [pos = 0.4, right, yshift = 0.1cm] {$\frac{p_n}{T}$} (sn)
					(s1) edge node [left, pos = 0.4] {$\frac{1}{2}$} (sa)
					(s1) edge node [above, pos = 0.25] {$\frac{1}{2}$} (sb)
					(sn) edge node [right, pos = 0.4] {$\frac{1}{2}$} (sb)
					(sn) edge node[above, pos = 0.25] {$\frac{1}{2}$} (sa)
					(sa) edge [loop left] node {$1$} (sa)
					(sb) edge [loop right] node {$1$} (sb)
					;
					
					\node[state] (t) [right of = s, node distance = 8cm] {$t$};
					\node (ttemp) [below of = t, node distance = 1.2cm] {};
					\node[state] (t1) [left of = ttemp, node distance = 2cm] {$t_y$};
					\node[state] (t2) [right of = ttemp, node distance = 2cm] {$t_n$};
					\node[state] (ta) [below of = t1, node distance = 1.5cm] {$t_a$};
					\node[state] (tb) [below of = t2, node distance = 1.5cm, fill = green] {$t_b$};
					
					\node (tinit) [left of = t, node distance = 1.2cm] {};
					
					\path
					(tinit) edge (t)
					(t) edge node [pos = 0.4, left, yshift = 0.1cm] {$\frac{N}{T}$} (t1)
					(t) edge node [pos = 0.4, right, yshift = 0.1cm] {$1 - \frac{N}{T}$} (t2)
					(t1) edge node [left, pos = 0.4] {$\frac{1}{2} - \varepsilon$} (ta)
					(t1) edge node [above, pos = 0.25, xshift = 0.2cm] {$\frac{1}{2}+ \varepsilon$} (tb)
					(t2) edge node [right, pos = 0.4] {$\frac{1}{2}- \varepsilon$} (tb)
					(t2) edge node[above, pos = 0.25] {$\frac{1}{2}  + \varepsilon$} (ta)
					(ta) edge [loop left] node {$1$} (ta)
					(tb) edge [loop right] node {$1$} (tb)
					;
					
					\node[right of = s, node distance = 0.75cm] {$\{a\}$};
					\node[left of = s1, node distance = 0.75cm] {$\{a\}$};
					\node[right of = sn, node distance = 0.75cm] {$\{a\}$};
					\node[right of = sa, node distance = 0.75cm, yshift = -0.2cm] {$\{a\}$};
					\node[left of = sb, node distance = 0.75cm, yshift = -0.2cm] {$\{b\}$};
					
					\node[right of = t, node distance = 0.75cm] {$\{a\}$};
					\node[left of = t1, node distance = 0.75cm] {$\{a\}$};
					\node[right of = t2, node distance = 0.75cm] {$\{a\}$};
					\node[right of = ta, node distance = 0.75cm, yshift = -0.2cm] {$\{a\}$};
					\node[left of = tb, node distance = 0.75cm, yshift = -0.2cm] {$\{b\}$};
			\end{tikzpicture}}
			\caption{The LMCs $\mathcal{M}$ (left) and $\mathcal{N}$ (right) used in the \textsf{NP}-hardness proof of \Cref{Theorem: NP-completeness of common epsilon-quotient} (adapted from \cite[Fig. 3]{ABM}).}
			\label{Figure: LMCs in hardness of common Epsilon-Quotient}
		\end{figure}
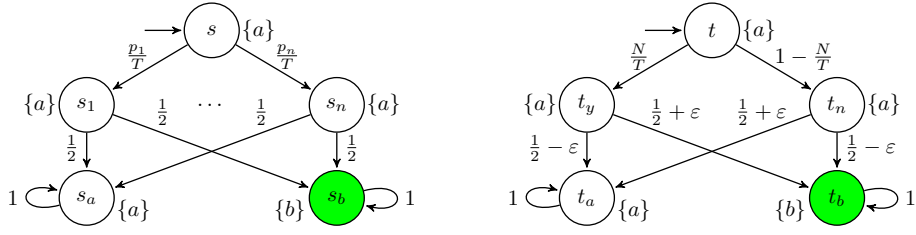
		
		\smallskip
		\noindent
		\textbf{\textsf{NP}-hardness: }
		We adapt the \textsf{NP}-hardness proof of \cite[Thm. 1]{ABM} (cf. the full version \cite{ABMfull} of \cite{ABM}) that uses a reduction from the well-known \textsc{SubsetSum}-problem.
		An instance $(P, N)$ of \textsc{SubsetSum} consists of a nonempty set $P = \{p_1, \dots, p_n\} \subseteq \mathbb{N}$ together with a value $N \in \mathbb{N}$, and poses the question if there is an index set $I \subseteq \{1,\dots,n\}$ such that the sum of all elements associated to the indices in $I$ equals $N$, i.e., such that $\sum_{i \in I} p_i = N$.
		
		For a \textsc{SubsetSum}-instance $(P,N)$, let $T = \sum_{i=1, \dots, n} p_i$ be the sum of all elements in $P$, and consider the LMCs $\mathcal{M}$ and $\mathcal{N}$ depicted in \Cref{Figure: LMCs in hardness of common Epsilon-Quotient}, which can be constructed from $(P,N)$ in polynomial time.
		
		In \cite{ABMfull}, it is shown that the \textsc{SubsetSum}-instance $(P, N)$ has a solution iff the LMC $\mathcal{M} \oplus \mathcal{N}$ has an $\varepsilon=\frac{1}{2T}$-quotient with $k=5$ states.
		We use this result to show hardness of the question whether $\mathcal{M} \simeq_\varepsilon \mathcal{N}$. 
		
		Let $\varepsilon=\frac{1}{2T}$.
		First observe that there is no $\varepsilon$-pertubation $\mathcal{N}'$ of $\mathcal{N}$ in which two distinct states are bisimilar. Now assume that $\mathcal{M} \simeq_\varepsilon \mathcal{N}$, and let $\mathcal{M}', \mathcal{N}'$ be the corresponding $\varepsilon$-perturbations with $\mathcal{M}' \sim \mathcal{N}'$.
		For $i \in \{ 1, \dots, n \}$, the perturbed states satisfy either $s_i' \sim t_y'$ or $s_i' \sim t_n'$.
		Due to $s' \sim t'$, it can be shown that $I = \{ i \in \{1, \dots, n \} \mid s_i' \sim t_y'\}$ is a solution to $(P,N)$. 
		On the other hand, any index set $I \subseteq \{1, \dots, n \}$ induces $\varepsilon$-perturbations such that $s_i' \sim t_y'$ if $i \in I$  and $s_i' \sim t_n'$ if $i \notin I$.
		If $I$ is a solution to $(P,N)$, this yields $s' \sim t'$ and hence that $\mathcal{M}' \sim \mathcal{N}'$.
		For more elaborative calculations, we refer to the proof of \cite[Thm. 1]{ABM} that can be found in the full version \cite{ABMfull} of \cite{ABM}. 
		
		\smallskip 
		
		Now we turn our attention to the second claim of the theorem, i.e., to the \textsf{NP}-completeness of deciding $\mathcal{M} \sim_\varepsilon^* \mathcal{N}$. As, by definition, this holds iff $s_{init}^\mathcal{M} \sim_\varepsilon^* s_{init}^\mathcal{N}$, we instead show \textsf{NP}-completeness of the problem to decide if, for given states $s,t$ of a finite LMC $\mathcal{M}$ and a given $\varepsilon \in (0,1]$, $s \sim_\varepsilon^* t$. 
		
		\smallskip 
		\noindent
		\textbf{\textsf{NP}-membership: } 
		The following nondeterministic, polynomial time guess-and-check algorithm computes a transitive $\varepsilon$-bisimulation $R$ on $\mathcal{M}$ with $(s,t) \in R$: 
		\begin{enumerate}
			\item Nondeterministically guess a partition $X$ of $S$ that only relates states with the same label and such that there is a block $B \in X$ with $s,t \in B$. 
			\item Check if the induced equivalence $R$ with $X = \xfrac{S}{R}$ is an $\varepsilon$-bisimulation on $\mathcal{M}$ by deciding for all $(u,v) \in R$ if $\prob(u)(A(u,v)) \leq \prob(v)(A(u,v)) + \varepsilon$, where $A(u,v)$ is the union of all $R$-equivalence classes $B \in X$ with $\prob(u)(B) > \prob(v)(B)$.
		\end{enumerate}
		
		Obviously, if $s \sim_\varepsilon^* t$, the algorithm finds a solution as the partition induced by $\sim_\varepsilon^*$ satisfies all the requirements for $X$ and $R$, respectively. On the other hand, if $X$ is a  solution of the algorithm, then the equivalence $R$ induced by $X$, i.e., such that $X = \xfrac{S}{R}$, is an $\varepsilon$-bisimulation as the condition on the probabilities posed in the second part of the algorithm ensures that, for all $(p,q) \in R$, the maximal difference in transition probabilities to subsets of $S$ is bounded from above by $\varepsilon$. As, by the first part of the algorithm $(s,t) \in X$, it follows that $s \sim_\varepsilon^* t$.
		
		\smallskip 
		\noindent
		\textbf{\textsf{NP}-hardness: } The hardness follows similar to that of the first claim, as in $\mathcal{M}$ and $\mathcal{N}$ as in \Cref{Figure: LMCs in hardness of common Epsilon-Quotient} it holds that $s \sim_\varepsilon^* t$ iff $\mathcal{M} \simeq_\varepsilon \mathcal{N}$. This is due to the fact that all states $s_1, \dots, s_n$ in $\mathcal{M}$ have the same next-state distribution, and so the only way for them to become bisimilar to either $t_y$ or $t_n$ is to be perturbed uniformly, making the induced $\varepsilon$-bisimulation transitive. 
	\end{proof}
	
	\PropPolyTime*
	\begin{proof}
		Let $R$ be an equivalence. As described in the \textsf{NP}-membership proof of (ii) in \Cref{Theorem: NP-completeness of common epsilon-quotient}, we can check in polynomial time if $R$ is a (transitive) $\varepsilon$-bisimulation by deciding, for all $(s,t) \in R$, if $\prob(s)(A(s,t)) \leq \prob(t)(A(s,t)) + \varepsilon$, where $A(s,t)$ is the union of all $R$-equivalence classes $C$ with $\prob(s)(C) > \prob(t)(C)$. This shows (i). Furthermore, we know from \Cref{Lemma: Tau-Star Condition Equivalence} that, given $R$, we can decide in polynomial time if, for each $C \in \xfrac{S}{R}$ the centroid property (see \Cref{centroid condition}) holds, i.e., if $R$ satisfies (iii) of \Cref{Theorem: Characterisation Perturbed Epsilon Bisimilar and Transitive Epsilon Bisimulation}. As it is decidable in polynomial time if $R$ is a transitive $\varepsilon$-bisimulation, deciding if $R$ is a $\varepsilon$-perturbed bisimulation is also possible in polynomial time. 
	\end{proof}
	
	\section{Proofs of Section 5}\label{Appendix: Proofs of Section 5}
	
	\LemAdditivityWeakEpsilonBisim*
	\begin{proof}
		Let $s \approx_\varepsilon^w t$ and $t \approx_\delta^w u$. Define the relation 
		\begin{align*}
			R= \{(p,q) \mid \exists \, a \in S \text{ with} p \approx_\varepsilon^w a \text{ and} a \approx_\delta^w q\}.
		\end{align*}
		Then $(s,u) \in R$, so the claim follows if $R$ is a weak $(\varepsilon + \delta)$-bisimulation. To this end, we first observe that for all $(p,q) \in R$ there is a $b \in 2^{AP}$ with $b = l(p) = l(q)$ as both $\approx_\varepsilon^w$ and $\approx_\delta^w$ only relate states with the same label. Furthermore, for any $A \subseteq S$ it holds that 
		\begin{align*}
			\mathrm{Pr}_p(b \Until A) \overset{p \approx_\varepsilon^w a}{\leq} \mathrm{Pr}_a(b \Until {\approx_\varepsilon^w}(A)) + \varepsilon \overset{a \approx_\delta^w q}{\leq} \mathrm{Pr}_q(b \Until {\approx_\delta^w}({\approx_\varepsilon^w}(A))) + \varepsilon + \delta
		\end{align*}
		and, as ${\approx_\delta^w}({\approx_\varepsilon^w}(A)) \subseteq R(A)$, this implies 
		\begin{align*}
			\mathrm{Pr}_p(b \Until A) \leq \mathrm{Pr}_q(b \Until R(A)) + \varepsilon + \delta.
		\end{align*}
		Therefore, $R$ is indeed a weak $(\varepsilon + \delta)$-bisimulation, so $s \approx_{\varepsilon + \delta}^w u$.
	\end{proof}
		
	\ProbBranchingDoesNotImplyWeak*
	\begin{proof}
		For the first claim, consider the LMC $\mathcal{M}$ depicted on the left of \Cref{Figure: Combined LMCs for proof of Proposition: Branching and Weak are incomparable}. The equivalence $R$ given via the equivalence classes $R = \{\{s, t\}, \{s_1, t_1\}, \{x\}, \{y\}\}$ is the largest branching $\varepsilon$-bisimulation on $\mathcal{M}$. In particular, $s \not \approx^b_\varepsilon s_1$ because otherwise we would have
		\begin{align*}
			\vert \mathrm{Pr}_{s}([s]_R \Until \{y\}) - \mathrm{Pr}_{s_1}([s_1]_R \Until \{y\}) \vert = \left \vert \frac{3}{8} - \frac{3}{4} \right \vert = \frac{3}{8} > \varepsilon,
		\end{align*}
		and $t \not \approx^b_\varepsilon t_1$ as, otherwise,
		\begin{align*}
			\vert \mathrm{Pr}_{t}([t]_R \Until \{x\}) - \mathrm{Pr}_{t_1}([t_1]_R \Until \{x\}) \vert = \left \vert \frac{5}{8} + \frac{5}{4} \varepsilon - \varepsilon^2 - \left (\frac{1}{4} + \varepsilon \right) \right \vert = \left \vert \frac{3}{8} + \frac{1}{4}  \varepsilon - \varepsilon^2 \right \vert > \varepsilon,
		\end{align*}
		where the inequality holds for all $\varepsilon < \frac{1}{8} \cdot (\sqrt{33} - 3) \approx 0.34307$, so in particular for all $\varepsilon < \frac{1}{4}$. 
		
		As $(s,t) \in R$, $s \approx^b_\varepsilon t$ in the LMC. However, $s \not \approx^w_\varepsilon t$ since 
		\begin{align*}
			\mathrm{Pr}_t(L(t) \Until \{x\}) = \frac{5}{8} + \frac{5}{4} \varepsilon - \varepsilon^2 > \frac{5}{8} + \varepsilon = \mathrm{Pr}_s(L(s) \Until \{x\})
		\end{align*}
		for all $0 < \varepsilon < \frac{1}{4}$. 
		Hence, $s \approx^b_\varepsilon t$ does, in general, not imply $s \approx^w_\varepsilon t$. 
		
		Regarding the second claim, consider the LMC on the right of \Cref{Figure: Combined LMCs for proof of Proposition: Branching and Weak are incomparable}. There, $\approx_\varepsilon^w$ is the symmetric and reflexive closure of $\{(s,t), (s,u), (s,v), (t,v), (t,w), (u,v), (v,w)\}$. In particular, $s \approx_\varepsilon^w t$ and $u \not \approx_\varepsilon^w w$. 
		
		However, $(s,t) \notin R$ for any branching $\varepsilon$-bisimulation $R$ on $\mathcal{M}$. To see this, first observe that, due to the state labeling, $x$ and $y$ cannot be related to any other state. It follows that $(u,w) \notin R$ for any such $R$ as, e.g., $\vert \mathrm{Pr}_u([u]_R \Until [x]_R) - \mathrm{Pr}_w([w]_R \Until [x]_R) \vert = 2 \varepsilon > \varepsilon$. 
		
		Now, assume that $(s,u) \in R$ for a branching $\varepsilon$-bisimulation $R$. Then 
		\begin{align*}
			\vert \mathrm{Pr}_u([u]_R \Until [x]_R) - \mathrm{Pr}_s([s]_R \Until [x]_R) \vert = \begin{cases}
				\frac{1}{4} + \frac{\varepsilon}{2}, &\text{if } (s, v) \notin R \\
				\frac{\varepsilon}{2}, & \text{if } (s,v) \in R
			\end{cases}
		\end{align*}
		and since in the first case we have $\frac{1}{4} + \frac{\varepsilon}{2} > \varepsilon$ because $0 < \varepsilon < \frac{1}{4}$, $(s,u) \in R$ only if $(s,v) \in R$. But then 
		\begin{align*}
			\vert \mathrm{Pr}_t([t]_R \Until [x]_R) - \mathrm{Pr}_s([s]_R \Until [x]_R) \vert = \begin{cases}
				\frac{1}{2} + \frac{\varepsilon}{2}, &\text{if } (t, v), (t, w) \notin R \\
				\frac{1}{4} + \frac{\varepsilon}{2}, &\text{if } (t,v) \in R, (t, w) \notin R \\
				\frac{1}{4} + \varepsilon, &\text{if } (t,v) \notin R, (t, w) \in R \\
				\varepsilon, & \text{if } (t,v), (t,w) \in R
			\end{cases}
		\end{align*}
		and as in the first three cases the differences are $> \varepsilon$, $(s,t) \in R$ requires $(t,v), (t,w) \in R$. However, in this case, the transitivity of $R$ implies $(u,w) \in R$, which is a contradiction. 
		
		Otherwise, if $(s,u) \notin R$, then $[s]_R \neq [u]_R$. Thus, $(s,t) \in R$ requires in particular that $\vert \mathrm{Pr}_s([s]_R \Until [u]_R) - \mathrm{Pr}_t([t]_R \Until [u]_R) \vert < \varepsilon$. If $(u,v) \notin R$ we have $\mathrm{Pr}_t([t]_R \Until [u]_R) = 0$, so the above difference becomes $\frac{1}{2} > \varepsilon$. Otherwise, if $(u,v) \in R$ then $\vert \mathrm{Pr}_s([s]_R \Until [u]_R) - \mathrm{Pr}_t([t]_R \Until [u]_R) \vert = \vert 1 - \mathrm{Pr}_t([t]_R \Until [u]_R) \vert$ and, since we know that $(u,w) \notin R$, we have  $\mathrm{Pr}_t([t]_R \Until [u]_R) \leq \frac{1}{2}$, so $\vert \mathrm{Pr}_s([s]_R \Until [u]_R) - \mathrm{Pr}_t([t]_R \Until [u]_R) \vert \geq \frac{1}{2} > \varepsilon$. Hence, $(s,t) \notin R$. 
		
		All in all, this shows that there can be no branching $\varepsilon$-bisimulation $R$ on $\mathcal{M}$ that contains $(s,t)$, and therefore it holds that $s \approx_\varepsilon^w t$ while $s \not \approx_\varepsilon^b t$. 	
	\end{proof}
	
	\begin{lemma}\label{Lemma: No Stuter Steps makes notions coincide}
		Assume that $\prob(s)(L(s)) = 0$ for all $s \in S$. Then ${\sim_\varepsilon^*} = {\approx_\varepsilon^b}$ and ${\sim_\varepsilon} = {\approx_\varepsilon^w}$.
	\end{lemma}
	\begin{proof}
		We start by showing the first claim, i.e. that, if there are no stutter steps ${\sim_\varepsilon^*} = {\approx_\varepsilon^b}$.
		
		Let $R$ be a branching $\varepsilon$-bisimulation, and let $A \subseteq S$ be $R$-closed. Then, for all $(s,t) \in R$, we have $l(s) = l(t)$ and further 
		\begin{align}
			\prob(s)(A) = \mathrm{Pr}_s([s]_R \Until A) \qquad \text{ and } \qquad \prob(t)(A) = \mathrm{Pr}_t([t]_R \Until A) \label{eq-1}
		\end{align}
		as no successor of either $s$ or $t$ is in $[s]_R = [t]_R$. But then it follows that 
		\begin{align*}
			\vert \prob(s)(A) - \prob(t)(A) \vert = \vert \mathrm{Pr}_s([s]_R \Until A) - \mathrm{Pr}_t([t]_R \Until A) \vert \overset{(s,t) \in R}{\leq} \varepsilon.
		\end{align*}
		Hence, $R$ is an equivalence and an $\varepsilon$-APB, and thus a transitive $\varepsilon$-bisimulation by \Cref{Lemma: Epsilon-Bisim = Epsilon-APB in equivalence case}.
		
		Now let $R$ be a transitive $\varepsilon$-bisimulation, and let $(s,t) \in R$. Then $R$ is also an $\varepsilon$-APB by \Cref{Lemma: Epsilon-Bisim = Epsilon-APB in equivalence case}, and it holds for all $R$-closed sets $A \subseteq S$ that
		\begin{align*}
			\vert \mathrm{Pr}_s([s]_R \Until A) - \mathrm{Pr}_t([t]_R \Until A) \vert = \vert \prob(s)(A) - \prob(t)(A) \vert \leq \varepsilon, 
		\end{align*}
		where the equality follows from \Cref{eq-1}. Therefore, $R$ is a branching $\varepsilon$-bisimulation.
		
		Regarding the second claim, the result that ${\sim_\varepsilon} = {\approx_\varepsilon^w}$ follows directly from the fact that, if $\prob(s)(L(s)) = 0$ for all $s \in S$, we have $\mathrm{Pr}_s(L(s) \Until A) = \prob(s)(A)$ for every $A \subseteq S$, so it follows that for any relation $R$
		\begin{align*}
			\prob(s)(A) &\leq \prob(t)(R(A)) + \varepsilon \qquad \text{ iff } \qquad \mathrm{Pr}_s(L(s) \Until A) \leq \mathrm{Pr}_t(L(t) \Until R(A)) + \varepsilon. 
		\end{align*}
		Note that we cannot replace ${\sim_\varepsilon}$ by $\sim_\varepsilon^*$. To see this consider, for example, a modification of the LMCs of \Cref{Figure: Example that simeq can differentiate bisimilar LMC}, where we replace the self-loops at $v, w, [v]$ and $[w]$ with the new transition probabilities $\prob(v)(w) = \prob(w)(v) = \prob([v])([w]) = \prob([w])([v]) = 1$. Then, in the direct sum $\mathcal{M}$ of these LMCs, we have $\prob(q)(L(q)) = 0$ for all states $q$. In particular, the largest weak $\varepsilon$-bisimulation $\approx_\varepsilon^w$ on $\mathcal{M}$ contains the pair $(t, [s])$, but we have already argued in \Cref{Section: Epsilon-Quotients} that there can be no transitive $\varepsilon$-bisimulation that relates these two states. 
	\end{proof}
	
	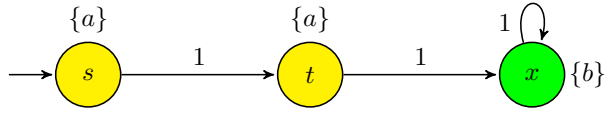
\begin{figure}[t]
		\centering
		\resizebox{!}{0.07\textheight}{\begin{tikzpicture}[->,>=stealth',shorten >=1pt,auto, semithick]
				\tikzstyle{every state} = [text = black]
				
				\node[state] (s0) [fill = yellow] {$s$}; 
				\node[state] (s1) [right of = s0, node distance = 3cm, fill = yellow] {$t$};
				\node[state] (x) [right of = s1, fill = green, node distance = 3cm] {$x$};
				\node(sinit) [left of = s0, node distance = 1.2cm] {};
				
				\path 
				(sinit) edge (s)
				(s0) edge node [above] {$1$} (s1)
				(s1) edge node [above] {$1$} (x)
				(x) edge [loop above] node [left, pos = 0.15] {$1$} (x)
				;
				
				\node [above of = s0, node distance = 0.75cm] {$\{a\}$}; 
				\node [above of = s1, node distance = 0.75cm] {$\{a\}$};
				\node [right of = x, node distance = 0.75cm] {$\{b\}$};
		\end{tikzpicture}}
		\caption{LMC in which $s \approx^b t$ and $s \approx^w t$, but neither $s \sim_\varepsilon t$ nor $s \equiv_\varepsilon t$ for any $\varepsilon \in [0,1)$.}
		\label{Figure: Branching Bisimilarity does not imply Epsilon-Bisimilarity or Epsilon-APB}
	\end{figure}
	
	\WeakAndBanchingIncompWithEpsAPBAndBisim*
	\begin{proof}
		We first deal with (ii). Consider the LMC depicted in \Cref{Figure: Branching Bisimilarity does not imply Epsilon-Bisimilarity or Epsilon-APB}, and let $\varepsilon \in [0,1)$. The equivalence $R$ induced by the classes $C_1 = \{s, t\}$ and $C_2 = \{x\}$ is both a branching bisimulation and a weak bisimulation, as $x$ has a unique label and every state in $C_1$ reaches $C_2$ via a $C_1$-path, resp. an $\{a\}$-labeled path, with probability $1$. Hence, $s \approx^b t$ and $s \approx^w t$, so in particular also $s \approx^b_\varepsilon t$ and $s \approx^w_\varepsilon t$ for any $\varepsilon$. 
		However, $\vert \prob(s)(\{x\}) - \prob(t)(\{x\}) \vert = 1 > \varepsilon$ and $\prob(s)(\{t\}) = 1 > \varepsilon = \prob(t)(R(\{t\}) + \varepsilon$ for any $R$ that only relates states with the same label, so neither $s \equiv_\varepsilon t$ nor $s \sim_\varepsilon t$.
		
		To show (i) we do a case distinction on $\approx_\varepsilon$. 
		If ${\approx_\varepsilon} = {\approx^b_\varepsilon}$, consider the LMC depicted on the left-hand side of \Cref{Figure: LMCs for showing that eps-bisim does not imply branching or weak} where $\varepsilon_1, \varepsilon_2 \in (0,1)$, $\varepsilon_1 \neq \varepsilon_2$, $\varepsilon_1 + \varepsilon_2 < 1$, and $\varepsilon = \vert \varepsilon_1 - \varepsilon_2 \vert$. 
		In this LMC, both $s \sim_\varepsilon t$ and $s \equiv_\varepsilon t$. However, $\mathrm{Pr}_s([s]_R \Until \{x_1\}) = \frac{\varepsilon_1}{\varepsilon_1 + \varepsilon_2} = \mathrm{Pr}_t([t]_R\Until \{x_2\})$ and $\mathrm{Pr}_s([s]_R \Until \{x_2\}) = \frac{\varepsilon_2}{\varepsilon_1 + \varepsilon_2} = \mathrm{Pr}_t([t]_R \Until \{x_1\})$ for any equivalence $R$ that only relates states with the same label. Hence,
		\begin{align*}
			\vert \mathrm{Pr}_s([s]_R \Until \{x_1\}) - \mathrm{Pr}_t([t]_R \Until \{x_1\}) \vert = \frac{\vert \varepsilon_1 - \varepsilon_2\vert }{\varepsilon_1 + \varepsilon_2} \overset{\varepsilon_1 + \varepsilon_2 < 1}{>} \vert \varepsilon_1 - \varepsilon_2 \vert = \varepsilon.
		\end{align*}
		Thus, there can be no branching bisimulation $R$ with $(s,t) \in R$, i.e., $s \not \approx_\varepsilon^b t$.
		
		On the other hand, if ${\approx_\varepsilon} = {\approx^w_\varepsilon}$, consider the LMC on the right-hand side of \Cref{Figure: LMCs for showing that eps-bisim does not imply branching or weak} with $\varepsilon \in (0,1)$. In this LMC, $\sim_\varepsilon$ is given as the symmetric and reflexive closure of $\{(s, t), (s_1, t_1)\}$. As $\sim_\varepsilon$ is an equivalence, it is also an $\varepsilon$-APB by \Cref{Lemma: Epsilon-Bisim = Epsilon-APB in equivalence case}, which yields $s \sim_\varepsilon t$ and $s \equiv_\varepsilon t$. Furthermore, $\mathrm{Pr}_t(L(t) \Until \{x\}) = 1$ while $\mathrm{Pr}_s(L(t) \Until \{x\}) = \frac{4(1-\varepsilon)^2}{(2-\varepsilon)^2}$. But then the inequality 
		\begin{align*}
			\mathrm{Pr}_t(L(t) \Until \{x\}) = 1 > \mathrm{Pr}_s(L(s) \Until \{x\}) + \varepsilon = \frac{4(1-\varepsilon)^2}{(2-\varepsilon)^2} + \varepsilon
		\end{align*}
		holds for all $\varepsilon \in (0,1)$, so $s \not \approx_\varepsilon^w t$.
		
		As $\sim_\varepsilon$ and $\equiv_\varepsilon$ above are transitive, the claim follows analogously for $\sim_\varepsilon^*$ and $\equiv_\varepsilon^*$.
	\end{proof}

	\EndlessStutteringBranchingBisim*
	\begin{proof}
		We start with the first claim. Let $R$ be a branching $\varepsilon$-bisimulation and $(s,t) \in R$. We write $C$ for $[s]_R = [t]_R$ and get, since $S \setminus C$ is $R$-closed,
		\begin{align*}
			\vert \mathrm{Pr}_s(\Box C) - \mathrm{Pr}_t(\Box C) \vert &= \vert 1 - \mathrm{Pr}_s(C \Until (S \setminus C)) - (1 - \mathrm{Pr}_t(C \Until (S \setminus C))) \vert \\
			&= \vert \mathrm{Pr}_s(C \Until (S \setminus C)) - \mathrm{Pr}_t(C \Until (S \setminus C)) \vert \overset{(s,t) \in R}{\leq} \varepsilon. 
		\end{align*}
		
		Regarding the second claim assume that, for $C \in \xfrac{S}{R}$, there is an $s \in C$ with $\mathrm{Pr}_s(\Box C) > 0$. Then $\mathcal{M}$ must contain a non-empty bottom strongly connected component $B \subseteq S$ such that $B \subseteq C$. In particular, for any $t \in B$ it holds that $(s,t) \in R$ and $\mathrm{Pr}_t(\Box C) = 1$. Hence, $\mathrm{Pr}_t(C \Until (S \setminus C)) = 0$ and it follows that, since again $S \setminus C$ is $R$-closed, 
		\begin{align*}
			\vert \mathrm{Pr}_s(C \Until (S \setminus C)) - \mathrm{Pr}_t(C \Until (S \setminus C)) \vert = \vert \mathrm{Pr}_s(C \Until (S \setminus C)) \vert \overset{(s,t) \in R}{\leq} \varepsilon. 
		\end{align*}
		Thus, the probability to reach any state outside of $C$ from $s$ is at most $\varepsilon$, and we get 
		\begin{align*}
			\mathrm{Pr}_s(\Box C) = 1 - \mathrm{Pr}_s(C \Until (S\setminus C)) \geq 1 - \varepsilon.
		\end{align*}
		Further, $B \neq \emptyset$ implies that there can be no $s \in C$ with $\mathrm{Pr}_s(\Box C) = 0$, as otherwise $1 = \mathrm{Pr}_s(C \Until (S \setminus C)) \leq \varepsilon$, which contradicts $\varepsilon \in [0,1)$. So, $\mathrm{Pr}_s(\Box C) \geq 1 - \varepsilon$ for all $s \in C$.
		
		To show the third claim, let $R$ be a weak $\varepsilon$-bisimulation on $\mathcal{M}$ with $(s,t) \in R$, set $b = l(s) = l(t)$ and let $A  = S \setminus L(b)$. Then $A$ is $R$-closed, and $R(A) = A$. Therefore,
		\begin{align*}
			\vert \mathrm{Pr}_s(\Box b) - \mathrm{Pr}_t (\Box b) \vert = \vert 1 - \mathrm{Pr}_s(b \Until A) - (1 - \mathrm{Pr}_t(b \Until A)) \vert = \vert \mathrm{Pr}_s(b \Until A) - \mathrm{Pr}_t (b \Until A) \vert \leq \varepsilon,
		\end{align*}
		where the inequality follows form the fact that $(s,t) \in R$, and hence 
		\begin{align*}
			\mathrm{Pr}_s(b \Until A) \leq \mathrm{Pr}_t(b \Until R(A)) + \varepsilon = \mathrm{Pr}_t(b \Until A) + \varepsilon
		\end{align*}
		as well as $\mathrm{Pr}_t(b \Until A) \leq \mathrm{Pr}_s(b \Until A) + \varepsilon$ by a similar argument.
	\end{proof}
	
	\EpsilonBisimFromBranchingBisimOnTranfsformation*
	\begin{proof}
		We start with the direction from left to right. Let $R$ be a branching $\varepsilon$-bisimulation on $\mathcal{M}$. By construction, $\xfrac{S_R}{R^b} = \{C \cup \{s_C\} \mid C \in div_R\} \cup \{C \mid C \in \xfrac{S}{R} \setminus div_R\}$, so the only difference between $\xfrac{S}{R}$ and $\xfrac{S_R}{R^b}$ is the addition of the ``divergence states'' $s_C$ to all equivalence classes in $div_R$.
		
		As $R^b$ is an equivalence that only relates states with the same label, the only thing left to show is that $R^b$ satisfies the condition on the probabilities. Following \Cref{Lemma: Epsilon-Bisim = Epsilon-APB in equivalence case}, however, it is sufficient to show that $R^*$ is an $\varepsilon$-APB on $S_R$, i.e., that for all $(s,t) \in R^b$ and all $R^b$-closed sets $A' \subseteq S_R$ it holds that 
		\begin{align*}
			\vert \prob_R(s)(A') - \prob_R(t)(A') \vert \leq \varepsilon.
		\end{align*}
		
		Let $C' \in \xfrac{S_R}{R^b}$ and define $C = C' \cap S$. Then $C = C'$ if $C \notin div_R$ and $C = C' \setminus \{s_C\}$ otherwise. Furthermore, let $(s,t) \in R^b$ with $s,t \in C'$, let $A'$ as above and define $A = A' \cap S$. Then, as $A'$ is $R^b$-closed, it is clear that $A$ must be $R$-closed.
		
		If $s_C \notin C'$, i.e., if $C \notin div_R$, then $\mathrm{Pr}_u(\Box C) = 0$ for all $u \in C$. Thus, neither $s$ nor $t$ can have a transition in $\mathcal{M}_R$ to any state in $A' \setminus A$. But then 
		\begin{align}
			\vert \prob_R(s)(A') - \prob_R(t)(A') \vert = \vert \prob_R(s)(A) - \prob_R(t)(A) \vert = \vert \mathrm{Pr}_s(C \Until A) - \mathrm{Pr}_t(C \Until A) \vert \leq \varepsilon \label{Proof: EpsilonBisimFromBranchingBisimOnTranfsformation - Eq1}
		\end{align}
		where the inequality follows from $(s,t) \in R$ and $A$ being $R$-closed. 
		
		Now assume that $C \in div_R$, i.e., that $s_C \in C'$. If both $s = t = s_C$ then the claim is clear. Thus, we can w.l.o.g. assume that $s \in C$ and $t \in C'$ (the case $s \in C'$ and $t \in C$ works analogously). We do a case distinction on the inclusion of $C'$ in $A'$. 
		
		\noindent
		\textbf{Case 1: $C'$ is not contained in $A'$.} Note that, since $A'$ is $R^b$-closed, the assumption implies that  $u \notin A'$ for all $u \in C'$, and in particular that $s_C \notin A'$.
		
		Now, if both $s,t \in S$, then again $\prob_R(s)(A') = \prob_R(s)(A)$ and $\prob_R(t)(A') = \prob_R(t)(A)$, so the claim follows as in \Cref{Proof: EpsilonBisimFromBranchingBisimOnTranfsformation - Eq1}. Otherwise, $s \in S$ and $t = s_C$. But then $\mathrm{Pr}_s(\Box C) \geq 1-\varepsilon$ by \Cref{Lemma: Properties Endless Stuttering} and $\prob_R(t)(s_C) = 1$. Hence, as $s_C \notin A'$, $\prob_R(t)(A') = 0$ and 
		\begin{align*}
			\prob_R(s)(A') = \prob_R(s)(A) = \mathrm{Pr}_s(C \Until A) \leq \mathrm{Pr}_s(C \Until (S \setminus C)) = 1 - \mathrm{Pr}_s(\Box C) \leq \varepsilon,
		\end{align*}
		which implies $\vert \prob_R(s)(A') - \prob_R(t)(A') \vert = \vert \prob_R(s)(A') \vert \leq \varepsilon$.
		
		\noindent
		\textbf{Case 2: $C' \subseteq A'$.} 
		By \Cref{Lemma: Properties Endless Stuttering} we know that, as $C \in div_R$, both $\mathrm{Pr}_s(\Box C) \geq 1 - \varepsilon$ and $\mathrm{Pr}_t(\Box C) \geq 1- \varepsilon$, so it follows that $\prob_R(s)(s_C) \geq 1- \varepsilon$ and $\prob_R(t)(s_C) \geq 1 - \varepsilon$. Hence, it holds for $u \in \{s,t\}$ that
		\begin{align*}
			1-\varepsilon \leq \prob_R(u)(s_C) \overset{s_C \in A'}{\leq} \prob_R(u)(A') \leq 1
		\end{align*}
		and thus $\prob_R(s)(A') - \prob_R(t)(A') \leq \varepsilon \text{ and } \prob_R(s)(A') - \prob_R(t)(A') \geq - \varepsilon$ or, equivalently, $\vert \prob_R(s)(A') - \prob_R(t)(A') \vert \leq \varepsilon$.
		Therefore, $R^b$ is indeed an $\varepsilon$-APB.
		
		\smallskip
		For the other direction, assume that $R^b$ is a transitive $\varepsilon$-bisimulation on $\mathcal{M}_R$, and let $(s,t) \in R$. Furthermore, let $C = [s]_R = [t]_R$ and $B \subseteq S$ be $R$-closed with $B \neq C$. In particular, $B$ is a disjoint union of $R$ equivalence classes $C_1, \dots, C_n$, i.e., $B = C_1 \uplus \dots \uplus C_n$. We can w.l.o.g. assume that $C \neq C_i$ for all $i$, as otherwise $\vert \mathrm{Pr}_s(C \Until B) - \mathrm{Pr}_t(C \Until B) \vert = 0$. But then it follows that $\mathrm{Pr}_s(C \Until B) = \prob_R(s)(B)$ and $\mathrm{Pr}_t(C \Until B) = \prob_R(s)(B)$ as $B \subseteq S \setminus C$. Moreover, since $\prob_R(s)(s_D) = 0$ for all $D \in \xfrac{S}{R}$ with $D \neq C$, we have $\mathrm{Pr}_s(C \Until B) = \prob_R(s)(B_R)$ and $\mathrm{Pr}_t(C \Until B) = \prob_R(s)(B_R)$, where $B_R = B \cup \bigcup_{i = 1, \dots, n} s_{C_i}$ is $R^b$-closed. Since $R^b$ is a transitive $\varepsilon$-bisimulation on $\mathcal{M}_R$, which by \Cref{Lemma: Epsilon-Bisim = Epsilon-APB in equivalence case} is the case iff it is an $\varepsilon$-APB, it follows that 
		\begin{align*}
			\vert \mathrm{Pr}_s(C \Until B) - \mathrm{Pr}_t(C \Until B) \vert = \vert \prob_R(s)(B_R) - \prob_R(t)(B_R) \vert \leq \varepsilon
		\end{align*}
		which, as $R$ only relates states with the same label, proves the claim.
	\end{proof}
	
	\BoundWeakBisim*
	\begin{proof}
		Let $\mathcal{M}$ be a finite LMC. Let $\mathcal{L} = \{b \in 2^{AP} \mid \exists \, s \in S \colon \mathrm{Pr}_s(\Box b) > 0\}$ be the set of ``divergent labels'' in $\mathcal{M}$. 
		As a first step, we construct from $\mathcal{M}$ a transformed LMC $\mathcal{M}^w = (S^w, \prob^w, s_{init}, l^w)$ with 
		\begin{itemize}
			\item $S^w = S \cup \{s_b \mid b \in \mathcal{L}\}$ where the $s_b$ are fresh, pairwise different states
			\item $l^w(s) = l(s)$ if $s \in S$ and $l(s_b) = b$ for all $b \in \mathcal{L}$
			\item For $s, t \in S^w$ the distribution $\prob^w(s)$ is defined by
			\begin{align*}
				\prob^w(s)(t) = \begin{cases}
					\mathrm{Pr}_s(L(s) \Until t), & \text{if } s,t \in S \text{ and} l(s) \neq l(t) \\
					\mathrm{Pr}_s(\Box b), & \text{if } s\neq s_b, t = s_b, b = l(s) \text{ and} b \in \mathcal{L} \\
					1, &\text{if} s = t = s_b \text{ for some } b \in \mathcal{L}\\
					0, &\text{otherwise}
				\end{cases}.
			\end{align*}
		\end{itemize}
	
		Then $\mathcal{M}^w$ is again a finite LMC, and it is clear that the transformation from $\mathcal{M}$ to $\mathcal{M}^w$ preserves the probabilities of any given state $s \in S$ to reach a differently labeled state $t \in S$. 
		
		Now let, for a given $\varepsilon$-bisimulation $R$ on $\mathcal{M}$, $R^w$ be the finest reflexive and symmetric relation on $S^w$ with $R \subseteq R^w$ and $(s, s_b) \in R^w$ iff $s \in S$, $b = l(s)$ and $\mathrm{Pr}_s(\Box b) \geq 1 - \varepsilon$. As $R \subseteq R^w$, it follows for all $(s,t) \in R$ that $(s,t) \in R^w$. 
		
		The main part of the proof is now to show that $R^w$ is actually an $\varepsilon$-bisimulation on $\mathcal{M}^w$. To this end, let $(s,t) \in R^w$, $b = l(s) = l(t)$ and $A \subseteq S^w$. We must show that 
		\begin{align*}
			\prob^w(s)(A) \leq \prob^w(t)(R^w(A)) + \varepsilon
		\end{align*}
		which we do by a case distinction on $s$ and $t$.
		
		\noindent
		\textbf{Case 1: $s = s_b$.}
		As $(s,t) = (s_b, t)\in R^w$ it holds by definition of $R^w$ that $\mathrm{Pr}_t(\Box b) \geq 1 - \varepsilon$. If now $s_b \in A$ then 
		\begin{align*}
			\prob^w(s)(A) = \prob^w(s_b)(A) = 1 = 1 - \varepsilon + \varepsilon \leq \prob^w(t)(s_B) + \varepsilon \overset{s_b \in A}{\leq} \prob^w(t)(R^w(A)) + \varepsilon, 
		\end{align*}
		while if $s_B \notin A$ we have 
		\begin{align*}
			\prob^w(s)(A) = \prob^w(s_b)(A) = 0 \leq  \varepsilon \leq \prob^w(t)(R^w(A)) +  \varepsilon.
		\end{align*}
		
		\noindent
		\textbf{Case 2: $t = s_b$.}
		Again, as $(s,t) = (s, s_b) \in R^w$ we have $\mathrm{Pr}_s(\Box b) \geq 1 - \varepsilon$. If $s_b \in R^w(A)$ then $\prob^w(t)(R^w(A)) \geq \prob^w(t)(s_b)  = 1$, so in particular 
		\begin{align*}
			\prob^w(t)(R^w(A)) +  \varepsilon \geq 1 +  \varepsilon \geq 1 \geq \prob^w(s)(A). 
		\end{align*}
		Otherwise, if $s_b \notin R^w(A)$, then also $s_b \notin A$. Thus, $\prob^w(t)(R^w(A)) = \prob^w(t)(A) = 0$ and hence 
		\begin{align*}
			\prob^w(s)(A) \leq \prob^w(s)(S \setminus \{s_b\}) = 1 - \underbrace{\prob^w(s)(\{s_b\})}_{\geq 1 - \varepsilon \text{ as} (s, s_b) \in R^w} \leq \varepsilon \leq \prob^w(t)(R^w(A)) + \varepsilon.
		\end{align*}
		
		\noindent
		\textbf{Case 3: $s,t \in S$ and $s_b \notin A$.}
		Be definition, neither $s$ nor $t$ can transition to any state $s_{b'} \in S^w \setminus S$ with $s_b \neq s_b'$. Thus, as $s_b \notin A$ by assumption, we get 
		\begin{align}
			\prob^w(s)(A) = \prob^w(s)(A \cap S) = \sum_{\substack{q \in A \cap S \\ l(q) \neq b}} \mathrm{Pr}_s(b \Until \{q\})  = \mathrm{Pr}_s(b \Until ((A \cap S) \setminus L(s))). \label{M*-eq4}
		\end{align}
		Next, we show that \begin{align}
			R((A \cap S) \setminus L(t)) = R(A \cap S) \setminus L(t). \label{M*-eq5}
		\end{align}
		To see this, let $q \in R((A \cap S) \setminus L(t))$. Then there is a $p \in (A \cap S) \setminus L(t)$ with $(p, q) \in R$. As $R$ only relates states with the same label, $q \notin L(t)$. Thus, as $p \in (A \cap S)$, we get $q \in R(A \cap S)$, and as $q \notin L(t)$ it follows that $q \in R(A \cap S) \setminus L(t)$. Hence, $R((A \cap S) \setminus L(t)) \subseteq R(A \cap S) \setminus L(t)$. Regarding the other direction, let $q \in R(A \cap S) \setminus L(t)$. Then $l(q) \neq l(t)$ and there is a $p \in A \cap S$ with $(p, q) \in R$. Again, as $R$ only relates states with the same label, we have $l(p) = l(q)$ and, in particular, $l(p) \neq l(t)$, so $p \in (A \cap S) \setminus L(t)$. Therefore $q \in R((A \cap S) \setminus L(t))$, yielding $R(A \cap S) \setminus L(t) \subseteq R((A \cap S) \setminus L(t))$. 
		
		Now, as $R \subseteq R^w$ by definition and $A \subseteq A \cap S$, it follows that 
		\begin{align}
			\prob^w(t)(R^w(A)) &\geq \prob^w(t)(R^w(A \cap S)) \geq \prob^w(t)(R(A \cap S))  = \sum_{\substack{q \in R(A \cap S) \\ L(q) \neq b}} \mathrm{Pr}_t(b \Until \{q\}) \nonumber \\&= \mathrm{Pr}_t(b \Until (R(A \cap S) \setminus L(t))) \overset{(\ref{M*-eq5})}{=} \mathrm{Pr}_t(b \Until R((A \cap S) \setminus L(t))). \label{M*-eq6}
		\end{align}
		By combining \Cref{M*-eq4,M*-eq6} we finally get
		\begin{align*}
			\prob^w(s)(A)& \overset{(\ref{M*-eq4})}{=} \mathrm{Pr}_s(b \Until ((A \cap S) \setminus L(s))) = \mathrm{Pr}_s(b \Until ((A \cap S) \setminus L(t))) \\
			&\overset{(s,t) \in R}{\leq} \mathrm{Pr}_t(b \Until R((A \cap S) \setminus L(t))) + \varepsilon \overset{(\ref{M*-eq6})}{\leq} \prob^w(t)(R^w(A)) + \varepsilon.
		\end{align*}
		
		\noindent
		\textbf{Case 4: $s,t \in S$ and $s_b \in A$.}
		First, we show that
		\begin{align}
			R^w(S^w \setminus R^w(A)) \subseteq S^w \setminus A. \label{M*-eq2}
		\end{align}
		To see this, let $q \in R^w(S^w \setminus R^w(A))$ which means that there is $p \in S^w \setminus R^w(A)$ with $(p,q) \in R^w$. But then we have $p \notin R^w(A)$ and hence it must hold that $q \notin A$. Thus, $q \in S^w \setminus A$. 
		
		Next, we argue that
		\begin{align}
			\prob^w(t)(S^w \setminus R^w(A)) \le \prob^w(s)(R^w(S^w \setminus R^w(A))) + \varepsilon. \label{M*-eq3}
		\end{align}
		This follows as in Case 3 by swapping the roles of states $s$ and $t$ and by using $A' = S^w \setminus R^w(A)$. In particular we have $s_b \notin A'$ since $s_b \in A \subseteq R^w(A)$.
		
		Together, this yields
		\begin{align*}
			\prob^w(s)(A) 
			~&=~  1- \prob^w(s)(S^w \setminus A)\\
			~&\overset{(\ref{M*-eq2})}{\leq}~ 1 - \prob^w(s)(R^w(S^w \setminus R^w(A)))\\
			~&\overset{(\ref{M*-eq3})}{\leq}~ 1 - \Big(\prob^w(t)(S^w \setminus R^w(A)) - \varepsilon\Big)\\
			~&=~ \Big(1 - \prob^w(t)(S^w \setminus R^w(A))\Big) + \varepsilon\\
			~&=~ \prob^w(t)( R^w(A)) + \varepsilon.
		\end{align*}
	
		All in all, this shows that $R^w$ is actually an $\varepsilon$-bisimulation on $\mathcal{M}^w$. The claim now follows from \Cref{thm:unbounded_reach} and the fact that the transformation from $\mathcal{M}$ to $\mathcal{M}^w$ preserves the probabilities for any state to reach a differently labeled state.
	\end{proof}
	
	\ThmNPCompleteBranching*
	\begin{proof}
		\smallskip 
		\noindent
		\textbf{\textsf{NP}-membership: } 
		The following nondeterministic, polynomial time guess-and-check algorithm computes a branching $\varepsilon$-bisimulation $R$ on $\mathcal{M}$ with $(s,t) \in R$: 
		\begin{enumerate}
			\item Nondeterministically guess a partition $X$ of $S$ that only relates states with the same label and such that there is a block $B \in X$ with $s,t \in B$. 
			\item Check if, for the induced equivalence $R$ with $X = \xfrac{S}{R}$, the equivalence $R^b$ from \Cref{Lemma: Derivation of Epsilon-Bisimulation in MR} is a transitive $\varepsilon$-bisimulation in the LMC $\mathcal{M}_R$ defined as in \Cref{Definition: MR}, by deciding for all $(u,v) \in R^b$ if $\prob_R(u)(A(u,v)) \leq \prob_R(v)(A(u,v)) + \varepsilon$, where $A$ is the union of all $R^b$-equivalence classes $B \in X$ with $\prob_R(u)(B) > \prob_R(v)(B)$. 
		\end{enumerate}
		
		If $s \approx_\varepsilon^b t$, the algorithm finds a solution as the partition induced by $\approx_\varepsilon^b$ satisfies all the requirements for $X$ and $R$, respectively. On the other hand, any solution $X$ of the algorithm induces a transitive $\varepsilon$-bisimulation $R^b$ on $\mathcal{M}_R$ that contains $(s,t)$, so $s \approx_\varepsilon^b t$ follows by \Cref{Lemma: Derivation of Epsilon-Bisimulation in MR}. 
		
		\smallskip 
		\noindent
		\textbf{\textsf{NP}-hardness: } The hardness follows by polynomial reduction from the decision problem if, for given $s,t \in S$ and $\varepsilon \in (0,1]$, $s \sim_\varepsilon^* t$, which is \textsf{NP}-complete by (the proof of) \Cref{Theorem: NP-completeness of common epsilon-quotient}. Given LMC $\mathcal{M}$, we construct in polynomial time a new LMC $\mathcal{N} = (S', \prob', s_{init}, l')$ as follows: 
		\begin{itemize}
			\item $S' = S \uplus \{s' \mid s \in S\}$, i.e., we add for each state $s \in S$ a fresh duplicate $s'$ to $S'$
			\item $\prob'(s)(s') = 1$ and $\prob'(s')(t) = \prob(s)(t)$ for all $s,t$
			\item $l'(s) = l(s)$ if $s \in S$ and $l'(s') = (l(s))' = l(s)'$ if $s \in S' \setminus S$, i.e., we duplicate the labels in $\mathcal{M}$ and label each state in $S' \setminus S$ by a primed version of the label of its corresponding state in $S$.
		\end{itemize}
		By construction, there are no stutter steps in $\mathcal{N}$, i.e., $\prob'(q)(L(q)) = 0$ for all $q \in S'$. Using \Cref{Lemma: No Stuter Steps makes notions coincide} it follows that ${\sim_\varepsilon^*} = {\approx_\varepsilon^b}$ in $\mathcal{N}$. Furthermore, there is a $1$-to-$1$-correspondence between $\varepsilon$-bisimulations in $\mathcal{M}$ and $\mathcal{N}$: from a given $\varepsilon$-bisimulation on $\mathcal{M}$ we can construct an $\varepsilon$-bisimulation on $\mathcal{N}$ by adding the primed version of all pairs of related states to the relation, and we obtain an $\varepsilon$-bisimulation on $\mathcal{M}$ from one on $\mathcal{N}$ by removing all of these pairs. All in all, this shows that $s \sim_\varepsilon^* t$ in $\mathcal{M}$ iff $s \sim_\varepsilon^* t$ in $\mathcal{N}$ iff $s \approx_\varepsilon^b t$ in $\mathcal{N}$. 
	\end{proof}
\end{document}